\newif\iflonglong\longlongfalse
\newif\iflong\longtrue

\documentclass[a4paper,UKenglish,cleveref, autoref,nameinlink,numberwithinsect]{lipics-v2021}
\hideLIPIcs
\nolinenumbers

\usepackage{amsfonts}
\usepackage{amssymb}
\usepackage{dsfont}
\usepackage{algpseudocode}
\usepackage[linesnumbered,lined,commentsnumbered,ruled]{algorithm2e}
\usepackage{tikz}

\usepackage[disable]{todonotes}

\usepackage{tabularx, multirow}
\usepackage{etoolbox}
\usepackage{mathtools}
\usepackage{blindtext}
\usepackage{comment}
\usepackage{multicol}
\usepackage{hyperref}
\usepackage{tabularx}
\usepackage{xspace}
\usetikzlibrary{calc}
	
\hypersetup{colorlinks = true, citecolor=blue}
\DeclareMathOperator{\poly}{poly}

\DeclareMathOperator{\cen}{cen}
\DeclareMathOperator{\can}{cand}
\newcommand{\Oh}{\ensuremath{\mathcal{O}}}

\newtheorem{fact}{Fact}{\upshape\itshape}{\upshape\rmfamily}
\newtheorem{property}{Property}{\bfseries}{\itshape}
\newtheorem{operation}{Operation}{\bfseries}{\itshape}
\newtheorem{clm}{Claim}{\upshape\itshape}{\upshape\rmfamily}
\crefname{red}{Reduction Rule}{Reduction Rules}
\crefname{brn}{Branching Rule}{Branching Rules}
\crefname{clm}{Claim}{Claims}
\crefname{property}{Property}{Properties}
\crefname{operation}{Operation}{Operations}
\crefname{observation}{Observation}{Observations}

\newcommand{\UCon}{\textsc{Unweighted Connectivity NWS}\xspace}
\newcommand{\Con}{\textsc{Connectivity NWS}\xspace}
\newcommand{\UStars}{\textsc{Unweighted Stars NWS}\xspace}
\newcommand{\Stars}{\textsc{Stars NWS}\xspace}

\newcommand{\W}[1]{\ensuremath{\mathrm{W}[#1]}\xspace}
\newcommand\NP{\ensuremath{\mathrm{NP}}\xspace}

\newcommand{\mc}{\mathcal{C}}

\newcommand{\problemdef}[3]{
    \begin{quote}
      \normalsize\textsc{#1} \smallskip \\
      \textbf{Input:} #2 \\
       \textbf{Question:} #3
    \end{quote}
}

\newcommand{\todomi}[1]{\todo[inline,color=green!50]{#1}}

\title{On the Complexity of Community-aware Network~Sparsification} 

\titlerunning{On the Complexity of Community-aware Network Sparsification} 

\author{Emanuel Herrendorf}{Philipps-Universität Marburg, Fachbereich Mathematik und Informatik,  Marburg, Germany}{}{}{}

\author{Christian Komusiewicz}{Friedrich-Schiller-Universität Jena, Fakultät für Mathematik und Informatik,  Jena, Germany}{c.komusiewicz@uni-jena.de}{https://orcid.org/0000-0003-0829-7032}{}

\author{Nils Morawietz}{Friedrich-Schiller-Universität Jena, Fakultät für Mathematik und Informatik,  Jena, Germany}{nils.morawietz@uni-jena.de}{https://orcid.org/0000-0002-7283-4982}{}

\author{Frank Sommer}{Friedrich-Schiller-Universität Jena, Fakultät für Mathematik und Informatik,  Jena, Germany}{frank.sommer@uni-jena.de}{https://orcid.org/0000-0003-4034-525X}{}

\authorrunning{E.~Herrendorf, C.~Komusiewicz, N.~Morawietz, and F.~Sommer}

\Copyright{Emanuel Herrendorf, Christian Komusiewicz, Nils Morawietz, and Frank Sommer} %TODO mandatory, please use full first names. LIPIcs license is "CC-BY";  http://creativecommons.org/licenses/by/3.0/

 %TODO mandatory; please add comma-separated list of keywords

\acknowledgements{Some of the results of this work are also contained in the first author’s Master thesis~\cite{Her21}.}

\begin{document}

%\title{On Community-aware Network Sparsification\thanks{Some results of this work are also contained in the first author's Master thesis~\cite{EH22}.}}
%\titlerunning{On Community-aware Network Sparsification}
%\authorrunning{E.~Herrendorf, C.~Komusiewicz, N.~Morawietz, and F.~Sommer}

%\author{Emanuel Herrendorf \and Christian Komusiewicz\orcidID{0000-0003-0829-7032} \and Nils Morawietz\thanks{Supported by the DFG, project OPERAH (KO 3669/5-1).}\orcidID{0000-0002-7283-4982}\and Frank Sommer \thanks{Supported by the DFG, project EAGR (KO 3669/6-1).}
%  \orcidID{0000-0003-4034-525X}}

%\institute{Fachbereich Mathematik und Informatik, Philipps-Universität Marburg,\\ Marburg, Germany \\ \email{\{komusiewicz,morawietz,fsommer\}@informatik.uni-marburg.de}}

\maketitle

\begin{abstract}
  Network sparsification is the task of reducing the number of edges of a given graph while preserving some crucial graph property. In community-aware network sparsification, the preserved property concerns the subgraphs that are induced by the communities of the graph which are given as vertex subsets. 
    This is formalized in the $\Pi$-\textsc{Network Sparsification} problem: given an edge-weighted graph~$G$, a collection~$\mathcal{C}$ of $c$~subsets of~$V(G)$, called communities, and two numbers~$\ell$ and~$b$,  
    the question is whether there exists a spanning subgraph~$G'$ of~$G$ with at most $\ell$~edges of total weight at most~$b$ such that~$G'[C]$ fulfills~$\Pi$ for each community~$C\in \mathcal{C}$. 
    In this work, we consider two graph properties~$\Pi$: the connectivity property and the property of having a spanning star. The corresponding problems are called \Con and \Stars. Since both problems are NP-hard, we study their parameterized and fine-grained complexity.
 
First, we provide a tight $2^{\Omega(n^2+c)} \cdot \poly(n+|\mathcal{C}|)$-time running time lower bound based on the ETH for both problems, where~$n$ is the number of vertices in~$G$. The lower bound holds even in the restricted case when all communities have size at most~4,~$G$ is a clique, and every edge has unit weight. For the connectivity property, the unit weight case with~$G$ being a clique is the well-studied problem of computing a hypergraph support with a minimum number of edges.    
We then study the complexity of both problems parameterized by the feedback edge number~$t$ of the solution graph~$G'$.  
For \Stars, we present an XP-algorithm for~$t$. This answers an open question by Korach and Stern [Discret. Appl. Math.~'08] who asked for the existence of polynomial-time algorithms for~$t=0$. Our result implies polynomial-time algorithms for all constant values of~$t$.
In contrast, we show for \Con that known polynomial-time algorithms for~$t=0$ [Korach and Stern, Math.~Program.~'03; Klemz et al.,~SWAT~'14] cannot be extended to larger values of~$t$ by showing that \Con is NP-hard for~$t=1$.

\end{abstract}
%\newpage

\section{Introduction}
%    \todo[inline]{mail von emanuel angeben?}
%\todo[inline]{weitere coole Anwendung: placing green bridges~\cite{FK21,HFGK22} (auch in keywords?)} 
%\todo[inline]{sometimes way say solution and sometimes we say sparsified graph}

%    Graphs are a natural data model for a wide range of applications.
%    For example in transport planning road networks or rail networks are modelled as graphs~\cite{yang1998models}.
%    Other applications where graphs are omnipresent are biological networks~\cite{pavlopoulos2011using} and social networks~\cite{wasserman1994social}.

  %  In this work, we focus on a family of problems in the context of social networks.
    % In social network analysis, a common task is to identify the central actors of a social network and important connections between these actors~\cite{wasserman1994social}.
    % A typical extension to the graph model in the context of social networks is the notion of communities which are groups of actors that interact closely or a have a common agenda.
    % Communities can overlap and might not be disjoint; for example if edges model relations and communities model common interests.
 %   For example, a community in a social network might represent a group of people who are interested in the same topic.
    A common goal in network analysis is to decrease the size of a given network to speed up downstream analysis algorithms or to decrease the memory footprint of the graphs. This leads to the task of network sparsification where one wants to reduce the number of edges of a network while preserving some important property~$\Pi$~\cite{chekuri2015element,lindner2015structure,zhou2010network}. Similarly, in network design the task is often to construct a minimum-size or minimum-weight network fulfilling a given property, the most famous example being \textsc{Minimum-Weight Spanning Tree}.
    
    In many applications the input contains, in addition to a network, a hypergraph on the same vertex set~\cite{FK21,HosodaHIOSW12,NHJ+18}. 
    The hyperedges of this hypergraph represent, for example, communities that are formed within the network.
    In presence of such community data, the sparsified network should preserve a property not for the whole network but instead for each community, that is, for each hyperedge of the hypergraph. This leads to the family of \emph{community-aware network sparsification} problems, subsumed % which aims to reduce the number of edges of a graph while preserving some graph property for each community.
 by the following problem introduced by Gionis et al.~\cite{gionis2017community}.\iflong \footnote{Compared to Gionis et al.~\cite{gionis2017community}, we consider the more general scenario where we simultaneously constrain the edge number and total edge weight of the solution graph.}\fi
    %In the corresponding problem \textsc{$\Pi$-Network Sparsification}, introduced by  we are given an edge-weighted graph~$G$ and a set of vertex subsets, referred to as \emph{communities} and numbers~$\ell$ and~$b$, and the task is to find a spanning subgraph~$G'$ of~$G$ consisting of at most $\ell$~edges of total weight at most~$b$ such that each community satisfies property~$\Pi$ in~$G'$.
   % This leads to the family of $\Pi$-\textsc{Network Sparsification} ($\Pi$-\textsc{NWS}) problems which was introduced by Gionis et al.~\cite{gionis2017community}:
 \problemdef{$\Pi$-Network Sparsification ($\Pi$-NWS)}{A graph~$G$, a collection~$\mathcal{C}$ of~$c$ subsets of~$V(G)$, called \emph{communities}, an edge-weight function~$\omega: E(G)\to \mathds{R}_+$, an integer $\ell$, and a positive real number~$b$.}{Is there a graph~$G' = (V(G), E')$ with~$E' \subseteq E(G)$, $|E'|\le\ell$, and total edge weight at most~$b$ such that for all communities $C_i \in \mathcal{C}$ the subgraph of~$G$ induced by~$C_i$ satisfies~$\Pi$?}

     We say that a graph~$G'$ fulfilling the requirements is a~\emph{solution} for the instance~$I$.
 An example instance of~\textsc{$\Pi$-NWS} and solutions for the two specific properties~$\Pi$ studied in this work is given in~\Cref{fig:example-probdef}.
A very well-studied property~$\Pi$, considered by~Gionis et al.~\cite{gionis2017community} but also in many previous works~\cite{angluin2015network,chen2015polynomial,DuM88,fan2008algorithms,KMN14}, in this context is that every community should induce a connected subgraph. When a graph~$G$ admits this property for some hypergraph~$H$, then~$G$ is called a \emph{support} for~$H$~\cite{BrandesCPS12,BKMSV11,KMN14}.
 %   A problem belonging the problem family $\Pi$-\textsc{NWS} gets as input an undirected graph $G=(V,E)$, a set of subsets over $V$ called communities and an integer $\ell$.
  %  In other words, the input consists of an undirected graph and a hypergraph over the same set of vertices where the hyperedges are the communities and an integer $\ell$.
   % It then asks whether there exists a spanning subgraph $G'$ of $G$ with at most $\ell$ edges such that each subgraph in $G'$ induced by a community satisfies the graph property $\Pi$.
%In this work, we assume that the graph property~$\Pi$ is computable in polynomial time.
 We denote the corresponding special case of \textsc{$\Pi$-NWS} as \Con. Another variant of~\textsc{$\Pi$-NWS}, also studied by Gionis et al.~\cite{gionis2017community}, is to demand that every community not only induces a connected subgraph but more strongly that it contains a \emph{spanning star}. In other words, in the solution graph~$G'$, every community must be contained in the neighborhood of at least one of its vertices, called a \emph{center vertex}. We refer to this variant as \Stars.

\Con and \Stars are both NP-hard~\cite{DuM88,CDS04,CHM+18,gionis2017community}. Motivated by this, we study both problems in terms of their parameterized and fine-grained complexity. 
We also investigate the versions of both problems where each edge has unit weight and refer to them as \UCon and \UStars.

Our two main results are as follows:
\begin{itemize}
\item We show that, based on the Exponential Time Hypothesis (ETH), \Con{} and \Stars{} do not admit algorithms with running time~$2^{o(n^2)+c} \poly(n+c)$, even if the input graph is a clique with unit weights and each community has size at most~4.  
This bound is matched by simple brute-force algorithms.
  \item We show that \Stars{} admits an XP-algorithm when parameterized by~$t$, the feedback edge number of the solution graph. This positively answers the question of Korach and Stern~\cite{korach2008complete} who asked whether there is a polynomial-time algorithm for finding an optimal solution for \Stars{} that is a tree. In fact, our algorithm extends the polynomial-ime solvable cases to solutions that are tree-like.
  \end{itemize}
We obtain several further results, for example a complexity dichotomy for \Stars{} and \UStars{} parameterized by~$c$, the number of communities.  
% Clearly, hardness results of the unweighted versions also transfer to the weighted versions.
\iflonglong
\todo{edge-weight vs edge weight}\fi
%    Gionis et al.~\cite{gionis2017community} studied $\Pi$-\textsc{NWS} with respect to three specific graph properties in their work.
%    The first property is a minimum density requirement, the second property is the containment of a spanning star, and third property is the property of being connected.
%    In this work, we study the complexity of $\Pi$-\textsc{NWS} for these graph properties specificall.
%    Since $\Pi$-\textsc{NWS} is NP-hard for these three graph properties, we study $\Pi$-\textsc{NWS} for these three graph properties in the framework of parameterized complexity.
%    In addition to the parameterized complexity, we study the fine-grained complexity for these three graph properties.
    %Moreover, we study the complexity of $\Pi$-\textsc{NWS} restricted to communities of size at most $3$ in general.
\begin{figure}[t]

\centering
\begin{tikzpicture}[x=0.75pt,y=0.75pt,yscale=-1.25,xscale=1.25]
\draw (60,30) node [anchor=north west][inner sep=0.75pt]    {$a)$};

\draw   [thick]  (130.07,99.93) -- (70.07,99.93) ;
\draw   [thick]  (100.07,40.54) -- (70.07,99.93) ;
\draw   [thick]  (100.07,40.54) -- (130.07,99.93) ;
\draw   [thick]  (130.07,99.93) -- (100.07,119.93) ;
\draw   [thick]  (100.08,80.04) -- (100.07,119.93) ;
\draw   [thick]  (130.07,99.93) -- (100.08,80.04) ;
\draw   [thick]  (70.07,99.93) -- (100.07,119.93) ;
\draw   [thick]  (70.07,99.93) -- (100.08,80.04) ;

\draw  [fill={rgb, 255:red, 255; green, 255; blue, 255 }  ,fill opacity=1 ] (102.58,80.1) .. controls (102.54,81.48) and (101.4,82.57) .. (100.02,82.54) .. controls (98.64,82.5) and (97.55,81.35) .. (97.58,79.97) .. controls (97.61,78.59) and (98.76,77.5) .. (100.14,77.54) .. controls (101.52,77.57) and (102.61,78.72) .. (102.58,80.1) -- cycle ;
\draw  [fill={rgb, 255:red, 255; green, 255; blue, 255 }  ,fill opacity=1 ] (72.57,100) .. controls (72.53,101.38) and (71.38,102.47) .. (70,102.43) .. controls (68.62,102.39) and (67.53,101.24) .. (67.57,99.86) .. controls (67.61,98.48) and (68.76,97.39) .. (70.14,97.43) .. controls (71.52,97.47) and (72.61,98.62) .. (72.57,100) -- cycle ;
\draw  [fill={rgb, 255:red, 255; green, 255; blue, 255 }  ,fill opacity=1 ] (132.57,100) .. controls (132.53,101.38) and (131.38,102.47) .. (130,102.43) .. controls (128.62,102.39) and (127.53,101.24) .. (127.57,99.86) .. controls (127.61,98.48) and (128.76,97.39) .. (130.14,97.43) .. controls (131.52,97.47) and (132.61,98.62) .. (132.57,100) -- cycle ;
\draw  [fill={rgb, 255:red, 255; green, 255; blue, 255 }  ,fill opacity=1 ] (102.57,120) .. controls (102.53,121.38) and (101.38,122.47) .. (100,122.43) .. controls (98.62,122.39) and (97.53,121.24) .. (97.57,119.86) .. controls (97.61,118.48) and (98.76,117.39) .. (100.14,117.43) .. controls (101.52,117.47) and (102.61,118.62) .. (102.57,120) -- cycle ;
\draw  [blue,line width=0.3mm,dashed] (97.14,74.54) .. controls (108.39,69.46) and (135.06,85.46) .. (135.72,100.12) .. controls (136.39,114.79) and (106.39,130.46) .. (97.39,125.46) .. controls (88.39,120.46) and (85.9,79.62) .. (97.14,74.54) -- cycle ;
\draw   [blue,line width=0.3mm] (104.72,74.79) .. controls (93.06,69.12) and (63.9,85.2) .. (64.57,99.86) .. controls (65.24,114.53) and (93.39,129.12) .. (104.06,125.79) .. controls (114.72,122.46) and (116.39,80.46) .. (104.72,74.79) -- cycle ;
\draw  [fill={rgb, 255:red, 255; green, 255; blue, 255 }  ,fill opacity=1 ] (102.57,40.61) .. controls (102.53,41.99) and (101.38,43.08) .. (100,43.04) .. controls (98.62,43) and (97.53,41.85) .. (97.57,40.47) .. controls (97.61,39.09) and (98.76,38) .. (100.14,38.04) .. controls (101.52,38.08) and (102.61,39.23) .. (102.57,40.61) -- cycle ;
\draw  [blue,line width=0.4mm,dotted] (100,34.33) .. controls (87.28,34.27) and (62.14,83.06) .. (63.57,98.86) .. controls (65,114.67) and (133.33,115.67) .. (135.33,100.33) .. controls (137.33,85) and (112.72,34.4) .. (100,34.33) -- cycle ;

% Figure 2       
\draw (160,30) node [anchor=north west][inner sep=0.75pt]    {$b)$};  
          
\draw   [thick]  (230.07,99.93) -- (170.07,99.93) ;
\draw   [thick]  (200.07,40.54) -- (230.07,99.93) ;
\draw   [thick]  (230.07,99.93) -- (200.07,119.93) ;
\draw   [thick]  (170.07,99.93) -- (200.07,119.93) ;
\draw   [red,very thick]  (200.07,40.54) -- (170.07,99.93) ;
\draw   [red,very thick]  (200.08,80.04) -- (200.07,119.93) ;
\draw   [red,very thick]  (230.07,99.93) -- (200.08,80.04) ;
\draw   [red,very thick]  (170.07,99.93) -- (200.08,80.04) ;

\draw  [fill={rgb, 255:red, 255; green, 255; blue, 255 }  ,fill opacity=1 ] (202.58,80.1) .. controls (202.54,81.48) and (201.4,82.57) .. (200.02,82.54) .. controls (198.64,82.5) and (197.55,81.35) .. (197.58,79.97) .. controls (197.61,78.59) and (198.76,77.5) .. (200.14,77.54) .. controls (201.52,77.57) and (202.61,78.72) .. (202.58,80.1) -- cycle ;
\draw  [fill={rgb, 255:red, 255; green, 255; blue, 255 }  ,fill opacity=1 ] (172.57,100) .. controls (172.53,101.38) and (171.38,102.47) .. (170,102.43) .. controls (168.62,102.39) and (167.53,101.24) .. (167.57,99.86) .. controls (167.61,98.48) and (168.76,97.39) .. (170.14,97.43) .. controls (171.52,97.47) and (172.61,98.62) .. (172.57,100) -- cycle ;
\draw  [fill={rgb, 255:red, 255; green, 255; blue, 255 }  ,fill opacity=1 ] (232.57,100) .. controls (232.53,101.38) and (231.38,102.47) .. (230,102.43) .. controls (228.62,102.39) and (227.53,101.24) .. (227.57,99.86) .. controls (227.61,98.48) and (228.76,97.39) .. (230.14,97.43) .. controls (231.52,97.47) and (232.61,98.62) .. (232.57,100) -- cycle ;
\draw  [fill={rgb, 255:red, 255; green, 255; blue, 255 }  ,fill opacity=1 ] (202.57,120) .. controls (202.53,121.38) and (201.38,122.47) .. (200,122.43) .. controls (198.62,122.39) and (197.53,121.24) .. (197.57,119.86) .. controls (197.61,118.48) and (198.76,117.39) .. (200.14,117.43) .. controls (201.52,117.47) and (202.61,118.62) .. (202.57,120) -- cycle ;
\draw  [blue,line width=0.3mm,dashed] (197.14,74.54) .. controls (208.39,69.46) and (235.06,85.46) .. (235.72,100.12) .. controls (236.39,114.79) and (206.39,130.46) .. (197.39,125.46) .. controls (188.39,120.46) and (185.9,79.62) .. (197.14,74.54) -- cycle ;
\draw  [blue,line width=0.3mm] (204.72,74.79) .. controls (193.06,69.12) and (163.9,85.2) .. (164.57,99.86) .. controls (165.24,114.53) and (193.39,129.12) .. (204.06,125.79) .. controls (214.72,122.46) and (216.39,80.46) .. (204.72,74.79) -- cycle ;
\draw  [fill={rgb, 255:red, 255; green, 255; blue, 255 }  ,fill opacity=1 ] (202.57,40.61) .. controls (202.53,41.99) and (201.38,43.08) .. (200,43.04) .. controls (198.62,43) and (197.53,41.85) .. (197.57,40.47) .. controls (197.61,39.09) and (198.76,38) .. (200.14,38.04) .. controls (201.52,38.08) and (202.61,39.23) .. (202.57,40.61) -- cycle ;
\draw  [blue,line width=0.4mm,dotted] (200,34.33) .. controls (187.28,34.27) and (162.14,83.06) .. (163.57,98.86) .. controls (165,114.67) and (233.33,115.67) .. (235.33,100.33) .. controls (237.33,85) and (212.72,34.4) .. (200,34.33) -- cycle ;

% Figure 3            
\draw (260,30) node [anchor=north west][inner sep=0.75pt]    {$c)$};  
\draw   [thick]  (300.07,40.54) -- (330.07,99.93) ;
\draw   [thick]  (300.08,80.04) -- (300.07,119.93) ;
\draw    [red,very thick] (230.07+100,99.93) -- (170.07+100,99.93) ;
\draw   [red,very thick]  (200.07+100,40.54) -- (170.07+100,99.93) ;
\draw   [red,very thick]  (230.07+100,99.93) -- (200.07+100,119.93) ;
\draw   [red,very thick]  (200.08+100,80.04) -- (230.07+100,99.93) ;
\draw    [red,very thick] (170.07+100,99.93) -- (200.07+100,119.93) ;
\draw    [red,very thick] (170.07+100,99.93) -- (200.08+100,80.04) ;

\draw  [fill={rgb, 255:red, 255; green, 255; blue, 255 }  ,fill opacity=1 ] (202.58+100,80.1) .. controls (202.54+100,81.48) and (201.4+100,82.57) .. (200.02+100,82.54) .. controls (198.64+100,82.5) and (197.55+100,81.35) .. (197.58+100,79.97) .. controls (197.61+100,78.59) and (198.76+100,77.5) .. (200.14+100,77.54) .. controls (201.52+100,77.57) and (202.61+100,78.72) .. (202.58+100,80.1) -- cycle ;
\draw  [fill={rgb, 255:red, 255; green, 255; blue, 255 }  ,fill opacity=1 ] (172.57+100,100) .. controls (172.53+100,101.38) and (171.38+100,102.47) .. (170+100,102.43) .. controls (168.62+100,102.39) and (167.53+100,101.24) .. (167.57+100,99.86) .. controls (167.61+100,98.48) and (168.76+100,97.39) .. (170.14+100,97.43) .. controls (171.52+100,97.47) and (172.61+100,98.62) .. (172.57+100,100) -- cycle ;
\draw  [fill={rgb, 255:red, 255; green, 255; blue, 255 }  ,fill opacity=1 ] (232.57+100,100) .. controls (232.53+100,101.38) and (231.38+100,102.47) .. (230+100,102.43) .. controls (228.62+100,102.39) and (227.53+100,101.24) .. (227.57+100,99.86) .. controls (227.61+100,98.48) and (228.76+100,97.39) .. (230.14+100,97.43) .. controls (231.52+100,97.47) and (232.61+100,98.62) .. (232.57+100,100) -- cycle ;
\draw  [fill={rgb, 255:red, 255; green, 255; blue, 255 }  ,fill opacity=1 ] (202.57+100,120) .. controls (202.53+100,121.38) and (201.38+100,122.47) .. (200+100,122.43) .. controls (198.62+100,122.39) and (197.53+100,121.24) .. (197.57+100,119.86) .. controls (197.61+100,118.48) and (198.76+100,117.39) .. (200.14+100,117.43) .. controls (201.52+100,117.47) and (202.61+100,118.62) .. (202.57+100,120) -- cycle ;
\draw  [blue,line width=0.3mm,dashed] (197.14+100,74.54) .. controls (208.39+100,69.46) and (235.06+100,85.46) .. (235.72+100,100.12) .. controls (236.39+100,114.79) and (206.39+100,130.46) .. (197.39+100,125.46) .. controls (188.39+100,120.46) and (185.9+100,79.62) .. (197.14+100,74.54) -- cycle ;
\draw  [blue,line width=0.3mm] (204.72+100,74.79) .. controls (193.06+100,69.12) and (163.9+100,85.2) .. (164.57+100,99.86) .. controls (165.24+100,114.53) and (193.39+100,129.12) .. (204.06+100,125.79) .. controls (214.72+100,122.46) and (216.39+100,80.46) .. (204.72+100,74.79) -- cycle ;
\draw  [fill={rgb, 255:red, 255; green, 255; blue, 255 }  ,fill opacity=1 ] (202.57+100,40.61) .. controls (202.53+100,41.99) and (201.38+100,43.08) .. (200+100,43.04) .. controls (198.62+100,43) and (197.53+100,41.85) .. (197.57+100,40.47) .. controls (197.61+100,39.09) and (198.76+100,38) .. (200.14+100,38.04) .. controls (201.52+100,38.08) and (202.61+100,39.23) .. (202.57+100,40.61) -- cycle ;
\draw  [blue,line width=0.4mm,dotted] (200+100,34.33) .. controls (187.28+100,34.27) and (162.14+100,83.06) .. (163.57+100,98.86) .. controls (165+100,114.67) and (233.33+100,115.67) .. (235.33+100,100.33) .. controls (237.33+100,85) and (212.72+100,34.4) .. (200+100,34.33) -- cycle ;
\end{tikzpicture}
\caption{$a)$ The communities (depicted in \textcolor{blue}{blue}) and the input graph of an instance of \textsc{$\Pi$-NWS} with unit weights. $b)$ and~$c)$  Optimal solutions (in \textcolor{red}{red}) for \UCon, and \UStars, respectively.}
\label{fig:example-probdef}
\end{figure}
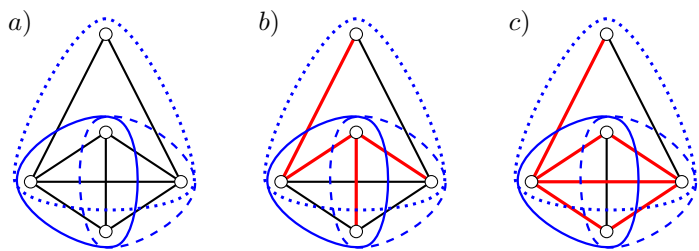

\subparagraph{Known Results.}

Already the most basic variant of \Con, where the edges have unit weights and
the input graph~$G$ is a clique, appears in many applications, ranging from explanation of protein complexes~\cite{NHJ+18} to combinatorial auctions~\cite{CDS04} to the construction of P2P overlay networks in publish/subscribe systems~\cite{ChocklerMTV07,HosodaHIOSW12}.
Consequently, the problem has been studied intensively under various names~\cite{angluin2015network,chen2015polynomial,ChocklerMTV07,DuM88,fan2008algorithms,HosodaHIOSW12}
from a parameterized complexity~\cite{chen2015polynomial,fan2008algorithms,HosodaHIOSW12}
and an approximation algorithms~\cite{angluin2015network,ChocklerMTV07,HosodaHIOSW12} perspective. 
For example, Du and Miller~\cite{DuM88} showed NP-hardness even for instances with maximum community size~$3$, and Chen et
al.~\cite{chen2015polynomial} presented an~FPT-algorithm for the
number of communities and an FPT-algorithm for the largest community size plus
the feedback edge number~$t$ of a solution. From a more practical perspective, (Mixed) Integer Linear Programming formulations were proposed~\cite{bonnet2019constraint,dar2021integer}. A particular restriction of the problem is to determine whether there is an acyclic solution. Such solutions are called \emph{tree supports} or \emph{clustered spanning trees} and hypergraphs that have a tree support are also known as \emph{hypertrees}. It can be determined in polynomial-time whether a hypergraph is a hypertree and different polynomial-time algorithms have been described over the years~\cite{BFMY83,CDS04,duchet1978,flament1978,GRS21,JohnsonP87,slater1978,TY84}. 

\UCon{}, with unit weights but with general input graphs~$G$, has applications in the context of placing green bridges~\cite{FK21,FK21A,HFGK22}.
\UCon{} is NP-hard even when the maximum degree of~$G$ is
3~\cite{HFGK22} and even for seven communities~\cite{FK21A}. 
%  Moreover, \UCon{} does not admit a polynomial kernel for~$\ell$, the number of edges in the solution graph~$G'$, unless NP
% $\subseteq$ coNP/poly~\cite{FK21A}.  
On the positive side, one can construct in polynomial time a tree support if one exists~\cite{GuttmannBeckSS19,KMN14,KorachS03}. 
%On the positive side, one can decide in polynomial time whether a \UCon instance has a solution graph that is a forest~\cite{GuttmannBeckSS19,KMN14,KorachS03}. 
   
For~\Con{} where we may have arbitrary edge-weights, the distinction whether or not~$G$
is restricted to be a clique vanishes: any
non-clique input graph~$G$ may be transformed into a clique by adding
the missing edges with a prohibitively large edge weight. The problem of finding a minimum-weight tree support received attention due to its applications in network visualization~\cite{KMN14}. As shown by~Korach and Stern~\cite{KorachS03} and~Klemz et al.~\cite{KMN14}, one can compute minimum-weight tree supports in polynomial time. Gioinis et al.~\cite{gionis2017community} provided approximation algorithms for the general problem.

\Stars{} has received less attention than \Con. Gionis et al.~\cite{gionis2017community} showed NP-hardness and provided approximation algorithms.
Korach and Stern~\cite{korach2008complete} studied a variant of \Stars{} where the input graph is a clique and the solution is constrained to be a tree~$T$ where the closed neighborhood of the center vertex of a community~$C_i$ is exactly the community~$C_i$.
This implies that two different communities need to have different center vertices and thus restricts the allowed set of solution graphs strictly compared to~\Stars.
    Korach and Stern~\cite{korach2008complete} showed that this problem is solvable in polynomial time~\cite{korach2008complete}.
As an open question, they ask whether this positive result can be lifted to \Stars.

Cohen et al.~\cite{CHM+18} studied the \textsc{Minimum $\mathcal{F}$-Overlay} problem which can be viewed as the following special case of~$\Pi$-NWS: The input graph~$G$ is a clique and all edges have unit weight; $\mathcal{F}$ is a family of graphs and the property~$\Pi$ is to have some spanning subgraph which is contained in~$\mathcal{F}$. It should be noted that~\Con{} and~\UStars{} with clique input graphs are special cases of \textsc{Minimum $\mathcal{F}$-Overlay}. Cohen et al.~\cite{CHM+18} provide a complexity dichotomy with respect to properties of~$\mathcal{F}$. For most cases of~$\mathcal{F}$,  \textsc{Minimum $\mathcal{F}$-Overlay} is NP-hard. In particular, the dichotomy of Cohen et al.~\cite{CHM+18} shows that~\UStars{} is NP-hard even when~$G$ is a clique.
    Gionis et al.~\cite{gionis2017community} also studied a third property~$\Pi$ where each community needs to induce a subgraph exceeding some prespecified density. 
    \iflong They showed NP-hardness and provided some approximation algorithms. \fi
    Fluschnik and Kellerhals~\cite{FK21} considered several further graph properties~$\Pi$, for example the property of having small diameter.

\subparagraph{Our Results and Organization of the Work.}
In order to put our main results into context, we first summarize in \Cref{sec:basic-results} some hardness and tractability results that follow either from  simple observations or from previous work. They imply in particular that \Stars and~\Con have an FPT-algorithm for the parameter solution size~$\ell$ and that they are W[1]-hard with respect to~$k\coloneqq m-\ell$, the number of edges not in the solution even in the unit weight case when~$G$ is a clique. 

Then, in \Cref{sec:eth-bound} we show that \UCon{} and \UStars do not admit algorithms with running time $2^{o(n^2+c)}\cdot\poly(n+c)$ even when~$G$ is a clique. This running time lower bound is based on the Exponential Time Hypothesis (ETH)~\cite{impagliazzo2001problems}. To show the bound, we develop a compression of \textsc{3-SAT} instances~$\phi$ into graphs with $\Oh(\sqrt{|\phi|})$~vertices.

In~\Cref{sec:fes}, we then consider parameterization by~$t$, the feedback edge number of the solution graph~$G'$. This is the minimum number of edges that need to be deleted to transform the solution into a forest.\footnote{The parameter~$t$ can be computed in polynomial time as discussed in~\Cref{sec:basic-results}.} 
The study of this parameter is  motivated by the following observation: The solution size parameter $\ell$ is essentially at least as large as~$n-1$, and thus neither small in practice nor particularly interesting from an algorithmic point of view. The parameter~$t$ can thus essentially be seen as a parameterization above the lower bound~$n-1$.
Our first main result for the parameter~$t$ is an XP-algorithm for \Stars parameterized by~$t$. 
%In this algorithm we distinguish two types of cycles in the solution; the first type is unavoidable due to the structure of the hypergraph and all cycles of the second type have constant length.
%After all cycles of the second type have been found, the problem can be solved in polynomial time, as we show.
Our result positively answers the question of Korach and Stern~\cite{korach2008complete} who asked whether there is a polynomial-time algorithm for~$t=0$ and substantially extends this tractability further to every constant value of~$t$. 
We then consider the parameter~$t$ for \Con. 
% This algorithm uses the classic Kruskal algorithm~\cite{K56} on an auxiliary graph.
We show that \UCon{} is NP-hard already if~$t=1$. Thus, the polynomial-time algorithms of Korach and Stern~\cite{KorachS03} and Klemz et al.~\cite{KMN14} for~$t=0$ cannot be lifted to larger values of~$t$. 

Finally, in~\Cref{sec:hyperedges} we study the complexity of \Stars{} with respect to the number~$c$ of input communities. The problem is easily seen to be in XP via an~$n^{\Oh(c)}$-time algorithm. In light of this, we obtain the following dichotomy: \UStars{} is FPT with respect to~$c$ and \Stars{} is W[1]-hard in the most restricted case when~$G$ is a clique and all edges have weight~$1$ or~$2$. 

\iflong
For an overview of the parameterized complexity results, refer to Table~\ref{tbl:parameterized-complexity-results}.
\fi

\iflong
    \begin{table}[t]
        \caption{An overview of the parameterized complexity results.
        A~$\ddagger$ indicates that this result also holds in the unweighted case and a~$\dagger$ indicates that this result only holds in the unweighted case.
        }
        \label{tbl:parameterized-complexity-results}
        \begin{tabularx}{\textwidth}{c X X}
            Parameter  & \textsc{Stars NWS} & \textsc{Connectivity NWS} \\
            \hline
            $\ell$ & \multicolumn{2}{c}{FPT~(\Cref{prop-fpt-algos-ell-con-and-stars},~\cite{FK21}), no polynomial kernel$^\ddagger$~(\Cref{cor-hardness-star+con},~\cite{FK21})} \\ \hline
            $k\coloneqq m-\ell$ & \multicolumn{2}{c}{$\W{1}$-hard$^\ddagger$~(\Cref{cor-hardness-star+con})} \\ \hline
           \multirow{2}{*}{$t$}                            & \multirow{2}{*}{XP~(\Cref{xp t})}                   & P for~$t=0$~(\cite{KMN14}, \Cref{thm:sparse-conn-with-l-equals-n-minus-1})\\%(\Cref{thm:sparse-conn-with-l-equals-n-minus-1})\\ 
                                    &  & NP-h for~$t=1^\ddagger$~(\Cref{thm-con-np-h-for-t=1})\\ \hline
            \multirow{3}{*}{$c$}  & FPT$^\dagger$~(\Cref{thm-stars-fpt-c}) & \multirow{3}{*}{NP-h for~$c=7^\ddagger$~(\cite{FK21A})} \\
             &  no polynomial kernel$^\ddagger$~(\Cref{thm-stars-no-poly-kernel-c}) &  \\
            & \W{1}-h~(\Cref{thm-weighted-stars-w-hard-for-c}) & \\ \hline
            $\Delta$ & NP-h for~$\Delta=6^\ddagger$~(\Cref{cor-stars-np-h-const-delta}) & NP-h for~$\Delta=3^\ddagger$~(\cite{HFGK22})
        \end{tabularx}
    \end{table} 
\fi

    \section{Preliminaries and Basic Observations}
    \label{sec:basic-results}

\subparagraph{Preliminaries.}
For a set~$X$, we denote by~$\binom{X}{2}$ the collection of all size-two subsets of~$X$. 
Moreover, for positive integers~$i$ and~$j$ with~$i\leq j$, we denote by~$[i,j] \coloneqq \{k\in \mathbb{N} : i \leq k \leq j\}$.

    An~\emph{undirected graph}~$G=(V,E)$ consists of a set of vertices~$V$ and a set of edges~$E\subseteq \binom{V}{2}$.
    We denote by~$V(G)$  and~$E(G)$ the vertex and edge set of~$G$, respectively.
    Furthermore, we let~$n=n(G)\coloneqq |V(G)|$ and~$m\coloneqq |E(G)|$.
    For an edge~$e = \{u,v\}\in E$, we say that~$u$ and~$v$ are~\emph{adjacent} and that~$e$ is~\emph{incident} with~$u$ and~$v$.  
    For a vertex set~$V'\subseteq V$, we denote by~$E_G(V')\coloneqq \{\{u,v\}\in E : u\in V', v\in V'\}$ the edges between the vertices of~$V'$ in~$G$.
If~$G$ is clear from the context, we may omit the subscript.    
A graph~$G'$ is a~\emph{subgraph} of~$G$ if~$V(G') \subseteq V(G)$ and~$E(G')\subseteq E(G)$.
Moreover, a subgraph $G'$ of~$G$ is \emph{spanning} if~$V(G') = V(G)$.
For a vertex set~$V'$, we denote by~$G[V']\coloneqq (V',E_G(V'))$ the~\emph{subgraph of~$G$ induced by~$V'$}.
    For a graph $G$, a set $S \subseteq V(G)$ with $E(S) = {\binom{S}{2}}$ is called a \emph{clique}.
    A size-three clique is also called \emph{triangle}.
    \iflong A set $X \subseteq V(G)$ is an \emph{independent set} of~$G$ if the vertices of~$X$ are pairwise non-adjacent in $G$, that is, if~$E_G(X) = \emptyset$. \fi
    A graph $G$ is a \emph{star} of size~$n-1$ with \emph{center} $z \in V(G)$ if $E(G) = \{\{z, v\} : v \in V(G) \setminus \{z\}\}$.
    We say that~$G$ contains a \emph{spanning star} if some subgraph~$G'$ of~$G$  is a star of size~$n-1$.
    The center of this star is \emph{universal} for~$G$.

\iflong
    A sequence of distinct vertices $(v_1, \dots, v_\kappa)$ of a graph $G$ is a \emph{$(v_1,v_\kappa)$-path}, if $\{v_i, v_{i+1}\} \in E(G)$ for each $i \in [1, \kappa-1]$.
    Let~$(v_1, \dots, v_\kappa)$ be a path with~$\kappa \geq 3$ and $\{v_1, v_\kappa\} \in E(G)$, then we call $(v_1, \dots, v_\kappa)$ a \emph{cycle} in $G$.
    A graph without a cycle is called \emph{acyclic}.
    \fi
    An edge set~$E'\subseteq E(G)$ is a~\emph{feedback edge set} of~$G$, if the graph~$G'\coloneqq (V(G),E(G)\setminus E')$ is acyclic.
    We say that two vertices $u$ and $v$ are \emph{connected} in $G$ if~$G$ contains a path between~$u$ and~$v$.
    A graph~$G$ is \emph{connected} if each pair of vertices $u,v \in V(G)$ is connected.
    A set $S \subseteq V(G)$ is a \emph{connected component} of $G$, if $G[S]$ is connected and $S$ is inclusion-maximal with this property.
    Connectivity in hypergraphs is defined similarly:
 Two vertices~$u$ and~$v$ are \emph{connected} if there exists a sequence~$C_1, C_2, \ldots, C_p$ of hyperedges, such that~$u\in C_1$,~$v\in C_p$, and consecutive communities have nonempty intersection.
A \emph{connected component} of this hypergraph is a maximal set of vertices which are connected.
\iflong We say that the hypergraph is \emph{connected} if it has is exactly one connected component. \fi
The number~$x$ of connected components of a hypergraph can be computed in polynomial time, for example by BFS. Observe that for a minimal solution graph~$G'$ for~\Stars and~\Con, the connected components of~$G'$ are exactly the connected components of the community hypergraph.  
Thus,~$t = \ell-n+x$ and the parameter~$t$ can be computed in polynomial time for a given input instance.

    For details about parameterized complexity and the exponential-time-hypothesis (ETH), we refer to the standard monographs~\cite{downey2012parameterized,cygan2015parameterized}.
    \iflonglong\todo[inline]{More Details about FPT!}\fi

\subparagraph{Basic Observations.}
    To put our main results for \Con and \Stars into context, we state some results that either follow easily from previous work or from simple observations.
% \subsection{A generic Brute-Force Algorithm for $\Pi$-\textsc{NWS}}

    \iflong 
    The naive brute-force approach for each $\Pi$-\textsc{NWS} is to perform an exhaustive search over the $\Oh(2^m)$ possibilities to select at most~$\ell$~edges from the input graph~$G$.
    This leads to the following general statement for $\Pi$-\textsc{NWS} problems.
    \begin{proposition}
        \label{thm:trivial-brute-force-result}
        Let~$\Pi$ be a property which can be decided in $\poly(n)$~time.
        Then, $\Pi$-\textsc{NWS} is solvable in $2^m \cdot c \cdot \poly(n)$~time.
      \end{proposition}      
%      \begin{proof}
%Let~$I=(G, \mathcal{C}, \omega,\ell,b)$ be an instance of a $\Pi$-\textsc{NWS}.
%The algorithm is as follows:
%Check for each graph~$G'\coloneqq(V(G),E')$ with~$E' \subseteq E(G)$, $|E'| \leq \ell$, and~$\omega(E')\le b$ whether each subgraph induced by a community~$C \in \mathcal{C}$ fulfills the property~$\Pi$.
%If such a graph~$G'$ exists, then~$I$ is a yes-instance of~$\Pi$-\textsc{NWS}.
%If no such graph~$G'$ exists, then~$I$ is a no-instance of~$\Pi$-\textsc{NWS}.
%Recall that we assumed that the graph property~$\Pi$ is verifiable in $\poly(n)$~time for each subgraph induced by a community.
%Since there are at most $2^m$~different subsets of~$E$, the described algorithm has a running time of~$\Oh(2^m \cdot c \cdot \poly(n))$.
%\end{proof}
    \fi
For the solution size parameter~$\ell$, one can obtain the following running time.
\iflong
\begin{proposition}
\label{prop-fpt-algos-ell-con-and-stars}
% \begin{enumerate}
% \item \Con{} can be solved in $\Oh(d^{2\ell}\cdot n^2m)$~time, and in $\Oh(\ell^{2\ell}\cdot n^2m)$~time, respectively.
% \item \Stars{} can be solved in $\Oh(d^{\ell}\cdot\poly(n+c))$~time, and in $\Oh(\ell^{\ell}\cdot\poly(n+c))$~time, respectively.
% \end{enumerate}
\Con{} and~\Stars{} can be solved in $\ell^{\Oh(\ell)}\cdot\poly(n+c)$ time.
\end{proposition}
\begin{proof}
  We may use the following branching strategy first observed by Fan et al.~\cite{fan2008algorithms} for the special case of \Con{} when~$G$ is a clique and the community size is at most~$d$. Later,~Cohen et al.~\cite{CHM+18} described the same branching strategy for some special cases of \textsc{Minimum $\textsc{F}$ Overlay} including \UStars{} with clique graph~$G$. We recall the strategy for sake of completeness.  %  also showed that \textsc{Interconnection Graph Problem} admits an FPT-algorithm with respect to~$d+\ell$ by presenting a search-tree algorithm with running-time~$\Oh(d^{2\ell}\cdot n^2m)$.
% %Here, $d$ is the maximal size of any community and~$\ell$ is the number of edges in the solution.
% This algorithm is based on the observation that each community can have at most $d^2$~edges.
% In each node of the tree a set of edges~$T$ is assumed to be contained in the solution.
% Furthermore, in each search tree node a community~$C$ such that~$C$ is not connected in~$G[T]$ is picked, and all up-to $d^2$~possibilities of adding an edge in~$E(C)\setminus T$ to~$T$ are checked.
Start with an initially empty partial solution. 
If some community is not yet connected by the current partial solution, then branch into the at most~$\binom{d}{2}$ possibilities to add some edge with both endpoints in the community. The depth of the branching is~$\ell$, the overall search tree size~$d^{\Oh(\ell)}$. The branching can be naturally adapted to \UStars{} by branching into the choice of one of~$d$ centers for communities without a star.
Now observe that instances with a community of size at least~$\ell+2$ are no-instances for~\Con{} and~\Stars. Thus, we may safely assume that~$d\le \ell+1$.
% Hence, the algorithm of Fan et al.~\cite{fan2008algorithms} implies that \textsc{Interconnection Graph Problem} admits an FPT-algorithm for~$\ell$ with running time~$\Oh(\ell^{2\ell}\cdot n^2m)$. 
% A follow-up question is whether \textsc{Interconnection Graph Problem} admits a polynomial kernel for~$\ell$.
% This is unlikely since~$\ell\le n^2$, which implies that~$\ell$ is a smaller parameter than~$n$ and, according to \Cref{cor-no-poly-kern-n-incon}, a polynomial kernel for~$n$, and thus also for~$\ell$, is unlikely.
% The FPT-algorithm for~$d+\ell$ of Fan et al.~\cite{fan2008algorithms} for \textsc{Interconnection Graph Problem} can also be adapted for \Con: only add those edges to~$T$ which are contained in the input graph and finally check whether the total sum of the edge weights of the solution is at most~$b$.
% This does not change the overall running time of the algorithm.
% Also, a similar algorithm is possible for \Stars:
% Again, let~$T$ denote the set of edges contained in every solution.
% Then, the algorithm picks a community~$C$ such that~$G[T]$ does not contain a spanning star.
% Afterwards, all $d$~possibilities of a potential center are checked by adding the corresponding edges to~$T$.
% This results in an algorithm with running time~$\Oh(d^\ell\cdot\poly(n+c))$.
% Clearly, for both \Con{} and \Stars{} the observation that instances with communities of size at least~$\ell+2$ are no-instances is true.
\end{proof}
\else
\begin{proposition}[$\star$]
\label{prop-fpt-algos-ell-con-and-stars}
% \begin{enumerate}
% \item \Con{} can be solved in $\Oh(d^{2\ell}\cdot n^2m)$~time, and in $\Oh(\ell^{2\ell}\cdot n^2m)$~time, respectively.
% \item \Stars{} can be solved in $\Oh(d^{\ell}\cdot\poly(n+c))$~time, and in $\Oh(\ell^{\ell}\cdot\poly(n+c))$~time, respectively.
% \end{enumerate}
\Con{} and~\Stars{} can be solved in $\ell^{\Oh(\ell)}\cdot\poly(n+c)$ time.
\end{proposition}
\fi
The fixed-parameter tractability of \UCon{} with respect to~$\ell$ was also shown by Fluschnik and Kellerhals via a kernelization \iflong that gives kernels with at most $2\ell$~vertices\fi~\cite{FK21}.
% Analogously to \textsc{Interconnection Graph Problem}, since~$\ell$ is a smaller parameter than~$n$, a polynomial kernel for \Con{} or \Stars{} with respect to~$\ell$ is unlikely.

% Also, note that from \Cref{prop-eth-bound-m-and-c} and from~$\ell\le m$, we obtain that the simple brute-force algorithms  in \Cref{prop-fpt-algos-ell-con-and-stars} are tight. 

%     \begin{corollary}
%         If the ETH is true, then \UStars and \UCon cannot be solved in $d^{o(\ell)} \cdot \poly(n+c)$~time.
%     \end{corollary}

% \subparagraph*{Kernelization.}
% Fluschnik and Kellerhals showed that based on the observation that each edge is adjacent with 2~vertices, a kernel with at most $2\ell$~vertices (but exponential many communities) can be computed for \UCon~\cite{FK21}.
% By the same argument, one also obtains a polynomial kernel of size~$\Oh(c\ell+\ell^2)$ for \UCon~\cite{FK21}.
% Recall that~$c$ is the number of communities.
% These arguments can also be directly used for (the weighted version of) \Con and \Stars, implying a $2\ell$~vertex kernel and a kernel of size~$\Oh(c\ell+\ell^2)$.

A further natural parameter that can be considered is~$k\coloneqq m-\ell$, the number of edges that are not in the solution.
    The following proposition summarizes the complexity of the problems with respect to this parameter.
    \begin{proposition}[\iflong\else$\star$\fi]
      \label{thm:brute-force-k}
      \begin{itemize}
      \item \Con{} and \Stars{} are NP-hard for~$k=0$.
      \item  \UCon{} and \UStars{} can be solved in~$n^{2k}\cdot \poly(n)$ time and are \W{1}-hard with respect to~$k$ even if~$G$ is a clique and if each community has size at most~$3$.
      \end{itemize}
    \end{proposition}
    \iflong
      \begin{proof}
      For \Con{} and \Stars{} the NP-hardness for~$k=0$ can be seen as
      follows:  reduce from the unweighted version, set the weight of
      each edge to one, $k\coloneqq 0$, and the weight bound~$b$ to
      the solution size bound of the unweighted instance.

      For \UCon{} and \UStars{} the existence of a solution with at
      most~$\ell$ edges implies the existence of one with
      exactly~$\ell$ edges. Thus, it is sufficient to consider all
      $\binom{m}{k}$ possibilities to remove~$k$ edges from~$G$. 

      The W[1]-hardness for \UCon{} can be seen via an adaption of reduction from Du and Miller~\cite[Theorem~3.1]{DuM88}. 
      More precisely, we reduce from \textsc{Independent Set} where we are given a graph~$\widetilde{G}$ on~$\widetilde{n}$ vertices and~$\widetilde{m}$ edges and integer~$\widetilde{k}$ and want to decide whether it contains a set of at least~$\widetilde{k}$ pairwise nonadjacent vertices. 
      The graph~$G$ of the constructed \UCon{} instance is a clique on the vertex set~$V(\widetilde{G})\cup \{x\}$ where~$x$ is a new vertex. 
      For each pair~$u,v$ of vertices in~$V(\widetilde{G})$, add a community~$\{u,v\}$. 
      Then, for each edge~$\{u,v\}\in E(\widetilde{G})$, add a community~$\{u,v,x\}$. 
      Finally, set the budget~$\ell$ to~$\binom{\widetilde{n}}{2}+\widetilde{n}-\widetilde{k}$. 
      Since the total number of edges in the constructed graph~$G$ is~$\binom{\widetilde{n}}{2}+\widetilde{n}$, we have~$k=m-\ell=\widetilde{k}$. 
      The equivalence of the instance can be seen as follows: 
      every solution contains all~$\widetilde{n}^2$ edges between vertices of~$V(\widetilde{G})$. 
      Thus, every solution misses at least~$k$ edges that are incident with~$x$. 
      These edges correspond to a vertex set~$S\subseteq V(G)$ of size at least~$k$. 
      The set~$S$ is an independent set in~$\widetilde{G}$: 
      if~$S$ contains two adjacent vertices~$u$ and~$v$, then the community~$\{u,v,x\}$ does not induce a  connected subgraph. 
      The converse direction can be seen analagously.
      
The results for \UStars{} follow since for communities of size at most~$3$ the properties of being connected and having a spanning star coincide.
% This reduction can be adapted in two ways to give further hardness results.     
% First, by reducing from \textsc{Independent Set} instead of \textsc{Vertex Cover}, adding a community~$\{u,v\}$ for each pair~$u,v$ of vertices in~$V(G)$ instead of only adding a community~$\{u,v\}$ for each~$\{u,v\}\in E(G)$, and by setting~$\ell= \binom{n}{2}+k$, one gets hardness also for the parameter~$k\coloneqq m-\ell$, the number of edges not in the solution. 
\end{proof}
\fi

    % A special case of \UCon, where the input graph~$G$ is a clique, was studied under the names \textsc{Subset Interconnection Design}~\cite{chen2015polynomial}, \textsc{Network Construction}~\cite{angluin2015network} and \textsc{Interconnection Graph Problem}~\cite{fan2008algorithms}.
\iflong
    By a reduction from \textsc{Vertex Cover}, Du and Miller~\cite{DuM88} showed  that the special case of \UCon{} where~$G$ is a clique is NP-hard even if restricted to communities of size at most~$3$.
% \begin{proposition}
% \label{cor-w-hard-k-incon}
% \textsc{Interconnection Graph Problem} is \W{1}-hard with respect to~$k$ even if each community has size at most~$3$.
% \end{proposition}
% \begin{proposition}
% \label{cor-no-poly-kern-n-incon}
% \textsc{Interconnection Graph Problem} does not admit a polynomial kernel for~$n$ even if each community has size at most~$3$, unless NP $\subseteq$ coNP/poly.
% \end{proposition}
% Since \textsc{Interconnection Graph Problem} is a special case of \Con,  \Cref{cor-w-hard-k-incon,cor-no-poly-kern-n-incon} also hold for \Con.
\else
% Du and Miller~\cite[Theorem~3.1]{DuM88} showed by a reduction from \textsc{Vertex Cover} that the special case of \UCon{} where~$G$ is a clique is NP-hard even if restricted to communities of size at most~$3$. 
% Since the problem was studied under several names, a similar reduction was presented by several papers~\cite{ChocklerMTV07,fan2008algorithms,HosodaHIOSW12,HFGK22}.
% Next, we adapt this reduction to obtain the following results:
%  By adapting a reduction from \textsc{Vertex Cover} due to Fan et al.~\cite{fan2008algorithms}, we can get the following hardness results for the parameters~$k$ and~$n$.
\fi
For \UCon, Fluschnik and Kellerhals showed that a polynomial kernel for~$\ell$ and (thus for~$n$) is unlikely, even on planar series-parallel graphs~\cite{FK21A}. This result can be also given for the case when the input graph~$G$ is a clique.  
\iflong
\begin{proposition}[]
\label{cor-hardness-star+con}
\begin{enumerate}
 \UCon{} and \UStars{} do not admit a polynomial kernel for~$n$ even if~$G$ is a clique, unless NP $\subseteq$ coNP/poly.
\end{enumerate}
\end{proposition}
\else
\begin{proposition}[$\star$]
\label{cor-hardness-star+con}
 \UCon{} and \UStars{} do not admit a polynomial kernel for~$n$ even if~$G$ is a clique, unless NP $\subseteq$ coNP/poly.
\end{proposition}
\fi
\iflong
  \begin{proof}
    We give a proof for \UCon{}, the result for \UStars{} follows again since for communities of size at most~$3$ the two properties coincide.

    To show the claim we give a polynomial-parameter reduction from \textsc{Hitting Set} parameterized by the size of the universe to \UCon{} parameterized by~$n$.
    \textsc{Hitting Set} parameterized by the size of the universe is known to not admit a polynomial kernel unless NP $\subseteq$ coNP/poly~\cite{cygan2015parameterized}.
        This reduction is again an adaption of the \textsc{Vertex Cover} reduction of Du and Miller~\cite{DuM88}.
        In \textsc{Hitting Set} we are given a set family~$\mathcal{F}$ over a universe~$U$ and an integer~$k$ and ask if there is a set~$S\subseteq U$ of size at most~$k$ such that for each set~$F\in \mathcal{F}$, we have~$S\cap F\neq \emptyset$. 
        For a given \textsc{Hitting Set} instance~$(\mathcal{F},U,k)$ we construct a graph~$G$ with communities as follows. 
        The graph~$G$ is a clique with vertex set~$U\cup \{x\}$ for some~$x\notin U$. For each pair of elements~$u,v\in U$, add a community~$\{u,v\}$. Moreover, for each set~$F\in  \mathcal{F}$, add the community~$F\cup \{x\}$. Finally, set the budget~$\ell$ to~$\binom{|U|}{2}+k$. Observe that~$|V(G)|=|U|+1$. 
        
The correctness of the reduction follows from the fact that all edges between vertices of~$U$ are fixed, that is, are contained in a community of size~$2$, and the up to~$k$ edges that are added between~$x$ and~$U$ must have one endpoint in~$F$ for each~$F\in \mathcal{F}$. Thus, these edges correspond to a solution for~$(\mathcal{F},U,k)$.
\end{proof}
\fi

\iflong
  We can also show that the simple brute-force algorithm behind \Cref{thm:trivial-brute-force-result} cannot be improved substantially.
  \begin{proposition}
\label{prop-eth-bound-m-and-c}
If the ETH is true, then \UStars and \UCon cannot be solved in $2^{o(n+m+c)} \cdot \poly(n+c)$~time, even if restricted to instances with community size at most~$3$.
\end{proposition}
\begin{proof}
  This can again be shown by adaption of the reduction of Du and Miller~\cite{DuM88}.
The constructed graph is not a clique; instead it consists of the graph~$G$ of the \textsc{Vertex Cover} instance plus the new universal vertex~$x$.
Now  for each edge~$\{u,v\}\in E(G)$, add the two communities~$\{u,v\}$ and~$\{u,v,x\}$. Finally, we set~$\ell\coloneqq m+k$.
Since the size of the resulting graph and the number of communities is linear in the size of~$G$, and \textsc{Vertex Cover} cannot be solved in time~$2^{o(|V(G)|+|E(G)|}$ if the ETH is true, we obtain the statement. 
\end{proof}
\else
  One can also show that the simple brute-force algorithm that considers all~$\Oh(2^m)$ subsets of~$E(G)$ cannot be improved substantially.
\begin{proposition}[$\star$]
\label{prop-eth-bound-m-and-c}
If the ETH is true, then \UStars and \UCon cannot be solved in $2^{o(n+m+c)} \cdot \poly(n+c)$~time, even if restricted to instances with community size at most~$3$.
\end{proposition} 
\fi

\section{A Stronger ETH-Bound}
\label{sec:eth-bound}

In \Cref{prop-eth-bound-m-and-c} we observed that algorithms with running time~$2^{o(n+m+c)}\cdot\poly(n+c)$ for \UCon{} and \UStars would violate the ETH.
We now provide a stronger $2^{\Omega(n^2+c)}\cdot\poly(n+c)$-time lower bound for both problems. Notably, this lower bound also applies to the case when all communities have constant size.

  % \subsection{A Simple Brute-Force Algorithm} 

\iflong \subsection{ETH Lower Bound for \UStars}\label{sec:lowerbounds}\fi

First, we present the lower bound for \UStars.

    \begin{theorem}
        \label{thm:eth-based-lowerbound-for-n-square}
        If the ETH is true, then \UStars cannot be solved in $2^{o(n^2+c)} \cdot \poly(n+c)$~time, even if~$G$ is a clique and if each community has size at most~$4$.
    \end{theorem}
    \begin{proof}
        We present a reduction from $3$-\textsc{SAT} to \UStars where~$G$ is a clique with maximum community size~$4$ such that the resulting instance has $\Oh(\sqrt {|\phi|})$ vertices and $\Oh(|\phi|)$~communities, where~$\phi$ denotes the total formula length.
        Then, the existence of an $2^{o(n^2+c)} \cdot \poly(n+c)$-time algorithm for \UStars implies the existence of a $2^{o(|\phi|)}$-time algorithm for $3$-\textsc{SAT} violating the ETH~\cite{impagliazzo2001complexity,impagliazzo2001problems}.
        %Then, we adapt the reduction to $4$-\textsc{Stars NWS} and $4$-\textsc{Connectivity NWS}.
The input formula~$\phi$ is over the variable set~$X$ and each clause~$q\in\Gamma$ contains exactly three literals. For a literal~$y$, we denote by~$\overline{y}$ its complement.
        A visualization of the construction is given in \Cref{fig-eth-bound-stars}.
		%$V(G)$ is partitioned into four~sets~$U,P$ (used for the variable gadget), and~$Y,Z$ (used for the clause gadget.
		%The precise definitions of these sets are given below, when we describe the gadgets.
		In all gadgets, we add several communities of size~$2$.
		These communities enforce that each solution has to contain the edge of this community.
		In the following we call such edges \emph{fixed}.

      \begin{figure}[t]
      \centering
      \scalebox{1}{
        \begin{minipage}[c]{0.48\textwidth}
                %\begin{figure}[t]
            \centering
            \begin{tikzpicture}[x=0.75pt,y=0.75pt,yscale=-1,xscale=1]
%uncomment if require: \path (0,235); %set diagram left start at 0, and has height of 235

%Rounded Rect [id:dp21914907185499322]
                \draw  [color={rgb, 255:red, 128; green, 128; blue, 128 }  ,draw opacity=1 ][dash pattern={on 4.5pt off 4.5pt}] (60,14.35) .. controls (60,12.14) and (61.79,10.35) .. (64,10.35) -- (186,10.35) .. controls (188.21,10.35) and (190,12.14) .. (190,14.35) -- (190,26.35) .. controls (190,28.56) and (188.21,30.35) .. (186,30.35) -- (64,30.35) .. controls (61.79,30.35) and (60,28.56) .. (60,26.35) -- cycle ;
%Rounded Rect [id:dp30430333774088125]
                \draw  [color={rgb, 255:red, 128; green, 128; blue, 128 }  ,draw opacity=1 ][dash pattern={on 4.5pt off 4.5pt}] (60,84.35) .. controls (60,82.14) and (61.79,80.35) .. (64,80.35) -- (186,80.35) .. controls (188.21,80.35) and (190,82.14) .. (190,84.35) -- (190,96.35) .. controls (190,98.56) and (188.21,100.35) .. (186,100.35) -- (64,100.35) .. controls (61.79,100.35) and (60,98.56) .. (60,96.35) -- cycle ;
%Straight Lines [id:da8139092565486107]
                \draw [color={rgb, 255:red, 0; green, 0; blue, 0 }  ,draw opacity=1 ]   (90.07,19.45) -- (90.07,40.35) -- (90.07,89.45) ;
%Straight Lines [id:da9526141073355061]
                \draw [color={rgb, 255:red, 0; green, 0; blue, 0 }  ,draw opacity=1 ]  (170.07,19.45) -- (170.07,89.45) ;
%Straight Lines [id:da33422731774004166]
                \draw [color={rgb, 255:red, 0; green, 0; blue, 0 }  ,draw opacity=1 ]  (90.07,89.45) -- (170.07,19.45) ;
%Shape: Circle [id:dp19190397617745936]
                \draw  [fill={rgb, 255:red, 255; green, 255; blue, 255 }  ,fill opacity=1 ] (172.57,19.52) .. controls (172.53,20.9) and (171.38,21.99) .. (170,21.95) .. controls (168.62,21.91) and (167.53,20.76) .. (167.57,19.38) .. controls (167.61,18) and (168.76,16.91) .. (170.14,16.95) .. controls (171.52,16.99) and (172.61,18.14) .. (172.57,19.52) -- cycle ;
%Shape: Circle [id:dp8673891783679247]
                \draw  [fill={rgb, 255:red, 255; green, 255; blue, 255 }  ,fill opacity=1 ] (172.57,89.52) .. controls (172.53,90.9) and (171.38,91.99) .. (170,91.95) .. controls (168.62,91.91) and (167.53,90.76) .. (167.57,89.38) .. controls (167.61,88) and (168.76,86.91) .. (170.14,86.95) .. controls (171.52,86.99) and (172.61,88.14) .. (172.57,89.52) -- cycle ;
%Shape: Circle [id:dp9342307875050735]
                \draw  [fill={rgb, 255:red, 255; green, 255; blue, 255 }  ,fill opacity=1 ] (92.57,89.52) .. controls (92.53,90.9) and (91.38,91.99) .. (90,91.95) .. controls (88.62,91.91) and (87.53,90.76) .. (87.57,89.38) .. controls (87.61,88) and (88.76,86.91) .. (90.14,86.95) .. controls (91.52,86.99) and (92.61,88.14) .. (92.57,89.52) -- cycle ;
%Shape: Circle [id:dp8922159368881891]
                \draw  [fill={rgb, 255:red, 255; green, 255; blue, 255 }  ,fill opacity=1 ] (92.57,19.52) .. controls (92.53,20.9) and (91.38,21.99) .. (90,21.95) .. controls (88.62,21.91) and (87.53,20.76) .. (87.57,19.38) .. controls (87.61,18) and (88.76,16.91) .. (90.14,16.95) .. controls (91.52,16.99) and (92.61,18.14) .. (92.57,19.52) -- cycle ;
                \draw [color={rgb, 255:red, 0; green, 0; blue, 0 }  ,draw opacity=1 ]   (90.07,19.45) -- (170.07,89.45) ;
%Shape: Rectangle [id:dp20304609709206578]
%Shape: Rectangle [id:dp45874477995335905]
                \draw  [color={rgb, 255:red, 0; green, 0; blue, 0 }  ,draw opacity=0 ][fill={rgb, 255:red, 255; green, 255; blue, 255 }  ,fill opacity=1 ] (139.77,32.58) -- (156.77,32.58) -- (156.77,47.58) -- (139.77,47.58) -- cycle ;
%Shape: Rectangle [id:dp8154418999770786]
                \draw  [color={rgb, 255:red, 0; green, 0; blue, 0 }  ,draw opacity=0 ][fill={rgb, 255:red, 255; green, 255; blue, 255 }  ,fill opacity=1 ] (164.77,42.58) -- (181.77,42.58) -- (181.77,63.08) -- (164.77,63.08) -- cycle ;
%Shape: Polygon Curved [id:ds5162895482390107]
                \draw  [blue,line width=0.3mm] (174.27,13.08) .. controls (163.77,3.08) and (79.27,84.58) .. (80.77,95.08) .. controls (82.27,105.58) and (171.77,106.58) .. (178.77,97.08) .. controls (185.77,87.58) and (184.77,23.08) .. (174.27,13.08) -- cycle ;
%Shape: Polygon Curved [id:ds05084983112328734]
                \draw  [blue,line width=0.4mm,dashed] (87.27,13.08) .. controls (95.77,4.08) and (185.27,82.58) .. (177.77,92.58) .. controls (170.27,102.58) and (90.77,101.58) .. (84.77,94.08) .. controls (78.77,86.58) and (78.77,22.08) .. (87.27,13.08) -- cycle ;
%Shape: Rectangle [id:dp8830996207263513]
                \draw  [color=red  ,draw opacity=0 ][fill={rgb, 255:red, 255; green, 255; blue, 255 }  ,fill opacity=1 ] (102.77,32.58) -- (119.77,32.58) -- (119.77,47.58) -- (102.77,47.58) -- cycle ;
                \draw  [color={rgb, 255:red, 0; green, 0; blue, 0 }  ,draw opacity=0 ][fill={rgb, 255:red, 255; green, 255; blue, 255 }  ,fill opacity=1 ] (84,40.35) -- (99.27,40.35) -- (99.27,60.08) -- (84,60.08) -- cycle ;

% Text Node
                \draw (66,14.59) node [anchor=north west][inner sep=0.75pt]    {$U$};
% Text Node
                \draw (65.65,85.25) node [anchor=north west][inner sep=0.75pt]    {$P$};
% Text Node
                \draw (106,34.9) node [anchor=north west][inner sep=0.75pt]    {$x_{1}$};
% Text Node
                \draw (84.5,43.4) node [anchor=north west][inner sep=0.75pt]    {$\overline{x}_{1}$};
% Text Node
                \draw (165,48.4) node [anchor=north west][inner sep=0.75pt]    {$\overline{x}_{2}$};
% Text Node
                \draw (140.5,35.4) node [anchor=north west][inner sep=0.75pt]    {$x_{2}$};
                \draw (60,-10) node [anchor=north west][inner sep=0.75pt]    {$a)$};

            \end{tikzpicture}
            %\caption{An example of the construction showing how the variables are represented and how the $P_3$s are arranged}\label{fig:variable-gadget}
        %\end{figure}
        \newline
        
        \begin{tikzpicture}[x=0.75pt,y=0.75pt,yscale=-0.8,xscale=0.8]
%uncomment if require: \path (0,235); %set diagram left start at 0, and has height of 235

%Straight Lines [id:da716491056546935]
                \draw    (130.07,99.93) -- (130.07,169.93) ;
%Straight Lines [id:da7626533430015954]
                \draw    (50.07,169.93) -- (130.07,99.93) ;
%Straight Lines [id:da9212364814344547]
                \draw    (210.07,169.93) -- (130.07,99.93) ;
%Shape: Circle [id:dp02739618461266069]
                \draw  [fill={rgb, 255:red, 255; green, 255; blue, 255 }  ,fill opacity=1 ] (132.57,100) .. controls (132.53,101.38) and (131.38,102.47) .. (130,102.43) .. controls (128.62,102.39) and (127.53,101.24) .. (127.57,99.86) .. controls (127.61,98.48) and (128.76,97.39) .. (130.14,97.43) .. controls (131.52,97.47) and (132.61,98.62) .. (132.57,100) -- cycle ;
%Shape: Circle [id:dp202502035054812]
                \draw  [fill={rgb, 255:red, 255; green, 255; blue, 255 }  ,fill opacity=1 ] (212.57,170) .. controls (212.53,171.38) and (211.38,172.47) .. (210,172.43) .. controls (208.62,172.39) and (207.53,171.24) .. (207.57,169.86) .. controls (207.61,168.48) and (208.76,167.39) .. (210.14,167.43) .. controls (211.52,167.47) and (212.61,168.62) .. (212.57,170) -- cycle ;
%Shape: Circle [id:dp1775511946108973]
                \draw  [fill={rgb, 255:red, 255; green, 255; blue, 255 }  ,fill opacity=1 ] (132.57,170) .. controls (132.53,171.38) and (131.38,172.47) .. (130,172.43) .. controls (128.62,172.39) and (127.53,171.24) .. (127.57,169.86) .. controls (127.61,168.48) and (128.76,167.39) .. (130.14,167.43) .. controls (131.52,167.47) and (132.61,168.62) .. (132.57,170) -- cycle ;
%Shape: Circle [id:dp1680398873638086]
                \draw  [fill={rgb, 255:red, 255; green, 255; blue, 255 }  ,fill opacity=1 ] (52.57,170) .. controls (52.53,171.38) and (51.38,172.47) .. (50,172.43) .. controls (48.62,172.39) and (47.53,171.24) .. (47.57,169.86) .. controls (47.61,168.48) and (48.76,167.39) .. (50.14,167.43) .. controls (51.52,167.47) and (52.61,168.62) .. (52.57,170) -- cycle ;
%Shape: Polygon Curved [id:ds7669405236497756]
                \draw  [blue,line width=0.3mm] (119.93,79.78) .. controls (129.43,70.28) and (239.93,169.78) .. (229.93,179.28) .. controls (219.93,188.78) and (129.43,189.28) .. (119.93,180.28) .. controls (110.43,171.28) and (110.43,89.28) .. (119.93,79.78) -- cycle ;
%Shape: Polygon Curved [id:ds6021294895187352]
                \draw  [blue,line width=0.5mm,dotted] (139.93,79.78) .. controls (129.93,70.28) and (20.43,170.28) .. (29.93,179.28) .. controls (39.43,188.28) and (129.43,189.28) .. (139.93,180.28) .. controls (150.43,171.28) and (149.93,89.28) .. (139.93,79.78) -- cycle ;
%Shape: Polygon Curved [id:ds2822229166853619]
                \draw [blue,line width=0.4mm,dashed]  (43.93,175.78) .. controls (34.43,165.78) and (119.43,88.78) .. (129.93,89.28) .. controls (140.43,89.78) and (226.43,165.78) .. (215.93,175.78) .. controls (205.43,185.78) and (147.43,121.28) .. (130.43,121.28) .. controls (113.43,121.28) and (53.43,185.78) .. (43.93,175.78) -- cycle ;
%Shape: Rectangle [id:dp3767727989555254]
                \draw  [color={rgb, 255:red, 0; green, 0; blue, 0 }  ,draw opacity=0 ][fill={rgb, 255:red, 255; green, 255; blue, 255 }  ,fill opacity=1 ] (88.07,122.65) -- (103.07,122.65) -- (103.07,137.65) -- (88.07,137.65) -- cycle ;
%Shape: Rectangle [id:dp009077777663877962]
                \draw  [color={rgb, 255:red, 0; green, 0; blue, 0 }  ,draw opacity=0 ][fill={rgb, 255:red, 255; green, 255; blue, 255 }  ,fill opacity=1 ] (159.07,122.65) -- (174.07,122.65) -- (174.07,137.65) -- (159.07,137.65) -- cycle ;
%Shape: Rectangle [id:dp7045997335770793]
                \draw  [color={rgb, 255:red, 0; green, 0; blue, 0 }  ,draw opacity=0 ][fill={rgb, 255:red, 255; green, 255; blue, 255 }  ,fill opacity=1 ] (123.07,140.65) -- (140.07,140.65) -- (140.07,160.65) -- (123.07,160.65) -- cycle ;

% Text Node
                \draw (89.5,126.4) node [anchor=north west][inner sep=0.75pt]    {$x_{1}$};
% Text Node
                \draw (160,126.4) node [anchor=north west][inner sep=0.75pt]    {$x_{3}$};
% Text Node
                \draw (123.5,143.9) node [anchor=north west][inner sep=0.75pt]    {$\overline{x}_{2}$};
\draw (60,80) node [anchor=north west][inner sep=0.75pt]    {$b)$};
            \end{tikzpicture}
        \end{minipage}
        \begin{minipage}[c]{0.48\textwidth}
        \begin{tikzpicture}[x=0.75pt,y=0.75pt,yscale=-1,xscale=1]
%uncomment if require: \path (0,235); %set diagram left start at 0, and has height of 235

%Rounded Rect [id:dp08361057590283216]
                \draw  [color={rgb, 255:red, 128; green, 128; blue, 128 }  ,draw opacity=1 ][dash pattern={on 4.5pt off 4.5pt}] (20,124.83) .. controls (20,122.62) and (21.79,120.83) .. (24,120.83) -- (216,120.83) .. controls (218.21,120.83) and (220,122.62) .. (220,124.83) -- (220,136.83) .. controls (220,139.04) and (218.21,140.83) .. (216,140.83) -- (24,140.83) .. controls (21.79,140.83) and (20,139.04) .. (20,136.83) -- cycle ;
%Rounded Rect [id:dp35139080252627586]
                \draw  [color={rgb, 255:red, 128; green, 128; blue, 128 }  ,draw opacity=1 ][dash pattern={on 4.5pt off 4.5pt}] (20,194.83) .. controls (20,192.62) and (21.79,190.83) .. (24,190.83) -- (216,190.83) .. controls (218.21,190.83) and (220,192.62) .. (220,194.83) -- (220,206.83) .. controls (220,209.04) and (218.21,210.83) .. (216,210.83) -- (24,210.83) .. controls (21.79,210.83) and (20,209.04) .. (20,206.83) -- cycle ;
%Straight Lines [id:da6093235686695971]
              %  \draw [color={rgb, 255:red, 200; green, 200; blue, 200 }  ,draw opacity=1 ]   (50.07,129.93) -- (50.07,150.83) -- (50.07,199.93) ;
%Straight Lines [id:da7031718115956955]
              %  \draw [color={rgb, 255:red, 200; green, 200; blue, 200 }  ,draw opacity=1 ]   (50.07,129.93) -- (210.07,199.93) ;
%Straight Lines [id:da6663946503285815]
              %  \draw [color={rgb, 255:red, 200; green, 200; blue, 200 }  ,draw opacity=1 ]   (210.07,129.93) -- (210.07,199.93) ;
%Straight Lines [id:da4816676458673772]
              %  \draw [color={rgb, 255:red, 200; green, 200; blue, 200 }  ,draw opacity=1 ]   (50.07,129.93) -- (130.07,199.93) ;
%Straight Lines [id:da8074382636473504]
               % \draw [color={rgb, 255:red, 200; green, 200; blue, 200 }  ,draw opacity=1 ]   (130.07,199.93) -- (210.07,129.93) ;
%Straight Lines [id:da40084942177584104]
                \draw    (50.07,199.93) -- (130.07,129.93) ;
%Straight Lines [id:da5372520357341299]
                %\draw [color={rgb, 255:red, 200; green, 200; blue, 200 }  ,draw opacity=1 ]   (50.07,199.93) -- (210.07,129.93) ;
%Shape: Circle [id:dp8605422561063327]
                \draw  [fill={rgb, 255:red, 255; green, 255; blue, 255 }  ,fill opacity=1 ] (132.57,130) .. controls (132.53,131.38) and (131.38,132.47) .. (130,132.43) .. controls (128.62,132.39) and (127.53,131.24) .. (127.57,129.86) .. controls (127.61,128.48) and (128.76,127.39) .. (130.14,127.43) .. controls (131.52,127.47) and (132.61,128.62) .. (132.57,130) -- cycle ;
%Shape: Circle [id:dp05053727940587327]
               % \draw  [fill={rgb, 255:red, 255; green, 255; blue, 255 }  ,fill opacity=1 ] (212.57,130) .. controls (212.53,131.38) and (211.38,132.47) .. (210,132.43) .. controls (208.62,132.39) and (207.53,131.24) .. (207.57,129.86) .. controls (207.61,128.48) and (208.76,127.39) .. (210.14,127.43) .. controls (211.52,127.47) and (212.61,128.62) .. (212.57,130) -- cycle ;
%Shape: Circle [id:dp5855304663893682]
                \draw  [fill={rgb, 255:red, 255; green, 255; blue, 255 }  ,fill opacity=1 ] (212.57,200) .. controls (212.53,201.38) and (211.38,202.47) .. (210,202.43) .. controls (208.62,202.39) and (207.53,201.24) .. (207.57,199.86) .. controls (207.61,198.48) and (208.76,197.39) .. (210.14,197.43) .. controls (211.52,197.47) and (212.61,198.62) .. (212.57,200) -- cycle ;
%Shape: Circle [id:dp022317632892974393]
                \draw  [fill={rgb, 255:red, 255; green, 255; blue, 255 }  ,fill opacity=1 ] (132.57,200) .. controls (132.53,201.38) and (131.38,202.47) .. (130,202.43) .. controls (128.62,202.39) and (127.53,201.24) .. (127.57,199.86) .. controls (127.61,198.48) and (128.76,197.39) .. (130.14,197.43) .. controls (131.52,197.47) and (132.61,198.62) .. (132.57,200) -- cycle ;
%Shape: Circle [id:dp8786252257647231]
                \draw  [fill={rgb, 255:red, 255; green, 255; blue, 255 }  ,fill opacity=1 ] (52.57,200) .. controls (52.53,201.38) and (51.38,202.47) .. (50,202.43) .. controls (48.62,202.39) and (47.53,201.24) .. (47.57,199.86) .. controls (47.61,198.48) and (48.76,197.39) .. (50.14,197.43) .. controls (51.52,197.47) and (52.61,198.62) .. (52.57,200) -- cycle ;
                \draw  [color={rgb, 255:red, 128; green, 128; blue, 128 }  ,draw opacity=1 ][dash pattern={on 4.5pt off 4.5pt}] (60,14.35) .. controls (60,12.14) and (61.79,10.35) .. (64,10.35) -- (186,10.35) .. controls (188.21,10.35) and (190,12.14) .. (190,14.35) -- (190,26.35) .. controls (190,28.56) and (188.21,30.35) .. (186,30.35) -- (64,30.35) .. controls (61.79,30.35) and (60,28.56) .. (60,26.35) -- cycle ;
%Rounded Rect [id:dp4363804618437406]
                \draw  [color={rgb, 255:red, 128; green, 128; blue, 128 }  ,draw opacity=1 ][dash pattern={on 4.5pt off 4.5pt}] (60,84.35) .. controls (60,82.14) and (61.79,80.35) .. (64,80.35) -- (186,80.35) .. controls (188.21,80.35) and (190,82.14) .. (190,84.35) -- (190,96.35) .. controls (190,98.56) and (188.21,100.35) .. (186,100.35) -- (64,100.35) .. controls (61.79,100.35) and (60,98.56) .. (60,96.35) -- cycle ;
%Straight Lines [id:da7866751268255155]
                \draw [color={rgb, 255:red, 0; green, 0; blue, 0 }  ,draw opacity=1 ]   (90.07,19.45) -- (90.07,40.35) -- (90.07,89.45) ;
                \draw   [red] (50.07,199.93) -- (90.07,89.45) ;
                \draw   [red] (90.07,89.45) -- (130.07,129.93) ;
                \draw   [red] (50.07,199.93) to  [bend right=25] (90.07,19.45) ;
%Straight Lines [id:da013969245630714]
              %  \draw [color={rgb, 255:red, 200; green, 200; blue, 200 }  ,draw opacity=1 ]   (170.07,19.45) -- (170.07,89.45) ;
%Straight Lines [id:da2488756749084644]
              %  \draw [color={rgb, 255:red, 200; green, 200; blue, 200 }  ,draw opacity=1 ]   (90.07,89.45) -- (170.07,19.45) ;
%Shape: Circle [id:dp3171203994878835]
               % \draw  [fill={rgb, 255:red, 255; green, 255; blue, 255 }  ,fill opacity=1 ] (172.57,19.52) .. controls (172.53,20.9) and (171.38,21.99) .. (170,21.95) .. controls (168.62,21.91) and (167.53,20.76) .. (167.57,19.38) .. controls (167.61,18) and (168.76,16.91) .. (170.14,16.95) .. controls (171.52,16.99) and (172.61,18.14) .. (172.57,19.52) -- cycle ;
%Shape: Circle [id:dp4119183241030231]
                \draw  [fill={rgb, 255:red, 255; green, 255; blue, 255 }  ,fill opacity=1 ] (172.57,89.52) .. controls (172.53,90.9) and (171.38,91.99) .. (170,91.95) .. controls (168.62,91.91) and (167.53,90.76) .. (167.57,89.38) .. controls (167.61,88) and (168.76,86.91) .. (170.14,86.95) .. controls (171.52,86.99) and (172.61,88.14) .. (172.57,89.52) -- cycle ;
%Shape: Circle [id:dp7171399258579877]
                \draw  [fill={rgb, 255:red, 255; green, 255; blue, 255 }  ,fill opacity=1 ] (92.57,89.52) .. controls (92.53,90.9) and (91.38,91.99) .. (90,91.95) .. controls (88.62,91.91) and (87.53,90.76) .. (87.57,89.38) .. controls (87.61,88) and (88.76,86.91) .. (90.14,86.95) .. controls (91.52,86.99) and (92.61,88.14) .. (92.57,89.52) -- cycle ;
%Shape: Circle [id:dp3803953915262809]
                \draw  [fill={rgb, 255:red, 255; green, 255; blue, 255 }  ,fill opacity=1 ] (92.57,19.52) .. controls (92.53,20.9) and (91.38,21.99) .. (90,21.95) .. controls (88.62,21.91) and (87.53,20.76) .. (87.57,19.38) .. controls (87.61,18) and (88.76,16.91) .. (90.14,16.95) .. controls (91.52,16.99) and (92.61,18.14) .. (92.57,19.52) -- cycle ;
                \draw  [blue, line width=0.3mm] (89.98,12.28) .. controls (106.48,12.78) and (90.98,133.78) .. (102.48,138.28) .. controls (113.98,142.78) and (123.98,114.28) .. (136.98,124.78) .. controls (149.98,135.28) and (46.48,219.78) .. (41.98,207.78) .. controls (37.48,195.78) and (75.98,167.28) .. (80.48,156.78) .. controls (84.98,146.28) and (73.48,11.78) .. (89.98,12.28) -- cycle ;
%Straight Lines [id:da7298510091218705]
                \draw    (210.07,199.93) -- (130.07,129.93) ;
%Straight Lines [id:da648535969864721]
                \draw    (130.07,129.93) -- (130.07,199.93) ;
%Straight Lines [id:da6346289849262575]
                \draw [color={rgb, 255:red, 0; green, 0; blue, 0 }  ,draw opacity=1 ]   (90.07,19.45) -- (170.07,89.45) ;
                \draw  [color={rgb, 255:red, 0; green, 0; blue, 0 }  ,draw opacity=0 ][fill={rgb, 255:red, 255; green, 255; blue, 255 }  ,fill opacity=1 ] (88.07,152.65) -- (103.07,152.65) -- (103.07,167.65) -- (88.07,167.65) -- cycle ;
%Shape: Rectangle [id:dp009077777663877962]
                \draw  [color={rgb, 255:red, 0; green, 0; blue, 0 }  ,draw opacity=0 ][fill={rgb, 255:red, 255; green, 255; blue, 255 }  ,fill opacity=1 ] (159.07,152.65) -- (174.07,152.65) -- (174.07,167.65) -- (159.07,167.65) -- cycle ;
%Shape: Rectangle [id:dp7045997335770793]
                \draw  [color={rgb, 255:red, 0; green, 0; blue, 0 }  ,draw opacity=0 ][fill={rgb, 255:red, 255; green, 255; blue, 255 }  ,fill opacity=1 ] (123.07,160.65) -- (140.07,160.65) -- (140.07,180.65) -- (123.07,180.65) -- cycle ;

\draw  [color=red  ,draw opacity=0 ][fill={rgb, 255:red, 255; green, 255; blue, 255 }  ,fill opacity=1 ] (102.77,32.58) -- (119.77,32.58) -- (119.77,47.58) -- (102.77,47.58) -- cycle ;
                \draw  [color={rgb, 255:red, 0; green, 0; blue, 0 }  ,draw opacity=0 ][fill={rgb, 255:red, 255; green, 255; blue, 255 }  ,fill opacity=1 ] (84,40.35) -- (90.27,40.35) -- (90.27,60.08) -- (84,60.08) -- cycle ;

% Text Node
                \draw (26,125.07) node [anchor=north west][inner sep=0.75pt]    {$Y$};
% Text Node
                \draw (25.67,195.73) node [anchor=north west][inner sep=0.75pt]    {$Z$};
% Text Node
                \draw (66,14.59) node [anchor=north west][inner sep=0.75pt]    {$U$};
% Text Node
                \draw (65.67,85.25) node [anchor=north west][inner sep=0.75pt]    {$P$};
% Text Node
                \draw (106,34.9) node [anchor=north west][inner sep=0.75pt]    {$x_{1}$};
% Text Node
                \draw (84.5,43.4) node [anchor=north west][inner sep=0.75pt]    {$\overline{x}_{1}$};
% Text Node
\draw (89.5,156.4) node [anchor=north west][inner sep=0.75pt]    {$x_{1}$};
% Text Node
                \draw (160,156.4) node [anchor=north west][inner sep=0.75pt]    {$x_{3}$};
% Text Node
                \draw (125.5,163.9) node [anchor=north west][inner sep=0.75pt]    {$\overline{x}_{2}$};
                \draw (20,15) node [anchor=north west][inner sep=0.75pt]    {$c)$};
            \end{tikzpicture}
        \end{minipage}
        }
        
        \caption{Sketch of the construction of \Cref{thm:eth-based-lowerbound-for-n-square}. The communities are \textcolor{blue}{blue} (solid, dashed and dotted). We only show edges which are contained in at least one community and only some fixed edges (\textcolor{red}{red}). $a)$~Part of the variable gadget for~$x_1$ and~$x_2$. $b)$~The variable communities for a clause~$q=x_1\lor \overline{x_2}\lor x_3$. $c)$~The assignment gadget for the first literal~$x_1$ of the clause~$q$. Here, the \textcolor{red}{red} edges are the fixed edges with one endpoint in the variable gadget and one in the clause gadget.}
        \label{fig-eth-bound-stars}
        \end{figure}
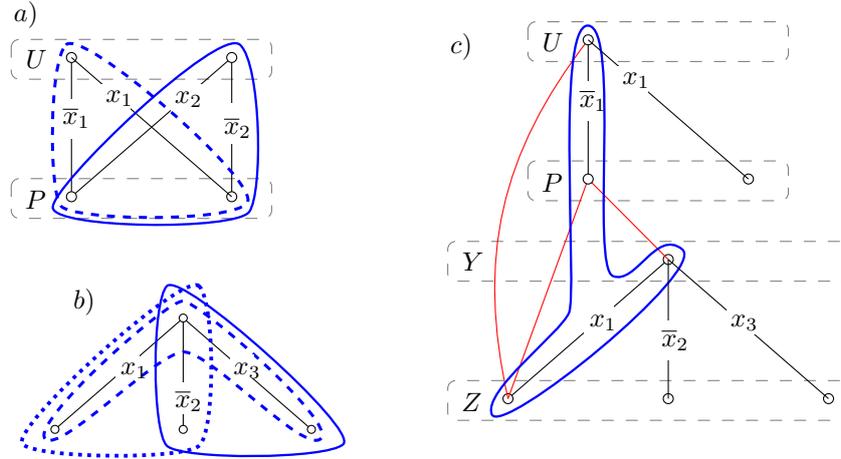

        \textbf{Variable Gadget:}
        We start by describing the construction of the variable gadget $G_X$. 
        Recall that~$G_X$ is a clique.
        The idea is to create for each variable a community~$C$ of size~$3$ with one \emph{fixed} edge. 
        The two remaining edges of~$C$ are referred to as \emph{selection} edges.
        The idea is that each solution contains exactly one selection edge of~$C$.
        One selection edge represents the positive literal, the other one represents the negative literal.
        The fixed edge of the triangle is used to model that one literal must be set to \texttt{true}.
        The selection edges are arranged compactly, to guarantee that~$|V(G_X)| \in \Oh(\sqrt{|\phi|})$.
In the following, we describe the graph~$G_X$ together with communities fulfilling the above-described properties.    
An example of a variable gadget is shown in part~$a)$ of \Cref{fig-eth-bound-stars}.

        Let~$V(G_X)=U\cup P$ where~$U \coloneqq \{u_1, \dots, u_{n_x}\}$,~$P=P_1\cup P_2$, and~$P_i \coloneqq \{p^i_1, \dots, p^i_{n_x}\}$ for~$i\in[2]$ consist of $n_x = \lceil\sqrt {|X|}\rceil$~vertices each.       
        It remains to describe the communities: 
        For each variable~$x \in X$, we add a community~$C_x \coloneqq \{u_j, p^1_{s}, p^2_{s}\}$ for~$j,s\in[\, n_x \,]$.
        This is possible since~$n_x\cdot n_x\ge |X|$.
        We refer to these communities as the \emph{variable communities}~$\mathcal{C}^X$.  
        Afterwards, we set~$\theta(x) \coloneqq \{u_j, p^1_{s}\}$ and~$\theta(\overline{x}) \coloneqq \{u_j, p^2_{s}\}$ to assign the positive and negative literal of $x$ to an edge of the variable gadget.
        Now, we fix the edges of~$G[P]=G[P_1\cup P_2]$. 
        Recall that this means that for each edge~$\{p^{i_1}_{j_1},p^{i_2}_{j_2}\}$ having both endpoints in~$P_1\cup P_2$, we add a community~$\{p^{i_1}_{j_1},p^{i_2}_{j_2}\}$.

       % Let~$G_X$ be a complete graph with vertex set~$U\cup P$ where~$U \coloneqq \{u_1, \dots, u_{n_x}\}$ and~$P \coloneqq \{p_1, \dots, p_{n_x}\}$ consist of~$n_x = 2\lceil\sqrt {|X|}\rceil$ vertices each.
       % Observe that $n_x$ is even.
       % Next, we define a partition~$\mathcal{P} \coloneqq \{\{p_{2i -1}, p_{2i }\} : i \in \{1, \dots, n_x/2\}\}$ of $P$ into sets of size $2$.
       % We now use~$\mathcal{P}$ to define communities with the desired properties.

     % First, note that since~$|U \times \mathcal{P}| = n_x \cdot \frac{n_x}{2} = 2\lceil\sqrt {|X|}\rceil \cdot \frac{2\lceil\sqrt {|X|}\rceil}{2} \ge 2|X| \geq |X|$, we  can construct in polynomial time an injection~$\eta: X \to U \times \mathcal{P}$.
      %  We use this injection to assign each variable a community whose induced subgraph is a triangle.
      %  Then, we add for each variable~$x \in X$ with~$\eta(x) = (u_j, \{p_{2i-1}, p_{2i}\})$ a community~$C_x \coloneqq \{u_j, p_{2i-1}, p_{2i}\}$.
      %  We refer to these communities as the \emph{variable communities}~$\mathcal{C}^X$.  
      %  Afterwards, we set~$\theta(x) \coloneqq \{u_j, p_{2i-1}\}$ and~$\theta(\overline{x}) \coloneqq \{u_j, p_{2i}\}$ to assign the positive and negative literal of $x$ to an edge of the variable gadget.
      %  Second, we \emph{fix} the edges~$\{v_i,v_j\}$ of the clique~$P$.
      %  More precisely, for each edge~$\{p_i,p_j\}$ having both endpoints in~$P$, we add a community~$\{p_i,p_j\}$.
        \iflong Note that edges with both endpoints in~$U$ are not contained in any community.
        We only add these edges to ensure that~$G$ is a clique.
        \fi
     
        Before describing the other gadgets, let us observe that each selection edge is indeed used in at most one variable gadget.
        
        \begin{clm}[$\star$]
            \label{claim-variable-gadget-communities-pairwise-disjoint}
            Each selection edge of~$E(G_X)$ is contained in only one subgraph induced by a variable community in~$\mathcal{C}^X$.
        \end{clm}
        \iflong
        \begin{claimproof}
            Each variable community~$C^x$ consists of one vertex~$u \in U$ and two vertices~$p_i, p_j \in P$ where~$\{u,p_i\}$ and~$\{u,p_j\}$ are the selection edges.
            Suppose there are two distinct variable communities~$C_{1}$ and~$C_{2}$ such that they have at least one selection edge in common.
            This implies that~$C_{1}$ and $C_{2}$ contain the same vertex of~$U$ and have at least one vertex of~$P$ in common.
            By construction, in each variable community~$C_i$ the two vertices of~$P_1\cup P_2$ are either identical or there is no overlap.
            Thus,~$C_1$ and~$C_2$ are identical, a contradiction.
        \end{claimproof}\fi

        \textbf{Clause Gadget:}
        We continue by describing the construction of the clause gadget~$G_\Gamma$.
The idea is that each clause is represented by four vertices of~$V(G_\Gamma)$ in which a triangle is \emph{fixed}.
All three remaining edges of this size-4 clique are referred to as \emph{free}.
Note that these free edges form a star with three~leaves.
Each free edge represents one literal of the clause.
For each pair containing two of these three edges, we then create a community containing the three endpoints of these two edges.
        As in the vertex gadget, these induced subgraphs are arranged compactly, to achieve a clause gadget with~$|V(G_\Gamma)| \in \Oh(\sqrt{|\Gamma|})$.
        In the following, we describe the graph~$G_\Gamma$ together with communities fulfilling these properties. 
        Part~$b)$ of \Cref{fig-eth-bound-stars} shows an example for the representation of a clause.
%       Let~$\nu: \{(q, \ell_1), (q, \ell_2), (q, \ell_3) : q = \{\ell_1, \ell_2, \ell_3\} \in \Gamma\} \to E_\Gamma$ be the mapping of the literal occurrences in clauses to the edges of the clause gadget.
%        This mapping will be populated throughout the following construction.

        Let~$V(G_\Gamma)=Y\cup Z$ where~$Y = \{y_1, \dots, y_{n_c}\}$,~$Z=Z_1\cup Z_2\cup Z_3$, and~$Z_i = \{z^1_1, \dots, z^i_{n_c}\}$ for~$i\in[3]$ consist of $n_c = \lceil\sqrt {|\Gamma|}\rceil$~vertices each.
        In the following, we assign each clause to a clique of~$G_\Gamma$ having vertex set~$y_j,z^1_s,z^2_s,z^3_s$ for~$j,s\in[n_c]$.
        This is possible since~$n_c\cdot n_c\ge |\Gamma|$.
        In this clique, we \emph{fix} the triangle having its endpoints in~$Z_1\cup Z_2\cup Z_3$. 
        Formally, for each clause~$q= \{q_1, q_2, q_3\} \in \Gamma$ we add three communities~$C^1_{q} = \{y_j, z^2_{s}, z^3_{s}\}$, $C^2_{q} = \{y_j, z^1_{s}, z^3_{s}\}$ and $C^3_{q} = \{y_j, z^1_{s}, z^2_{s}\}$.
        We refer to these communities as the \emph{clause} communities~$\mathcal{C}^\Gamma$.
        Afterwards, we set $\nu(q, q_1) \coloneqq \{y_j, z^1_{s}\}$, $\nu(q, q_2) \coloneqq \{y_j, z^2_{s}\}$, and $\nu(q, q_3) \coloneqq \{y_j, z^3_{s}\}$ to assign each literal in clause $q$ to an edge of the clause gadget.
        These edges are referred to as \emph{free}.
        Second, we fix the edges of the clique~$Z_1\cup Z_2\cup Z_3$.

        \iflong Note that, similar to the variable gadget, edges with both endpoints in~$Y$ are not contained in any community.
        Again, we only add these edges to ensure that~$G$ is a clique. \fi
        
Observe that the sets of free edges corresponding to two distinct clauses are disjoint:

               \begin{clm}[$\star$]
            \label{claim-clause-gadget-communities-pairwise-disjoint}
            Each free edge of~$E(G_\Gamma)$ is contained in exactly one subgraph induced by a clause community in~$\mathcal{C}^\Gamma$.
        \end{clm}
        \iflong
        \begin{claimproof}
            Each clause community~$C^q$ consists of one vertex~$y\in Y$ and one vertex~$z^i_s\in Z_i$.
            In~$C^q$ the free edges form a star with~$y$ as there center and the vertices~$z^i_s$ as leaves.
            Suppose there are two distinct clause communities~$C_1$ and~$C_2$ such that they have at least one free edge in common.
            This implies that~$C_1$ and~$C_2$ contain the same vertex of~$Y$ and have at least one vertex~$z^i_s$ of~$Z_i$ in common.
            By construction, in each clause community~$C_i$ the three vertices of~$Z_1\cup Z_2\cup Z_3$ are either identical or there is no overlap.
            Thus,~$C_1$ and~$C_2$ are identical, a contradiction.
        \end{claimproof} 
        \fi

        \textbf{Connecting the Gadgets:}
        We complete the construction by describing how the variable and clause gadget are connected, using new \emph{assignment} communities.
		The idea is to add a new community containing the endpoints of a free edge describing a literal in a clause together with the endpoints of the selection edge describing the same opposite literal in the variable gadget.
        These communities model occurrences of variables in the clauses.
        Roughly speaking, these communities are satisfied if the selection edge of the variable gadget or the free edge of the clause gadget is part of the solution.      
        To enforce this, we fix further edges of~$G$.
        An example of an assignment community for one literal is shown in part~$c)$ of \Cref{fig-eth-bound-stars}.
        
        We create for each clause~$q = \{q_1, q_2, q_3\} \in \Gamma$ three assignment communities~$C^{q_1}_{q} = \nu(q, q_1) \cup \theta(\overline{q_{1}})$, $C^{q_2}_{q} = \nu(q, q_2) \cup \theta(\overline{q_{2}})$, and~$C^{q_3}_{q} = \nu(q, q_3) \cup \theta(\overline{q_{3}})$.
        We denote the assignment communities with~$\mathcal{C}^X_\Gamma$.
To enforce that each solution contains the selection edge or the free edge of each assignment community, we fix all edges between the vertex sets~$U$ and~$Z$, between the vertex sets~$P$ and~$Y$, and between the vertex sets~$P$ and~$Z$.
%We denote all these edges as \emph{fixed}.
        %Also, for each edge~$e$ added this way, we add a community consisting of the endpoints of~$e$.
        %We also call these communities \emph{fixing}.
        
        \iflong Note that edges between~$U$ and~$Y$ are not part of any community and only added to ensure that~$G$ is a clique.\fi
        %These edges are \emph{not} fixed.
        %We only add these edges to make~$G$ a clique.

        Finally, we set~$\ell\coloneqq |X| + 2\cdot |\Gamma|+ \binom{|P|}{2}+\binom{|Z|}{2}+|U|\cdot |Z|+|P|\cdot |Y|+|P|\cdot |Z|$.
        \iflong  Let~$I$ denote the constructed instance of \UStars.
        %Observe that the vertex set~$V$ of~$I$ is~$V_X\cup V_\Gamma$.
        Clearly,~$|V(G)|\in\Oh(\sqrt{|\phi|})$ and the maximum community size is~$4$.
		Instance~$I$ has at most~$\Oh(\sqrt{|\phi|}\cdot\sqrt{|\phi|})=\Oh(|\phi|)$~communities of size~$2$, $|X|+3\cdot|\Gamma|$~communities of size~$3$, and~$3\cdot|\Gamma|$~communities of size~$4$. 
		Thus,~$I$ has~$\Oh(|\phi|)$~communities.       
        \fi

        \textbf{Correctness:}
\iflong\else The correctness is based on the observations that each solution for~$I$ contains~$a)$ all fixed edges, $b)$ exactly $|X|$~selection edges, and~$c)$ exactly $2|\Gamma|$~free edges.  
Fact~$b)$ ensures that this models an assignment of the variables of~$X$ and fact~$c)$ ensures that each clause is satisfied by at least one literal of this assignment.
The detailed correctness proof is deferred to the full version of this article. \fi  
        \iflong
        We show that the formula~$\phi$ is satisfiable if and only if~$I$ is a yes-instance of \UStars.
        Before we prove this statement, we make three observations about the solution.
        
        Since fixed edges are contained in a community of size~$2$, we have the following
        
        \begin{observation}
            \label{claim-parameter-l-observation-1}
            Each solution of~$I$ contains all fixed edges.
        \end{observation}
        %\iflong
        %\begin{claimproof}
        %Observe that each fixing community contains exactly two~vertices, that is, the two endpoints of a fixed edge, and that each two fixing communities have at most one vertex in common.
        %Thus, any solution for~$I$ contains each fixed edge.
        %\todomi{warum wird hier über schnitt maximal einsgeredet??? jede loesung muss die einzigartige kante jeder community der groesse 2 enthalten}
        %\end{claimproof}\fi
        
        Observe that~$I$ consists of precisely $\binom{|P|}{2}+\binom{|Z|}{2}+|U|\cdot |Z|+|P|\cdot |Y|+|P|\cdot |Z| = \ell - (|X| + 2\cdot |\Gamma|)$~fixed edges.
        Hence, each solution can contain at most~$|X|+2\cdot |\Gamma|$ selection and free edges.
        Thus, we obtain the following.

        \begin{clm}[$\star$]
            \label{claim-parameter-l-observation-2}
            Each solution of~$I$ contains at least $|X|$~selection edges.
        \end{clm}
        \iflong
        \begin{claimproof}
        Observe that at least one selection edge of each variable community~$C$ in the variable gadget is contained in the solution to satisfy~$C$.
            By \Cref{claim-variable-gadget-communities-pairwise-disjoint}, these are at least~$|X|$ edges.
        \end{claimproof}\fi
        
        \begin{clm}[$\star$]
            \label{claim-parameter-l-observation-3}
            Each solution of~$I$ contains at least $2\cdot |\Gamma|$~free edges.
        \end{clm}
        \iflong
        \begin{claimproof}
        Observe that at least two free edges of each triple of communities $C_q^1, C_q^2, C_q^3$ of the clause gadget for each clause~$q$ are contained in the solution to fulfill the requirements of these communities.
            By \Cref{claim-clause-gadget-communities-pairwise-disjoint}, and since no free edge is contained in any community of size~$2$, these are at least $2\cdot |\Gamma|$~edges.
        %    Since each variable community and each clause community do not have any vertex in common, this implies that~$I$ is a no-instance of $4$-\textsc{Stars NWS}.
        \end{claimproof}\fi
        
        From \Cref{claim-parameter-l-observation-1,claim-parameter-l-observation-2,claim-parameter-l-observation-3} we conclude that each solution consists of at least $\ell$~edges.
        \iflong\else The remaining part of the correctness proof is deferred to the full version of this article.\fi
        \iflong 
        Now, we show the correctness.

        $(\Rightarrow)$ Let~$A: X \to \{0,1\}$ be an assignment satisfying~$\phi$.
        We describe how to obtain a solution~$G'=(V(G),E')$ with $|E'| = \ell$ from~$A$.
        First, by \cref{claim-parameter-l-observation-1}, the solution contains all~$\ell-|X|-2\cdot |\Gamma|$ fixed edges.
        Second, for each variable, we chose the selection edge of the variable gadget representing the literal that is not satisfied by~$A$.
        For a formal definition, recall that for a variable~$x\in X$,~$\theta(x)$ is the selection edge representing~$x$ and~$\theta(\overline{x})$ is the selection edge representing~$\overline{x}$.
         We set~$E'_X \coloneqq \{\theta(x) : x \in X, A(x) = 0\} \cup \{\theta(\overline{x}) : x \in X, A(x) = 1\}$.
        Note that~$E'_X$ and the fixed edges fulfill the spanning star property of the variable communities.
        Third, we construct the set~$E'_\Gamma$ of free edges which are contained in the solution.
        For each clause $q = \{q_1, q_2, q_3\}\ \in\Gamma$, there is at least one variable~$x$ such that with the assignment~$A$ the literal, say~$q_3$, corresponding to~$x$ satisfies~$q$.
        We add the edges~$\nu(q, q_1)$ and~$\nu(q, q_2)$ to~$E'_\Gamma$.
        Recall that in the construction of the clause gadget of~$q$, we created three clause communities $C^1_{q} = \nu(q, q_1) \cup \nu(q, q_2)$, $C^2_{q} = \nu(q, q_1) \cup \nu(q,q_3)$, and $C^3_{q} = \nu(q, q_2) \cup \nu(q, q_3)$.
        Observe that the endpoint of~$\nu(q, q_1)$ in~$Z$ or the endpoint of~$\nu(q, q_2)$ in~$Z$ is a center for each of these three clause communities.
        Thus, the spanning star property of all clause communities is fulfilled.
        
        Observe that the graph~$G'$ has exactly $\ell$~edges.
        To verify that~$I$ is a yes-instance of \UStars, it remains to show that~$G'$ also fulfills the spanning star property of the assignment communities.
For this, consider the three assignment communities~$C^{q_1}_{q} = \nu(q, q_1) \cup \theta(\overline{q_1})$, $C^{q_2}_{q} = \nu(q, q_2) \cup \theta(\overline{q_2})$, and~$C^{q_3}_{q} = \nu(q, q_3) \cup \theta(\overline{q_3})$ corresponding to clause~$q$.
Recall that we assume that~$q$ is satisfied by literal~$q_3$ and that we defined~$E'$ accordingly.
Also, recall that all fixed edges are contained in~$E'$.
Since~$\nu(q, q_1)$ and~$\nu(q, q_2)$ are edges of~$E'$, we conclude that the unique vertex in~$Z$ of~$C_q^{q_1}$ and~$C_q^{q_2}$ is the center of a star for these assignment communities.
Furthermore, since~$\theta(\overline{q_3})\in E'$, we conclude that the unique vertex in~$P$ of~$C_q^{q_3}$ is the center of the assignment community~$C_q^{q_3}$.
Hence, the requirements of all communities are fulfilled by~$G'$.

        $(\Leftarrow)$
        Let~$G'=(V(G),E')$ be a solution of~$I$ with $\ell$~edges.
From \cref{claim-parameter-l-observation-1,claim-parameter-l-observation-2,claim-parameter-l-observation-3}, we conclude that~$E'$ contains exactly $|X|$~selection edges, denoted by~$E'_X$, and exactly $2\cdot |\Gamma|$~free edges.
Furthermore, note that for each variable community~$C_x$ corresponding to some variable~$x\in X$, the solution contains at least one of the selection edges~$\theta(x)$ or~$\theta(\overline{x})$.
Recall that, according to \cref{claim-variable-gadget-communities-pairwise-disjoint}, the selection edges of each two variable communities are disjoint.
Since~$G'$ contains exactly $|X|$~selection edges, we conclude that~$G'$ contains exactly one selection edge of each variable~$x\in X$.
This, allows us to properly define an assignment~$A$ by setting~$A(x)=0$ if~$\theta(x) \in E'_X$ and~$A(x)=1$ if~$\theta(\overline{x}) \in E'_X $. 

Similarly, according to \cref{claim-clause-gadget-communities-pairwise-disjoint}, the free edges of each two clause communities are disjoint.
Since~$G'$ contains exactly $2\cdot |\Gamma|$~free edges, we conclude that~$G'$ contains exactly two free edges per clause~$q\in\Gamma$ (according to \cref{claim-parameter-l-observation-3}), that is, we have $|E' \cap \{\nu(q, q_1), \nu(q, q_2), \nu(q, q_3)\}| = 2$ for each clause~$q\in\Gamma$.
        
        Since~$G'$ contains exactly~$\ell$ edges and~$|X|$ of them are selection edges,~$2\cdot |\Gamma|$ of them are free edges, and~$\ell-|X|-2\cdot |\Gamma|$ of them are fixed edges, we conclude that no edge which is not a selection edge, a free edge or a fixed edge can be part of~$G'$.
        Recall that we created for each clause~$q = \{q_1, q_2, q_3\}\in\Gamma$ three assignment communities~$C^{q_1}_{q}$, $C^{q_2}_{q}$, and~$C^{q_3}_{q}$ connecting the variable and the clause gadget.
        Observe that community~$C^{q_i}_{q}$ contains three fixed edges, the selection edge~$\nu(q, q_i)$, the free edge~$\theta(\overline{q_i})$, and one edge which has none of these $3$~types.
        By the above argumentation on the tightness of the budget~$\ell$ and the fact that~$G'$ is a solution, we conclude that in the assignment community~$C^{q_i}_{q}$ the edge~$\nu(q, q_i)$ or the edge~$\theta(\overline{q_i})$ is contained in~$E'$ for each~$i\in\{1,2,3\}$ to fulfill the spanning star property of these assignment communities.
        Now recall that we have~$|E' \cap \{\nu(q, q_1), \nu(q, q_2), \nu(q, q_3)\}| = 2$ for each clause~$q$.
        Thus, at least one of the edges~$\theta(\overline{q_1})$, $\theta(\overline{q_2})$, or~$\theta(\overline{q_3})$ is contained in~$E'$.
        Without loss of generality, we assume that~$\theta(\overline{q_3})\in E'$.
		This implies that the assignment of the variable~$x$ corresponding to literal~$q_3$ satisfies clause~$q$.        
        Thus, if~$q_3 = x$, we set~$A(x)=1$, and otherwise if~$q_3 = \overline{x}$, we set~$A(x)=0$.
        This applies to all clauses, and thus~$\phi$ is satisfied by~$A$.\fi\fi
    \end{proof}

        \iflong
        \subsection{Adaptions of \Cref{thm:eth-based-lowerbound-for-n-square}}     
\fi        
\iflong       
In the following, we present two further results: First, the same ETH-based bound for \UCon (\Cref{cor-eth-based-lowerbound-for-n-square-connected}). Second, we show that \UStars remains NP-hard on graphs with constant maximum degree (\Cref{cor-stars-np-h-const-delta}). 
The latter result can be shown by reducing from the NP-hard \textsc{$(3,$B$2)$-SAT}~\cite{berman03} and by using a similar, but uncompressed, construction to the one introduced in the proof of \Cref{thm:eth-based-lowerbound-for-n-square}.
\else
  By adapting the construction, we obtain the following two further results.
\fi
      \iflong  
\subparagraph*{Adaption for \UCon.}

        Now, we modify the construction to replace the star requirement for each community with the requirement of being connected.
        The variable gadget and the clause gadget is constructed exactly as in the construction in \Cref{thm:eth-based-lowerbound-for-n-square} for \UStars.
        The assignment gadget, however, is constructed differently: instead of fixing the edges between the vertex sets~$U\cup P$ and~$Z$, and between~$P$ and~$Y$, we now only fix the edges between~$U$ and~$Z$, and between~$P$ and~$Y$.
        Consequently, all edges with one endpoint in~$P$ and one in~$Z$, or one endpoint in~$U$ and one in~$Y$ are not contained in any community.
        Then, we set~$\ell\coloneqq |X| + 2\cdot |\Gamma|+ \binom{|P|}{2}+\binom{|Z|}{2}+|U|\cdot |Z|+|P|\cdot |Y|$, that is, compared to \UStars, the budget~$\ell$ is decreased by~$|P|\cdot |Z|$.

Now the proof is similar to the proof of \Cref{thm:eth-based-lowerbound-for-n-square}: 
First, the size~$2$ communities still imply that all fixed edges are part of any solution.
Second, since for communities of size 3, the star property and the connectivity property coincide, the argumentation for \UCon is analogously to that of \UStars.
Third, since also in this construction the budget is tight, the solution contains all fixed edges, exactly $|X|$~selection edges, and exactly $2\cdot|\Gamma|$~free edges.
Now, since the edges of these $3$~types in each assignment community form a cycle of length~$4$ ($2$ fixed edges, $1$~selection edge, and $1$~free edge), we conclude that in each assignment community, the selection edge or the free edge (or both) are contained in the solution.
From now on, the argumentation is analogous to the argumentation in \Cref{thm:eth-based-lowerbound-for-n-square} for \UStars.
Thus, we obtain the following.
\fi
\begin{corollary}
        \label{cor-eth-based-lowerbound-for-n-square-connected}
        If the ETH is true, then \UCon cannot be solved in $2^{o(n^2)} \cdot \poly(n+c)$~time, even if~$G$ is a clique and each community has size at most~$4$.
    \end{corollary}

\iflong
    
\subparagraph*{Adaption for \UStars with Constant Maximum Degree.}

Next, we show that \UStars remains NP-hard on graphs with constant maximum degree.
We achieve this by three adjustments:
first, we do not compress the variable and clause gadget in the construction of \Cref{thm:eth-based-lowerbound-for-n-square}, second we reduce from the NP-hard \textsc{$(3,$B$2)$-SAT}~\cite{berman03} in which each clause consists of three literals and each literal (positive and negative) appears exactly twice, and third, we do not add edges which are not contained in any community.

More precisely, for each variable, we add a triangle with one fixed edge; again this is achieved by adding a community of size~$2$ containing this edge.
Afterwards, for each clause we add a size~$4$-clique with a fixed triangle.
Then, we add the assignment communities and the corresponding fixed edges as described in the construction of \Cref{thm:eth-based-lowerbound-for-n-square}.
Note that each vertex is either contained in a variable triangle or in a size~$4$-clique corresponding to a clause.
Since there are only edges between the variable gadgets and the clause gadgets, each vertex in the variable gadget has exactly two~neighbors in the variable gadget and each vertex in the clause gadget has exactly three neighbors in the clause gadget.
It remains to show that each vertex in the variable gadget has only a constant number of neighbors in the clause gadget and vice versa.

First, we consider a vertex~$v$ in the variable gadget.
Since each literal appears exactly twice, there are exactly four~assignment communities having vertices of the triangle in the variable gadget corresponding to~$v$.
With this observation at hand, one can show that~$v$ has at most four~neighbors in the clause gadget.
Second, we consider a vertex~$v$ in the clause gadget.
Since each clause consists of three literals, there are exactly three~assignment communities having vertices of the size~$4$-clique in the clause gadget corresponding to~$v$.
With this observation at hand, one can show that~$v$ has at most three~neighbors in the clause gadget.
Hence, we obtain the following.
\fi

    \begin{corollary}
    \label{cor-stars-np-h-const-delta}
        \UStars remains NP-hard and, assuming the ETH, cannot be solved in $2^{o(n+m+c)} \cdot \poly(n+c)$~time on graphs with maximum degree six and community size at most~$4$.
    \end{corollary}

    \section{Parameterization by the Feedback Edge Number of a Solution}
\label{sec:fes}

   The parameter~$\ell$, the number of edges in the solution is in most cases not independent from the size of the input instance of \textsc{Stars NWS} or \textsc{Connectivity NWS}:
    %Since in solutions of \textsc{Stars NWS} and \textsc{Connectivity NWS}, each subgraph induced by a community is connected, a graph $G$ with more vertices leads to a greater parameter $\ell$.
    % Each subgraph induced by a community~$C$ has at least $|C| - 1$~edges and~$G$ has order~$n$. 
    if the hypergraph~$(V,\mc)$ is connected, a solution~$G'$ has at least $n-1$~edges.
    In other words, $n-1$ is a lower bound for~$\ell$ in this case. 
    In this section, we study \textsc{Stars NWS} and \textsc{Connectivity NWS} parameterized above this lower bound.
    Formally, the parameter~$t$ is defined as the size of a minimum feedback edge set of the solution of an instance of \textsc{Stars NWS} or \textsc{Connectivity NWS}.
    Thus, the parameter~$t$ measures how close the solution is to a forest.
    Formally, the definition is~$t \coloneqq \ell-n+x$ where~$x$ denotes the number of connected components of~$G'$.
Recall that~$t$ can be computed in polynomial time (see \Cref{sec:basic-results}.).

\subsection{An XP-Algorithm for \Stars}

\newcommand{\univ}{\mathrm{univ}}
In this subsection, we show that~\Stars parameterized by~$t$ admits an~XP-algorithm with the following running time.
\begin{theorem}\label{xp t}
\Stars can be solved in $m^{4t} \cdot \poly(|I|)$~time.
\end{theorem}

    Korach and Stern~\cite{korach2008complete} asked whether \Stars{} is polynomial-time solvable if~$t=0$.
    \Cref{xp t} answers this question positively.
    
Our XP-algorithm, exploits the fact that there are two different kinds of cycles in~$G'$:
First, there are \emph{global} cycles. These are the cycles in the solutions that are directly caused by cycles in the input hypergraph. No solution may avoid these cycles.
Second, there are \emph{local} cycles. These are cycles which are entirely contained in the subgraph induced by two communities.
Since in each solution, each community contains a spanning star, local cycles can only have length~$3$ or~$4$. 
This allows us to bound the number of possible local cycles and thus to consider all possibilities for the local cycles in XP-time with respect to~$t$.
Then, the crux of our algorithm is that after all local cycles have been fixed, all remaining cycles added by our algorithm have to be global and are thus unavoidable. 
Using this fact, we show that in polynomial time we can compute an optimal solution with feedback edge number at most~$t$ that extends a fixed set of local cycles without introducing any further local cycles.
To do this, for each community~$C$, we store a set of \emph{potential centers}, that is, vertices of~$C$ that may be the center of a spanning star of~$C$ in any solution that does not produce new local cycles.
We define several \emph{operations} that restrict the potential centers of each community.
We show that after all operations have been applied exhaustively, one can greedily pick the best remaining center for each community.

\subparagraph{Algorithm-specific notation.}
Next, we present the formal definition of local cycles; an example is shown in \Cref{figure cycles}.
For a spanning subgraph~$H$ of~$G$ and a community~$C\in \mc$, let~$\univ_H(C)$ denote the vertices of~$C$ that are \emph{universal} for~$C$ in~$H$.
Recall that a vertex~$u$ is universal for some vertex set~$S$ in a graph~$F$ if~$\{u,w\}\in E(F)$ for each~$w\in S$.
Note that~$\univ_H(C) \subseteq \univ_G(C)$.
In the following, we assume that for each community~$C\in\mc$, $\univ_G(C)\neq \emptyset$, as otherwise, there is no solution for the instance~$I$ of \Stars, and~$I$ is a trivial no-instance.

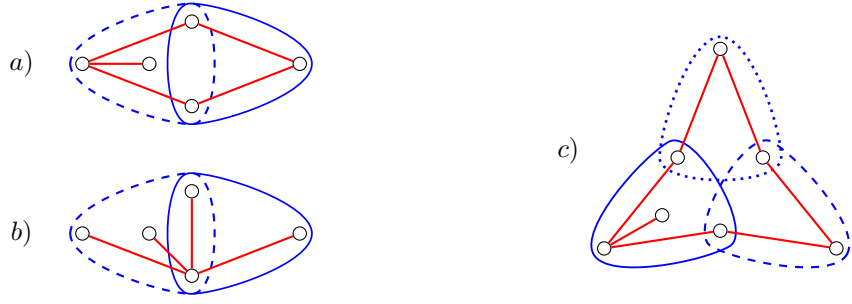
\begin{figure}[t]
\begin{center}
\begin{tikzpicture}[scale = .8]

\tikzstyle{knoten}=[circle,fill=white,draw=black,minimum size=5pt,inner sep=0pt]
\node[knoten] (vA) at (1.2,6) {};

\node (label) at ($(vA) - (1,0)$) {{$b)$}};

\node[knoten] (vA2) at (2.3,6) {};

            \draw  [dashed, line width=.8,blue] (1,6) .. controls (1,5.5) and (2.5,5) .. (3,5) .. controls (3.5,5) and (3.5,7) .. (3,7) .. controls (2.5,7) and (1,6.5) .. (1,6) -- cycle ;

\begin{scope}[xshift=170]
            \draw  [line width=0.65,blue] (-1,6) .. controls (-1,5.5) and (-2.5,5) .. (-3,5) .. controls (-3.5,5) and (-3.5,7) .. (-3,7) .. controls (-2.5,7) and (-1,6.5) .. (-1,6) -- cycle ;

\node[knoten] (vB) at (-1.2,6) {};
\end{scope}

\node[knoten] (vC) at (3,5.3) {};
\node[knoten] (vD) at (3,6.7) {};

\draw[red,thick] (vC) to (vA2);
\draw[red,thick] (vC) to (vA);
\draw[red,thick] (vC) to (vB);
\draw[red,thick] (vC) to (vD);

\begin{scope}[yshift=80]

\node[knoten] (vA) at (1.2,6) {};

\node (label) at ($(vA) - (1,0)$) {{$a)$}};

\node[knoten] (vA2) at (2.3,6) {};

            \draw  [dashed, line width=.8,blue] (1,6) .. controls (1,5.5) and (2.5,5) .. (3,5) .. controls (3.5,5) and (3.5,7) .. (3,7) .. controls (2.5,7) and (1,6.5) .. (1,6) -- cycle ;

\begin{scope}[xshift=170]
            \draw  [line width=0.65,blue] (-1,6) .. controls (-1,5.5) and (-2.5,5) .. (-3,5) .. controls (-3.5,5) and (-3.5,7) .. (-3,7) .. controls (-2.5,7) and (-1,6.5) .. (-1,6) -- cycle ;

\node[knoten] (vB) at (-1.2,6) {};
\end{scope}

\node[knoten] (vC) at (3,5.3) {};
\node[knoten] (vD) at (3,6.7) {};

\draw[red,thick] (vA) to (vA2);
\draw[red,thick] (vA) to (vC);
\draw[red,thick] (vA) to (vD);
\draw[red,thick] (vB) to (vC);
\draw[red,thick] (vB) to (vD);

\end{scope}

\begin{scope}[yshift=30,xshift=280,rotate around={30:(3,3)}]

\node[knoten] (vA) at (1.2,6) {};

\node[knoten] (vA2) at (2.3,6) {};

            \draw  [line width=.7,blue] (1,6) .. controls (1,5.5) and (2.5,5) .. (3,5) .. controls (3.5,5) and (3.5,7) .. (3,7) .. controls (2.5,7) and (1,6.5) .. (1,6) -- cycle ;

%\draw (96.89,38.00) node [anchor=north west][inner sep=0.75pt]    {$u_{2}$};
\begin{scope}[xshift=91.2,yshift=118.1]
\begin{scope}[rotate around={120:(3,3)}]
\begin{scope}[xshift=0]
            \draw  [dashed, line width=.8,blue] (1,6) .. controls (1,5.5) and (2.5,5) .. (3,5) .. controls (3.5,5) and (3.5,7) .. (3,7) .. controls (2.5,7) and (1,6.5) .. (1,6) -- cycle ;

\node[knoten] (vC2) at (3,5.3) {};
\node[knoten] (vB1) at (1.2,6) {};
\end{scope}
\end{scope}
\end{scope}

%\draw (96.89,38.00) node [anchor=north west][inner sep=0.75pt]    {$u_{2}$};
\begin{scope}[xshift=-56.6,yshift=138]
\begin{scope}[rotate around={240:(3,3)}]
\begin{scope}[xshift=0]
            \draw  [dotted, line width=1,blue] (1,6) .. controls (1,5.5) and (2.5,5) .. (3,5) .. controls (3.5,5) and (3.5,7) .. (3,7) .. controls (2.5,7) and (1,6.5) .. (1,6) -- cycle ;

\node[knoten] (vB2) at (1.2,6) {};
\end{scope}
\end{scope}
\end{scope}

\node[knoten] (vC) at (3,5.3) {};
\node[knoten] (vD) at (3,6.7) {};

\node (label) at ($(vD) - (1.5,-1)$) {{$c)$}};

\draw[red,thick] (vA) to (vA2);
\draw[red,thick] (vA) to (vC);
\draw[red,thick] (vA) to (vD);
\draw[red,thick] (vB1) to (vC);
\draw[red,thick] (vB2) to (vD);
\draw[red,thick] (vB1) to (vC2);
\draw[red,thick] (vB2) to (vC2);

\end{scope}
\end{tikzpicture}
\end{center}
\caption{Examples for solutions with and without local cycles. 
Red edges indicate the edges of the solution. 
Part a)~shows an example, where both communities induce a local cycle. 
Part b)~shows an example, where the two communities do not induce local cycle. 
Finally, part c)~shows an example, where the solution contains a cycle but no two communities induce a local cycle.}
\label{figure cycles}
\end{figure}

For a solution~$G'$, we say that two distinct communities~$C_1$ and~$C_2$~\emph{induce a local cycle} if for each~$i\in \{1,2\}$, there is a vertex~$c_i\in \univ_{G'}(C_i)$ such that the graph~$S_1\cup S_2$ contains a cycle.
Here, for each~$i\in \{1,2\}$, $S_i$ is the spanning star of~$C_i$ with center~$c_i$ and~$S_1\cup S_2$ is the union of both these stars defined by$S_1\cup S_2 \coloneqq(C_1\cup C_2, \{\{c_i,w_i\}\colon w_i\in C_i\setminus \{w_i\}, i\in \{1,2\}\})$.  
Moreover, we say that each cycle of~$S_1\cup S_2$ is a~\emph{local cycle} in~$G'$.
Note that each local cycle has length at most four, and if~$C_1$ and~$C_2$ induce a local cycle, then~$|C_1\cap C_2|\geq 2$.

As described above, the first step of the algorithm behind~\Cref{xp t} is to test each possibility for the local cycles of the solution.
For a fixed guess, we let~$E^*$ denote the set of all edges contained in at least one local cycle and in the following we refer to them as \emph{local edges}.
Moreover, we call a minimum solution~$G'$~\emph{fitting for~$E^*$} if each local cycle of~$G'$ uses only edges of~$E^*$ and each edge of~$E^*$ is contained in~$G'$.
%To present the corresponding algorithms for~\Stars, we first establish the following.
%Intuitively, if there is a solution where~$E^*$ are the edges of all local cycles, then that solution is fitting for~$E^*$ and vice versa.
Hence, to determine whether the choice of local edges~$E^*$ can lead to a solution, we only have to check, whether there is a fitting solution for~$E^*$.
In the following, we show that this can be done in polynomial time.

\begin{theorem}\label{thm:fitting}
Let~$I=(G=(V,E),\mathcal{C},\omega,\ell,b)$ be an instance of \Stars, and let~$E^*\subseteq E$.
In polynomial time, we can  
\begin{itemize}
\item find a solution~$G'=(V,E')$ for~$I$ with~$E^*\subseteq E'$, $|E'| \leq \ell$, and~$\omega(E') \leq b$ or
\item correctly output that there is no minimum solution that is fitting for~$E^*$.
%correctly output that there is no minimum solution~$G'$ for~$I$ with at most~$\ell$ edges and total weight at most~$b$, such that each local cycle of~$G'$ uses only edges of~$E^*$ and each edge of~$E^*$ is contained in~$G'$.
\end{itemize}
\end{theorem}
\newcommand{\locpot}[1][E^*]{\mathrm{fit}_{#1}}
\newcommand{\LO}{L_{E^*}}

Based on the definition of fitting solutions, we define for each community~$C\in \mc$ a set~$\locpot(C)$ of possible centers.
We initialize~$\locpot(C)\coloneqq \univ_G(C)$ for each community~$C\in\mc$.
The goal is to reduce these sets of possible centers of each community as much as possible while preserving the following property, which trivially holds for the initial~$\locpot(C)$ for each community~$C\in\mc$.

\iflonglong
\todomi{ist `property` ein ungünstiger name, weil `stars`und `connected`auch properties sind? lieber 'attribute' oder 'feature'?}\fi
\begin{property}\label{prop}
For each minimum solution~$G'$ which is fitting for~$E^*$ and each community~$C\in\mc$, we have~$\univ_{G'}(C) \subseteq \locpot(C)$.
\end{property}
Note that if~\Cref{prop} is fulfilled and if~$\locpot(C) = \emptyset$ for some~$C\in\mc$, then we can correctly output that there is no fitting solution for~$E^*$.

In the following, we define several operations that for some communities~$C\in \mc$ remove vertices from~$\locpot(C)$ which---when taken as a center vertex for~$C$---would introduce new local cycles, violating the properties of a fitting solution.
We show that all of these operations preserve~\Cref{prop} and that after all these operations are applied exhaustively, the task of~\Cref{thm:fitting} can be performed greedily based on~$\locpot$. 
\iflong Examples for each of our operations are shown in \Cref{figure operations}.\fi

In the following, we say that a vertex~$v\in V$ is~\emph{locally universal} for a vertex set~$A\subseteq V$, if for each vertex~$w\in A\setminus \{v\}$, the vertex pair~$\{v,w\}$ is a local edge.
Based on this definition, we are now able to present the first operation. 

\begin{operation}\label{op endpoints of local}
Let~$C\in\mc$ be a community and let~$\{y,z\}\subseteq C$ be a local edge.
Remove each vertex~$v$ from~$\locpot(C)$ which is not locally universal for~$\{y,z\}$.
\end{operation}

The following lemma shows that~\Cref{op endpoints of local} preserves~\Cref{prop}.
\begin{lemma}[$\star$]
Let~$G'$ be a minimum solution for~$I$, let~$C$ be a community of~$\mc$ and let~$x\in \univ_{G'}(C)$ such that~$x$ is not locally universal for some local edge~$\{y,z\}\subseteq C$.
Then, $G'$ is not fitting for~$E^*$.
\end{lemma}
\iflong
\begin{proof}
If~$\{y,z\}$ is not an edge of~$G'$, then~$G'$ is not fitting for~$E^*$, since~$G'$ does not contain all edges of~$E^*$.
Hence, in the following we assume that~$\{y,z\}$ is an edge of~$G'$.
If the graph~$G''$ obtained by removing the edge~$\{y,z\}$ from~$G'$ is a solution for~$I$, then~$G'$ is not a minimum solution.
Thus, we assume that~$G''$ is not a solution for~$I$. 
Consequently, there is some community~$D\in \mc$ such that~$\univ_{G''}(D)=\emptyset$.
Since~$G'$ is a solution for~$I$, $\univ_{G'}(D)\neq\emptyset$.
Moreover, since~$G''$ is obtained from~$G'$ by removing the edge~$\{y,z\}$, we have~$y \in \univ_{G'}(D)$ or~$z \in \univ_{G'}(D)$. 
Hence, $C$ and~$D$ induce a local cycle in~$G'$ on the vertices of~$\{x,y,z\}$.
Since~$x$ is not locally universal for~$\{y,z\}$, at least one edge of this local cycle is not a local edge.
Consequently, $G'$ is not fitting for~$E^*$.
\end{proof}
\fi

Note that after the exhaustive application of~\Cref{op endpoints of local}, for each community~$C\in \mc$  with at least one local edge, the vertices of~$\locpot(C)$ induce a clique with only local edges.
\iflong
In the following, we assume that~\Cref{op endpoints of local} is applied exhaustively.
\fi

\newcommand{\FL}{\mathfrak{C}}
Next, we define a partition~$\FL$ of the communities of~$\mc$.
The idea of this partition is that in each fitting solution for~$E^*$, all communities of the same part of the partition~$\FL$ have the same unique center.
The definition of the partition~$\FL$ is based on the following lemma.
\begin{lemma}[$\star$]\label{cut3 local}
Let~$C$ and~$D$ be distinct communities of~$\mc$ with~$|C\cap D|\geq 3$ and where no vertex~$v\in C\cup D$ is locally universal for~$C\cap D$.
Let~$G'$ be a solution such that there is no vertex~$w\in C\cap D$ with~$\univ_{G'}(C) = \univ_{G'}(D) = \{w\}$.
Then, $C$ and~$D$ induce a local cycle in~$G'$ that uses at least one edge which is not a local edge. \end{lemma}
\iflong
\begin{proof}
Let~$G'$ be a solution such that there is no vertex~$w\in C\cap D$ with~$\univ_{G'}(C) = \univ_{G'}(D) = \{w\}$.
Since~$\univ_{G'}(C)$ and~$\univ_{G'}(D)$ are nonempty, there is a vertex~$x\in \univ_{G'}(C)$ and a vertex~$y\in \univ_{G'}(D)$ such that~$x \neq y$.
By the fact that~$x$ is not locally universal for~$C\cap D$, there is some vertex~$x'\in (C\cap D)\setminus \{x\}$, such that~$\{x,x'\}$ is not a local edge.
Similarly, there is some vertex~$y'\in (C\cap D)\setminus \{y\}$, such that~$\{y,y'\}$ is not a local edge.
We distinguish two cases.

\textbf{Case 1:}~$x = y'$ and~$y = x'$\textbf{.}
Since~$|C\cap D|\geq 3$, there is a vertex~$z \in C\cap D$ distinct from both~$x$ and~$y$.
Hence, $C$ and~$D$ induce a local cycle in~$G'$ on the vertices~$\{x,y,z\}$.
This local cycle contains the edge~$\{x,y\}$, which is not a local edge, since~$y = x'$.

\textbf{Case 2:}~$x \neq y'$ or~$y \neq x'$\textbf{.}
Assume without loss of generality that~$x \neq y'$.
Then, if~$x\in C\cap D$ or~$y\in C\cap D$, then~$\{x,y,y'\}$ is a local cycle (induces by~$C$ and~$D$) in~$G'$.
Otherwise, if~$C\cap D$ contains neither~$x$ nor~$y$, let~$z$ be an arbitrary vertex of~$(C\cap D) \setminus \{y'\}$.
Then, $\{x,y,y',z\}$ is a local cycle (induced by~$C$ and~$D$) in~$G'$.
Both local cycles contain the edge~$\{y,y'\}$ which is not a local edge.
\end{proof}
\fi

Consider the auxiliary graph~$G_\FL$ with vertex set~$\mc$ and where two distinct communities~$C$ and~$D$ are adjacent if and only if~$a)$~$|C\cap D|\geq 3$ and~$b)$ there is no locally universal vertex for~$C\cap D$ in~$C\cup D$.
The partition~$\FL$ consists of the connected components of~$G_\FL$ and for a community~$C\in \mc$, we denote by~$\FL(C)$ the collection of communities in the connected component of~$C$ in~$G_\FL$.
An example is shown in \Cref{figure equiv classes}.

By~\Cref{cut3 local} and due to transitivity, we obtain the following.

\begin{corollary}\label{corr equiv eq center}
For each community~$C\in \mc$ with~$|\FL(C)| \geq 2$ and each fitting solution~$G'$ for~$E^*$, there is a vertex~$v\in \bigcap_{\widetilde{C}\in \FL(C)}\widetilde{C}$ such that~$\univ_{G'}(\widetilde{C}) = \{v\}$ for each~$\widetilde{C}\in \FL(C)$.
\end{corollary}

This implies that the following operation preserves~\Cref{prop}.

\begin{operation}\label{op equivclasses}
Let~$C\in\mc$.
Remove each vertex~$v$ from~$\locpot(C)$ if~$v$ is not contained in~$\bigcap_{\widehat{C}\in \FL(C)}\locpot(\widehat{C})$.
\end{operation}

\iflong
After applying~\Cref{op equivclasses} exhaustively, for each community~$\widetilde{C}\in\FL(C)$, $\locpot(\widetilde{C}) = \locpot(C)$.
\iflonglong
\todomi{was ist bis jetzt geschafft?}
\todomi{hier über kombination der operationen reden und erklären, dass $\locpot$ von einer community leer wird durch die beiden regeln}\fi

Next, we describe an operation for communities that do not contain any local edge.
To this end, we observe the following.

\begin{lemma}\label{nonlocal cuts}
Let~$C\in\mc$ be a community that contains no local edge.
Moreover, let~$D\in \mc$ such that~$|C\cap D|\geq 2$.
Then, for each solution~$G'$ where~$\univ_{G'}(C) \not\subseteq C\cap D$, $C$ and~$D$ induce a local cycle in~$G'$ that uses at least one edge which is not a local edge.
\end{lemma}
\begin{proof}
Let~$G'$ be a solution, let~$x$ and~$y$ be distinct vertices of~$C\cap D$, and let~$z\in \univ_{G'}(C) \setminus (C\cap D)$.
Since~$z\notin C\cap D$, $G'$ contains the edges~$\{z,x\}$ and~$\{z,y\}$.
Hence, for each vertex~$w\in \univ_{G'}(D)$, $C$ and~$D$ induce a local cycle in~$G'$ on the vertices~$\{x,y,z,w\}$, since~$z\neq w$.
If~$w\in C\cap D$, then this cycle has length~$3$.
Otherwise, this cycle has length~$4
$.
In both cases, the local cycle contains the edges~$\{z,x\}$ and~$\{z,y\}$.
Since~$C$ contains no local edge, neither of these two edges is a local edge.
\end{proof}

Note that~\Cref{nonlocal cuts} implies that the following operation preserves~\Cref{prop}.
\fi
\iflong
\else
Next, we define an operation for each possibility how two communities may intersect. 
\fi

\begin{operation}\label{op nonlocal coms}
Let~$C\in\mc$ such that~$C$ contains no local edge.
Moreover, let~$D\in \mc$ such that~$|C\cap D|\geq 2$.
Remove all vertices from~$\locpot(C)$ that are not contained in~$C\cap D$.
\end{operation}

\iflong

Next, we describe two operations for communities that contain at least one local edge.
\fi

\newcommand{\myShape}[1]{
            \draw  [#1] (1,6) .. controls (1,5.5) and (2.5,5) .. (3,5) .. controls (3.5,5) and  (4.6,5)  .. (4.7,6) .. controls (4.6,7) and (3.5,7) .. (3,7) .. controls (2.5,7) and (1,6.5) .. (1,6) -- cycle ;
}
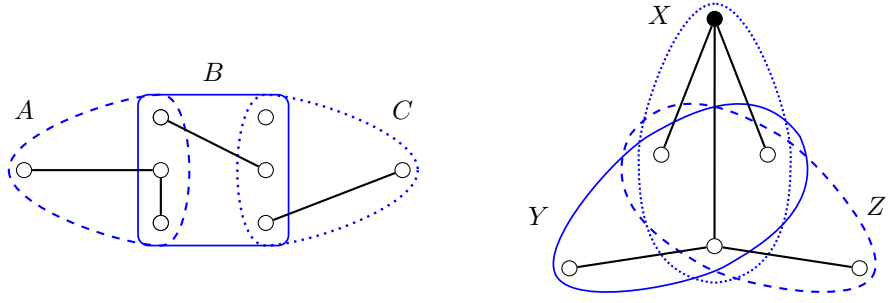
\begin{figure}[t]
\begin{center}
\begin{tikzpicture}[scale = 1]

\tikzstyle{knoten2}=[circle,fill=white,minimum size=5pt,inner sep=2pt]
\tikzstyle{knoten}=[circle,fill=white,draw=black,minimum size=5pt,inner sep=2pt]

\begin{scope}[yshift=30,xshift=240,rotate around={30:(3,3)}]

\node[knoten] (vA) at (1.2,6) {};

            \myShape{line width=.65, blue} 

%\draw (96.89,38.00) node [anchor=north west][inner sep=0.75pt]    {$u_{2}$};
\begin{scope}[xshift=91.2,yshift=118.1]
\begin{scope}[rotate around={120:(3,3)}]
\begin{scope}[xshift=0]

            \myShape{dashed, line width=.8,blue}

\node[knoten] (vC2) at (3,5.3) {};
\node[knoten] (vB1) at (1.2,6) {};
\end{scope}
\end{scope}
\end{scope}

%\draw (96.89,38.00) node [anchor=north west][inner sep=0.75pt]    {$u_{2}$};
\begin{scope}[xshift=-56.6,yshift=138]
\begin{scope}[rotate around={240:(3,3)}]
\begin{scope}[xshift=0]

            \myShape{densely dotted, line width=.8,blue}

\node[knoten, fill=black] (vB2) at (1.2,6) {};
\end{scope}
\end{scope}
\end{scope}

\node[knoten] (vC) at (3,5.3) {};
\node[knoten] (vD) at (3,6.7) {};

\draw[thick] (vA) to (vC);
\draw[thick] (vB2) to (vC);
\draw[thick] (vB2) to (vD);
\draw[thick] (vB1) to (vC);
\draw[thick] (vB2) to (vC2);

\node (label) at ($(vA) + (0,.8)$) {{$Y$}};
\node (label) at ($(vB2) + (-.6,0.4)$) {{$X$}};
\node (label) at ($(vB1) + (.6,.6)$) {{$Z$}};

\end{scope}

%%%%%%%%%%%%%%%%%%%%%%%%%%%%%%%%%%%%%%%%%%%%%%%%%%%%%%%%%%%%%

\begin{scope}[yshift=30]

\node[knoten] (vA) at (1.2,6) {};
\node[knoten] (vA2) at (3,6) {};
\node[knoten] (vC) at (3,5.3) {};
\node[knoten] (vD) at (3,6.7) {};

            \draw  [dashed, line width=.8,blue] (1,6) .. controls (1,5.5) and (2.5,5) .. (3,5) .. controls (3.5,5) and (3.5,7) .. (3,7) .. controls (2.5,7) and (1,6.5) .. (1,6) -- cycle ;

\begin{scope}[xshift=210]
            \draw  [line width=0.9,blue,dotted] (-1,6) .. controls (-1,5.5) and (-2.5,5) .. (-3,5) .. controls (-3.5,5) and (-3.5,7) .. (-3,7) .. controls (-2.5,7) and (-1,6.5) .. (-1,6) -- cycle ;

\node[knoten] (vB) at (-1.2,6) {};

\node[knoten] (vBA2) at (-3,6) {};
\node[knoten] (vBC) at (-3,5.3) {};
\node[knoten] (vBD) at (-3,6.7) {};
\end{scope}

            \draw  [line width=0.65,blue, rounded corners] ($(vD)+(-.3,.3)$)--($(vBD)+(.3,.3)$)--($(vBC)+(.3,-.3)$)--($(vC)+(-.3,-.3)$)-- cycle ;

\node[knoten] (vC) at (3,5.3) {};
\node[knoten] (vD) at (3,6.7) {};

\draw[thick] (vA) to (vA2);
\draw[thick] (vA2) to (vC);
\draw[thick] (vB) to (vBC);
\draw[thick] (vD) to (vBA2);

\node (label) at ($(vA) + (0,.8)$) {{$A$}};
\node (label) at ($(vB) + (0,.8)$) {{$C$}};
\node (label) at ($0.5*(vD) +0.5*(vBD) + (0,.6)$) {{$B$}};

\end{scope}

\end{tikzpicture}
\end{center}
\caption{Examples for parts of the partition~$\FL$. 
Only the local edges are shown.
Note that~$A$ and~$C$ are both contained in~$\FL(B)$, since~$A$ and~$C$ share at least three vertices with~$B$ and no vertex of~$A\cup B$ or~$C\cup B$ is locally universal for~$A\cap B$ or~$C\cap B$, respectively.
Hence, after exhaustive application of~\Cref{op equivclasses}, $\locpot(A) =\locpot(B) =\locpot(C) = \emptyset$, since~$A$ and~$C$ share no vertices.
Furthermore, $Y\in \FL(Z)$, since no vertex of~$Y\cup Z$ is locally universal for~$Y \cap Z$.
Note that~$X\notin\FL(Z)$, since the black vertex of~$X$ is locally universal for~$X\cap Y$ and~$X\cap Z$.
Observe that an exhaustive application of~\Cref{op equivclasses} yields~$\locpot(Z) \subseteq Y \cap Z$ and an exhaustive application of~\Cref{op local coms 3} yields~$\locpot(Z) \cap (X\cap Z) = \locpot(Z) \cap (Y\cap Z) = \emptyset$, since~$X$ contains at least one local edge and~$X\cap Z$ contains no local edge.
Hence, for both shown hypergraphs, there is no fitting solution for the given set of local edges.
}
\label{figure equiv classes}
\end{figure}

\todomi{FIGURE CAPTIONS!!!!!!!!!!!!!!!!!!!!!!!!!!!!! AND REF SOMEWHERE!!!!!!!!!!!!!!!!!!!!!!!!!}

\iflong
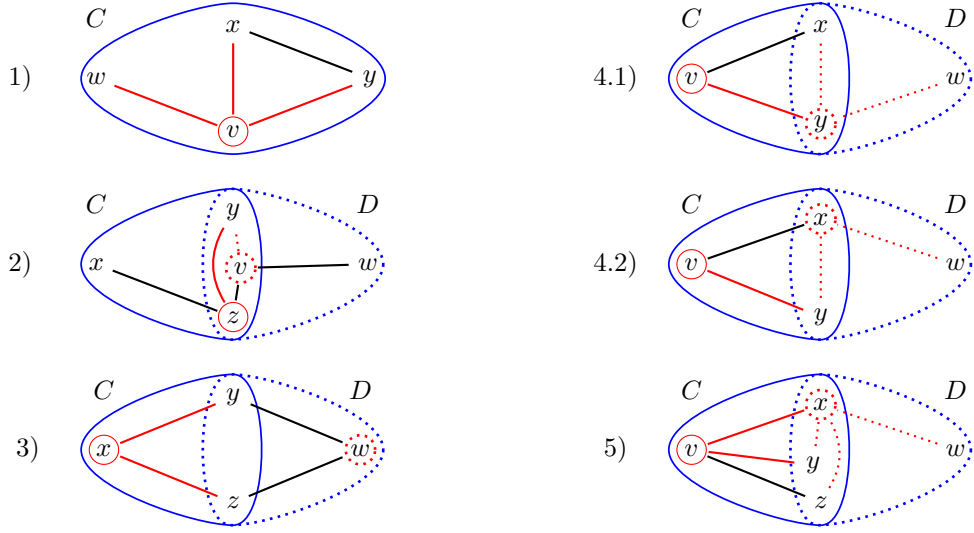
\begin{figure}[t]
\begin{center}
\begin{tikzpicture}[scale = 1]

\tikzstyle{knoten2}=[circle,fill=white,minimum size=5pt,inner sep=2pt]
\tikzstyle{knoten}=[circle,fill=white,draw=black,minimum size=5pt,inner sep=2pt]

\begin{scope}[yshift=0]

\node[knoten2] (vA) at (1.2,6) {$w$};
\node[knoten2] (vB) at (4.8,6) {$y$};
\node[knoten2] (vC) at (3,5.3) {$v$};
\node[knoten2] (vD) at (3,6.7) {$x$};

\draw[red] (vC) circle (5.5pt);
\draw  [line width=.65,blue] (1,6) .. controls (1,5.5) and (2.5,5) .. (3,5) .. controls (3.5,5) and (5,5.5) .. (5,6) .. controls (5,6.5) and (3.5,7) .. (3,7) .. controls (2.5,7) and (1,6.5) .. (1,6) -- cycle ;

%\node[knoten] (vA2) at (2.3,6) {};

%\draw[red,thick] (vA) to (vC);
%\draw[red,thick] (vA) to (vD);
\draw[black, thick] (vB) to (vD);
\draw[red, thick] (vB) to (vC);
\draw[red, thick] (vC) to (vD);
\draw[thick, red] (vC) to (vA);

\node (label) at ($(vA) + (0,.8)$) {{$C$}};
\end{scope}

\node (label) at ($(vA) - (1,0)$) {{$1)$}};

% \draw  [dashed, line width=.8,blue] (1,6) .. controls (1,5.5) and (2.5,5) .. (3,5) .. controls (3.5,5) and (3.5,7) .. (3,7) .. controls (2.5,7) and (1,6.5) .. (1,6) -- cycle ;

%%%%%%%%%%%%%%%%%%%%%%%%%%%%%%%%%%%%%%%%%%%%%%%%%%%%%%%%%%%%%
%%%%%%%%%%%%%%%%%%%%%%%%%%%%%%%%%%%%%%%%%%%%%%%%%%%%%%%%%%%%%

\begin{scope}[yshift=-70]

\node[knoten2] (vA) at (1.2,6) {$x$};

\node (label) at ($(vA) - (1,0)$) {{$2)$}};

\node[knoten2] (vD) at (3,6.7) {$y$};
\node[knoten2] (vC) at (3,5.3) {$z$};

\node[knoten2] (vA2) at (3.1,5.95) {$v$};

            \draw  [line width=0.65,blue] (1,6) .. controls (1,5.5) and (2.5,5) .. (3,5) .. controls (3.5,5) and (3.5,7) .. (3,7) .. controls (2.5,7) and (1,6.5) .. (1,6) -- cycle ;

\begin{scope}[xshift=170]
\node[knoten2] (vB) at (-1.2,6) {$w$};

            \draw  [dotted, line width=1,blue] (-1,6) .. controls (-1,5.5) and (-2.5,5) .. (-3,5) .. controls (-3.5,5) and (-3.5,7) .. (-3,7) .. controls (-2.5,7) and (-1,6.5) .. (-1,6) -- cycle ;
\draw[red] (vC) circle (5.5pt);
\draw[red,dotted,line width =1] (vA2) circle (5.5pt);
\end{scope}

\draw[thick] (vB) to (vA2);
\draw[thick] (vA2) to (vC);
\draw[thick] (vA) to (vC);

\draw[thick,red,dotted] (vA2) to (vD);
\draw[thick,red,bend left] (vC) to (vD);

\node (label) at ($(vA) + (0,.8)$) {{$C$}};
\node (label) at ($(vB) + (0,.8)$) {{$D$}};
\end{scope}

%%%%%%%%%%%%%%%%%%%%%%%%%%%%%%%%%%%%%%%%%%%%%%%%%%%%%%%%%%%%%

\begin{scope}[yshift=-140]

\node[knoten2] (vA) at (1.3,6) {$x$};

\node (label) at ($(vA) - (1,0)$) {{$3)$}};

\node[knoten2] (vD) at (3,6.7) {$y$};
\node[knoten2] (vC) at (3,5.3) {$z$};

            \draw  [line width=0.65,blue] (1,6) .. controls (1,5.5) and (2.5,5) .. (3,5) .. controls (3.5,5) and (3.5,7) .. (3,7) .. controls (2.5,7) and (1,6.5) .. (1,6) -- cycle ;

\begin{scope}[xshift=170]
\node[knoten2] (vB) at (-1.3,6) {$w$};

            \draw  [dotted, line width=1,blue] (-1,6) .. controls (-1,5.5) and (-2.5,5) .. (-3,5) .. controls (-3.5,5) and (-3.5,7) .. (-3,7) .. controls (-2.5,7) and (-1,6.5) .. (-1,6) -- cycle ;
\draw[red] (vA) circle (5.5pt);
\draw[red,dotted,line width=1] (vB) circle (5.5pt);
\end{scope}

\draw[thick] (vB) to (vC);
\draw[thick] (vB) to (vD);
\draw[thick,red] (vA) to (vC);
\draw[thick,red] (vA) to (vD);

\node (label) at ($(vA) + (0,.8)$) {{$C$}};
\node (label) at ($(vB) + (0,.8)$) {{$D$}};
\end{scope}

%%%%%%%%%%%%%%%%%%%%%%%%%%%%%%%%%%%%%%%%%%%%%%%%%%%%%%%%%%%%%

\begin{scope}[xshift=220]
\begin{scope}[yshift=0]

\node[knoten2] (vA) at (1.3,6) {$v$};

\node (label) at ($(vA) - (1,0)$) {{$4.1)$}};

\node[knoten2] (vD) at (3,6.7) {$x$};
\node[knoten2] (vC) at (3,5.4) {$y$};

            \draw  [line width=0.65,blue] (1,6) .. controls (1,5.5) and (2.5,5) .. (3,5) .. controls (3.5,5) and (3.5,7) .. (3,7) .. controls (2.5,7) and (1,6.5) .. (1,6) -- cycle ;

\begin{scope}[xshift=170]
\node[knoten2] (vB) at (-1.2,6) {$w$};

            \draw  [dotted, line width=1,blue] (-1,6) .. controls (-1,5.5) and (-2.5,5) .. (-3,5) .. controls (-3.5,5) and (-3.5,7) .. (-3,7) .. controls (-2.5,7) and (-1,6.5) .. (-1,6) -- cycle ;
\draw[red] (vA) circle (5.5pt);
\draw[red,dotted,line width=1] (vC) circle (5.5pt);
\end{scope}

\draw[thick] (vA) to (vD);
\draw[thick,red] (vA) to (vC);
\draw[thick,red,dotted] (vB) to (vC);
\draw[thick,red,dotted] (vC) to (vD);
%\draw[thick, dashed,red] (vA) to (vD);

\node (label) at ($(vA) + (0,.8)$) {{$C$}};
\node (label) at ($(vB) + (0,.8)$) {{$D$}};
\end{scope}

%%%%%%%%%%%%%%%%%%%%%%%%%%%%%%%%%%%%%%%%%%%%%%%%%%%%%%%%%%%%%

\begin{scope}[yshift=-70]

\node[knoten2] (vA) at (1.3,6) {$v$};

\node (label) at ($(vA) - (1,0)$) {{$4.2)$}};

\node[knoten2] (vD) at (3,6.6) {$x$};
\node[knoten2] (vC) at (3,5.3) {$y$};

            \draw  [line width=0.65,blue] (1,6) .. controls (1,5.5) and (2.5,5) .. (3,5) .. controls (3.5,5) and (3.5,7) .. (3,7) .. controls (2.5,7) and (1,6.5) .. (1,6) -- cycle ;

\begin{scope}[xshift=170]
\node[knoten2] (vB) at (-1.2,6) {$w$};

            \draw  [dotted, line width=1,blue] (-1,6) .. controls (-1,5.5) and (-2.5,5) .. (-3,5) .. controls (-3.5,5) and (-3.5,7) .. (-3,7) .. controls (-2.5,7) and (-1,6.5) .. (-1,6) -- cycle ;
\draw[red] (vA) circle (5.5pt);
\draw[red,dotted,line width=1] (vD) circle (5.5pt);
\end{scope}

\draw[thick] (vA) to (vD);
\draw[thick, red] (vA) to (vC);
\draw[thick,red,dotted] (vD) to (vC);
\draw[thick,red,dotted] (vB) to (vD);
%\draw[thick, dashed,red] (vA) to (vD);

\node (label) at ($(vA) + (0,.8)$) {{$C$}};
\node (label) at ($(vB) + (0,.8)$) {{$D$}};
\end{scope}

%%%%%%%%%%%%%%%%%%%%%%%%%%%%%%%%%%%%%%%%%%%%%%%%%%%%%%%%%%%%%

\begin{scope}[yshift=-140]

\node[knoten2] (vA) at (1.3,6) {$v$};

\node (label) at ($(vA) - (1,0)$) {{$5)$}};

\node[knoten2] (vD) at (3,6.6) {$x$};
\node[knoten2] (vC) at (3,5.3) {$z$};

\node[knoten2] (vA2) at (2.9,5.8) {$y$};

            \draw  [line width=0.65,blue] (1,6) .. controls (1,5.5) and (2.5,5) .. (3,5) .. controls (3.5,5) and (3.5,7) .. (3,7) .. controls (2.5,7) and (1,6.5) .. (1,6) -- cycle ;

\begin{scope}[xshift=170]
\node[knoten2] (vB) at (-1.2,6) {$w$};

            \draw  [dotted, line width=1,blue] (-1,6) .. controls (-1,5.5) and (-2.5,5) .. (-3,5) .. controls (-3.5,5) and (-3.5,7) .. (-3,7) .. controls (-2.5,7) and (-1,6.5) .. (-1,6) -- cycle ;

\draw[red] (vA) circle (5.5pt);
\draw[red,dotted,line width=1] (vD) circle (5.5pt);
\end{scope}

\draw[thick] (vA) to (vC);

\draw[thick,red,dotted] (vA2) to (vD);
\draw[thick,red,dotted] (vB) to (vD);
\draw[thick,red] (vA) to (vA2);
\draw[thick,red] (vA) to (vD);
\draw[thick,red,bend right,dotted] (vC) to (vD);

\node (label) at ($(vA) + (0,.8)$) {{$C$}};
\node (label) at ($(vB) + (0,.8)$) {{$D$}};
\end{scope}
\end{scope}

\end{tikzpicture}
\end{center}
\caption{Examples of applications of \Cref{op endpoints of local,op equivclasses,op nonlocal coms,op local coms 2,op local coms 3}.
The black edges represent the local edges, the solid (for~$C$) or dashed (for~$D$) red edges show the non-local edges resulting from choosing the respective center for community~$C$ or~$D$.
For example in~2), $z$ is the center of community~$C$ and~$v$ is the center of community~$D$, and the edges~$\{z,y\}$ and~$\{v,y\}$ are non-local edges in the solution.
For each operation, the violation of the property of being a fitting solution is shown, if a vertex~$a$ is selected as a center of a community~$A$ where the application of the corresponding operation would remove~$a$ from~$\locpot(A)$.
In~$1)$, $2)$, $3)$, and~$4.2)$, the vertex selected as center for community~$C$ is removed from~$\locpot(C)$ by the respective operation.
For example, in~$4.2)$, (assuming~$\locpot(C)\cap \{x,y\} = \{x\}$) \Cref{op local coms 2} removes~$v$ from~$\locpot(C)$, as otherwise selecting~$v$ as center of~$C$ results in the depicted non-fitting solution.
In~$4.1)$  and~$5)$, the vertex selected as center for community~$C$ is removed from~$\locpot(C)$ by the respective operation.
For example in~$4.1)$, (assuming~$\locpot(C)\cap \{x,y\} = \emptyset$) \Cref{op local coms 2} removes~$y$ from~$\locpot(D)$, as otherwise selecting~$y$ as center of~$D$ results in the depicted non-fitting solution.
}
\label{figure operations}
\end{figure}
\fi

\begin{operation}\label{op local coms 2}
Let~$C\in\mc$ such that~$C$ contains at least one local edge.
Moreover, let~$D\not\in \FL(C)$ be a community, such that~$|C\cap D|= 2$ and~$\{x,y\}\coloneqq C\cap D$ is not a local edge. 
\begin{enumerate}
\item If~$\locpot(C)\cap \{x,y\} = \emptyset$, then remove~$x$ and~$y$ from~$\locpot(D)$ or
\item if~$\locpot(C)\cap \{x,y\} = \{x\}$, then set~$\locpot(C) \coloneqq \{x\}$. 
\end{enumerate}
\end{operation}
\iflong

\begin{lemma}
If~\Cref{op endpoints of local} is exhaustively applied, then~\Cref{op local coms 2} preserves~\Cref{prop}.
\end{lemma}
\iflonglong
\todo[inline]{noten, dass man Operation 1 exhaustively applied nur bei Case 2 braucht?}
\fi
\begin{proof}
Suppose that~\Cref{prop} holds.
%Let~$C\in\mc$ such that~$C$ contains at least one local edge.
%Moreover, let~$D\not\in \FL(C)$ such that~$|C\cap D|= 2$ and let~$\{x,y\} \coloneqq C\cap D$.
%Suppose that~$\{x,y\}$ is not a local edge.

First we show that~\Cref{op local coms 2} preserves~\Cref{prop} if~$\locpot(C)\cap \{x,y\} = \emptyset$.
Suppose that~$\locpot(C) \cap \{x,y\} = \emptyset$.
Let~$G'$ be a solution for~$I$ containing all edges of~$E^*$.
We show that~$G'$ is not fitting for~$E^*$ if~$\univ_{G'}(D) \cap \{x,y\} \neq \emptyset$.
Suppose that there is some vertex~$z\in \univ_{G'}(D) \cap \{x,y\}$.
Since~\Cref{prop} holds and~$\locpot(C) \cap \{x,y\} = \emptyset$, there is some vertex~$w\in \univ_{G'}(C) \setminus \{x,y\}$.
Then, $C$ and~$D$ induce a local cycle in~$G'$ on the vertices of~$\{w,x,y\}$.
Since~$\{x,y\}$ is not a local edge, $G'$ is not fitting for~$E^*$.

Next, we show that~\Cref{op local coms 2} preserves~\Cref{prop} if~\Cref{op endpoints of local} is exhaustively applied and if~$\locpot(C)\cap \{x,y\} = \{x\}$.
Suppose that~\Cref{op endpoints of local} is exhaustively applied and that~$\locpot(C) \cap \{x,y\} = \{x\}$.
Let~$G'$ be a solution for~$I$ containing all edges of~$E^*$ and let~$z$ be an arbitrary vertex of~$\univ_{G'}(D)$.
We show that~$G'$ is not fitting for~$E^*$ if there is some vertex~$w\in \univ_{G'}(C)$ distinct from~$x$.
Since~$y\notin \locpot(C)$, $w$ is not a vertex of~$C\cap D = \{x,y\}$.

We show that~$C$ and~$D$ induce a local cycle in~$G'$ containing at least one edge which is not a local edge.
If~$z\in \{x,y\}$, then~$C$ and~$D$ induce a local cycle in~$G'$ on the vertices of~$\{w,x,y\}$.
Since~$\{x,y\}$ is not a local edge, $G'$ is not fitting for~$E^*$.
Otherwise, that is, if~$z\notin \{x,y\}$, $C$ and~$D$ induce a local cycle in~$G'$ on the vertices of~$\{w,x,y,z\}$.
This local cycles contains the edge~$\{w,y\}$.
Hence, to show that~$G'$ is not fitting for~$E^*$, it is sufficient to show that~$\{w,y\}$ is not a local edge.
Recall that~$x$ is contained in~$\locpot(C)$ and that~\Cref{op endpoints of local} is exhaustively applied.
Hence, $x$ is locally universal for each local edge~$\{u,v\}$ in~$C$.
By assumption, $C$ contains a local edge and~$\{x,y\}$ is not a local edge.
Consequently, no local edge in~$C$ is incident with~$y$ and thus~$\{w,y\}$ is not local edge.
Hence, $G'$ is not fitting for~$E^*$.
\end{proof}
\fi

\begin{operation}\label{op local coms 3}
Let~$C\in\mc$ be a community containing at least one local edge.
Moreover, let~$D\not\in \FL(C)$ such that~$|C\cap D| \geq 3$.
For each pair of distinct vertices~$x$ and~$y$ of~$C\cap D$, where~$\{x,y\}$ is not a local edge, remove~$x$ and~$y$ from~$\locpot(D)$.
\end{operation}

\iflong
\begin{lemma}
If~\Cref{op endpoints of local} is exhaustively applied, then~\Cref{op local coms 3} preserves~\Cref{prop}.
\end{lemma}
\iflonglong
\todomi{anpassen wie oben? für long?}
\fi
\begin{proof}
Suppose that~\Cref{prop} holds.
%Let~$C\in\mc$ such that~$C$ contains at least one local edge.
%Moreover, let~$D\not\in \FL(C)$ such that~$|C\cap D|\geq 3$ and let~$x$ and~$y$ be distinct vertices of~$C\cap D$.
%Suppose that~$\{x,y\}$ is not a local edge.

We show that~\Cref{op local coms 3} preserves~\Cref{prop} if~\Cref{op endpoints of local} is exhaustively applied.
Hence, assume in the following that~\Cref{op endpoints of local} is exhaustively applied.
Let~$G'$ be a solution for~$I$ containing all edges of~$E^*$.
We show that~$G'$ is not fitting for~$E^*$ if~$\univ_{G'}(D)$ contains~$x$ or~$y$.

Since~$|C\cap D|\geq 3$ and~$D\notin \FL(C)$, there is some vertex~$v\in C\cup D$ which is locally universal for~$C\cap D$.
Note that since~$\{x,y\}$ is not a local edge, $v$ is neither~$x$ nor~$y$.
We distinguish two cases.

\textbf{Case 1:}~$v\in D$\textbf{.}
Since~\Cref{op endpoints of local} is exhaustively applied and~$\{v,x\}$ and~$\{v,y\}$ are local edges in~$D$, we conclude that each vertex on~$\locpot(D)$ is locally universal for $x$~and~$y$.
Since~$\{x,y\}$ is not a local edge, $\locpot(D)$ contains neither~$x$ nor~$y$.
%Consequently, since~\Cref{prop} holds, $G'$ is not fitting for~$E^*$ if~$\univ_{G'}(D)$ contains~$x$ or~$y$.

\textbf{Case 2:}~$v\in C\setminus D$\textbf{.}
Since~\Cref{op endpoints of local} is exhaustively applied, we conclude that each vertex of~$\locpot(C)$ is locally universal for~$C\cap D$.
Let~$w$ be an arbitrary vertex of~$\univ_{G'}(C)$.
If~$w\notin\locpot(C)$, then~$G'$ is not fitting for~$E^*$ since~\Cref{prop} holds.
Otherwise, $w\in \locpot(C)$.
Consequently, $w$ is distinct from both~$x$ and~$y$, since~$\{x,y\}$ is not a local edge.
Since~$\{x,y\}\subseteq C\cap D$ and~$w\in \univ_{G'}(C)$, $G'$ contains the edges~$\{w,x\}$ and~$\{w,y\}$.
If~$\univ_{G'}(D)$ contains~$x$ or~$y$, $G'$ contains the edge~$\{x,y\}$.
Hence, $C$ and~$D$ induce a local cycle in~$G'$ on the vertices~$\{x,y,w\}$.
This local cycle contains the edge~$\{x,y\}$ which is not a local edge.
Hence, $G'$ is not fitting for~$E^*$ if~$\univ_{G'}(D)$ contains~$x$ or~$y$.
\begin{comment}
Let~$\locpot'$ denote the updated function after the application of~\Cref{op local coms 3}.
Assume towards a contradiction that~\Cref{op local coms 3} does not preserve~\Cref{prop}.
Hence, there is a fitting solution~$G'$ for~$E^*$, a community~$C\in\mc$ that contains a local edge, a community~$D\in\mc\setminus \FL(C)$ with~$|C\cap D|\geq 3$, and a vertex~$x\in \univ_{G'}(D) \setminus \locpot'(D)$.
Since~$\locpot$ fulfills~\Cref{prop}, $x$ is contained in~$\locpot(D)$.
Moreover, since~$x\notin \locpot'(D)$, $x\in C\cap D$ and there is a vertex~$y\in C\cap D$ such that~$\{x,y\}$ is not a local edge.
Since~$|C\cap D|\geq 3$ and~$D\notin \FL(C)$, there is some vertex~$v\in C\cup D$ such that~$v$ is locally universal for~$C\cap D$.
If~$v\in D$, then~$D$ contains at least one local edge incident with~$y$.
Since~$\{x,y\}$ is not a local edge, $v$ is distinct from~$x$.
Hence, since~\Cref{op endpoints of local} is exhaustively applied, $x\notin \locpot(D)$, a contradiction.
Otherwise, $v\in C\setminus D$.
Let~$z$ be an arbitrary vertex of~$\univ_{G'}(C)$.
Since~\Cref{op endpoints of local} is exhaustively applied, $z$ is locally universal for~$C\cap D$.
Since~$\{x,y\}$ is not a local edge, $z$ is neither~$x$ nor~$y$.
Hence, $C$ and~$D$ induce the local cycle on the vertices~$\{x,y,z\}$.
Since~$\{x,y\}$ is not a local edge, $G'$ is not fitting for~$E^*$, a contradiction.
\end{comment}
\end{proof}
\fi
\iflong
\else
\begin{lemma}[$\star$]
\Cref{op nonlocal coms} preserves~\Cref{prop}.
Moreover, if~\Cref{op endpoints of local} is exhaustively applied, then~\Cref{op local coms 2} and~\Cref{op local coms 3} preserve~\Cref{prop}.
\end{lemma}
\fi

Based on these operations, we are now able to present the algorithm (see \Cref{alg:algorithmFit}) behind~\Cref{thm:fitting}.
%This algorithm is shown in~\Cref{alg:algorithmFit}.
\Cref{alg:algorithmFit} works as follows:
First, we apply~\Cref{op endpoints of local,op equivclasses,op nonlocal coms,op local coms 2,op local coms 3} exhaustively.
Next, if there is a community~$C\in \mc$ with~$\locpot(C)= \emptyset$, then we return that there is no fitting solution for~$E^*$.
This is correct, since all defined operations preserve~\Cref{prop}.
Afterwards, we start with an auxiliary graph~$G_A$ with vertex set~$V$ and edge set~$E^*$ and we iterate over the partition~$\FL$.
Recall that since~\Cref{op equivclasses} is exhaustively applied, for each~$\mathcal{L}\in \FL$, $\locpot(C) = \locpot(D)$ for any two communities~$C$ and~$D$ of~$\mathcal{L}$.
For each~$\mathcal{L}\in \FL$, we find a vertex~$y\in \locpot(C)$ that minimizes the total weight of non-local edges required to make~$y$ the center of all communities of~$\mathcal{L}$, where~$C$ is an arbitrary community of~$\mathcal{L}$.
Finally, we add all edges between~$y$ and each vertex of any community of~$\mathcal{L}$ to~$G_A$.
After the iteration over the partition~$\FL$ is completed, we output~$G_A$ if it contains at most~$\ell$ edges and has total weight at most~$b$.
Otherwise, we return that there is no fitting solution for~$E^*$.
It remains to show that this greedy choice for the center vertices is correct.

\newcommand{\cent}{\mathrm{center}}
\newcommand{\agree}{\mathrm{Agree}}

        \begin{algorithm}[t]
            \SetAlgoNoEnd
            \DontPrintSemicolon
            \caption{Algorithm solving the problem described in~\Cref{thm:fitting}.}
            \label{alg:algorithmFit}
            \SetKwInOut{Input}{Input}\SetKwInOut{Output}{Output}
            \Input{$I=(G=(V,E), \mathcal{C}, \omega, \ell, b), E^*\subseteq E$}
            \Output{A solution~$G'=(V,E')$ with at most $\ell$ edges and total weight at most~$b$, or no, if there is no minimal solution which is fitting for~$E^*$}
            Compute the partition~$\FL$ of~$\mc$\;\label{line-compute parition}
            For each~$C\in\mc$, initialize~$\locpot(C)\leftarrow \univ_G(C)$ and apply~\Cref{op endpoints of local}\;\label{line-initial-fit}
                      Apply~\Cref{op endpoints of local,op equivclasses,op nonlocal coms,op local coms 2,op local coms 3} exhaustively\;\label{line-operations}
            
            \lIf{$\locpot(C)=\emptyset$ for some~$C\in\mc$}{
                \Return{no} \label{alg:first-return-no-instance}     
            }

            $G_A \leftarrow (V,E^*)$\;\label{line-local-edges}
            \ForAll{$\mathcal{L}\in \FL$}{\label{alg:outer-loop-begin}
            $C\leftarrow$ some community of~$\mathcal{L}$\;
            $V_{\mathcal{L}}\leftarrow \bigcup_{\widetilde{C}\in \mathcal{L}} \widetilde{C}$\;
            $y \leftarrow \arg\min_{u\in \locpot(C)} \omega(\{\{u,v\}  :  v\in V_{\mathcal{L}} \setminus \{u\}\}\setminus E^*)$\;\label{alg:center-selection}
                add all edges of~$\{\{y,v\}  :  v\in V_{\mathcal{L}} \setminus \{y\}\}$ to~$G_A$\;
            }\label{alg:outer-loop-end}
            \lIf{$|E'| \leq \ell$ and $\omega(E') \leq b$}{\label{alg:parameter-check-begin}
            \Return{$G_A$}\label{alg:final-return}
            }\label{alg:parameter-check-end}
            \Return{no}\label{alg:second-return-no-instance}
        \end{algorithm}
        
\begin{lemma}
\Cref{alg:algorithmFit} is correct.
\end{lemma}
\begin{proof}
If~\Cref{alg:algorithmFit} reaches Line~\ref{alg:first-return-no-instance}, then~$\locpot$ does not fulfill~\Cref{prop}, and thus \Cref{alg:algorithmFit} correctly outputs that there is no fitting solution for~$E^*$.
Otherwise, let~$G_A$ denote the graph constructed by~\Cref{alg:algorithmFit} and let for each community~$C\in \mc$, $\cent(C)$ denote the vertex~$y$ chosen to be the center of all communities of~$\FL(C)$ in Line~\ref{alg:center-selection}.
By construction, $G_A$ is a solution since for each community~$C\in \mc$, $\cent(C)$ is a vertex of~$\locpot(C)\subseteq \univ_G(C)$.
If~$G_A$ contains at most~$\ell$ edges and has total weight at most~$b$, then the algorithm correctly outputs the solution~$G_A$ which is fitting for~$E^*$.

Hence, in the following assume that~$G_A$ contains more than~$\ell$ edges or has weight more than~$b$.
% We show that there is no fitting solution for~$E^*$.
Assume towards a contradiction that there is a fitting solution for~$E^*$.
Let~$G_F$ be a fitting solution for~$E^*$ such that~$\agree(G_F) \coloneqq \{C\in\mc :  \cent(C) \in\univ_{G_F}(C)\}$ is as large as possible. 

\textbf{Case 1:}~$\agree(G_F) = \mc$\textbf{.}
By construction, $G_A$ contains all edges of~$E^*$ and only the required edges to achieve that for each community~$C\in\mc$, $\cent(C) \in \univ_{G_A}(C)$.
Consequently, $G_A$ is a subgraph of~$G_F$ and thus $G_F$ contains more than~$\ell$ edges or has weight more than~$b$, a contradiction.

\textbf{Case 2:} There is a community~$C\in \mc\setminus\agree(G_F)$\textbf{.}
In the following, we define a fitting solution~$G_F'$ for~$E^*$ with~$\agree(G_F')\supsetneq\agree(G_F)$.
By definition, $\cent(C) = \cent(\widetilde{C})$ for each community~$\widetilde{C}\in\FL(C)$.
Let~$V_C \coloneqq \bigcup_{\widetilde{C}\in \FL(C)} \widetilde{C}$ and let~$y \coloneqq \cent(C)$.
Moreover, let~$x$ be an arbitrary vertex of~$V_C$ such that~$x\in \univ_{G_F}(\widetilde{C})$ for each community~$\widetilde{C}\in \FL(C)$.
Due to~\Cref{corr equiv eq center} and since~$G_F$ is fitting for~$E^*$, this vertex exists and is unique if~$|\FL(C)|\geq 2$.
Note that~$C\in \mc\setminus\agree(G_F)$ implies that~$x\neq y$.
This also implies that~$C$ has size at least 3, and thus, each community of~$\FL(C)$ has size at least 3.
We obtain~$G_F'$ as follows: 
First, initialize~$G_F'$ as~$G_F$.
Second, for each community~$\widetilde{C}\in\FL(C)$, remove all edges that are not local edges of~$G_F[\widetilde{C}]$ from~$G_F'$.
Finally, for each community~$\widetilde{C}\in\FL(C)$, add the minimum number of edges to~$G_F'$ such that~$y\in \univ_{G_F'}(\widetilde{C})$, that is, the edges~$\{\{y, v\} :  v\in V_C \setminus \{y\}\}\setminus E^*$.
\iflonglong
\todomi{Guideline: Erst zeigen wir das, dann das, dann das, dann das....}
\fi

First, we show that~$G_F'$ contains at most as many edges as~$G_F$.
To this end, we first observe the following.
\begin{clm}[$\star$]\label{claim eq univ}
For each~$z\in V_C\setminus \{x,y\}$, the edge~$\{x,z\}$ is a local edge if and only if~$\{y,z\}$ is a local edge.
\end{clm}
\iflong
\begin{claimproof}
%Note that only edges incident with~$y$ and vertices of~$V_C$ were added to obtain~$G_F'$.
Assume towards a contradiction that there is some~$z\in V_C\setminus\{x,y\}$ such that exactly one of~$\{y,z\}$ and~$\{x,z\}$ is a local edge.
Assume that~$\{y,z\}$ is a local edge and~$\{x,z\}$ is not a local edge.
The case if~$\{x,z\}$ is a local edge and~$\{y,z\}$ is not a local edge then follows by similar arguments.
Let~$\widetilde{C}\in \FL(C)$ be a community that contains~$z$.
Note that~$\widetilde{C}$ also contains the vertices~$x$ and~$y$, since they are both contained in~$\locpot(\widetilde{C})$.
Based on the facts that~$\{y,z\}$ is a local edge and~\Cref{op endpoints of local} is exhaustively applied, $x$ is locally universal for both~$y$ and~$z$, since~$x\in \locpot(\widetilde{C})$.
Consequently, $\{x,z\}$ is a local edge, a contradiction.
\end{claimproof}\fi

Recall that each edge which is in~$G_F'$ and not in~$G_F$ is incident with~$y$ and some vertex of~$V_C\setminus\{x,y\}$. 
Hence, for each~$z\in V_C\setminus \{x,y\}$ where the edge~$\{y,z\}$ was added to obtain~$G_F'$, the edge~$\{x,z\}$ was removed to obtain~$G_F'$.
Consequently, $G_F'$ contains at most as many edges as~$G_F$.
Moreover, this implies that the difference between the total weight of~$G_F'$ and the total weight of~$G_F$ is at most~$\omega(\{\{y,z\} :  z\in V_C\setminus \{y\}\} \setminus E^*) - \omega(\{\{x,z\} :  z\in V_C\setminus \{x\}\} \setminus E^*)$.
Due to Line~\ref{alg:center-selection}, this weight difference is not positive.
Hence, since~$G_F$ has total weight at most~$b$, $G_F'$ has total weight at most~$b$.

To show that~$G_F'$ is a solution, it thus remains to show that for each community~$C\in\mc$, $C$~has at least one center in~$G_F'$.
To show this, it suffices to show that all communities outside of~$\FL(C)$ have the same centers in~$G_F$ and~$G_F'$, since~$y$ is a center of all communities of~$\FL(C)$. 
\begin{clm}\label{claim supset of univ}
For each community~$D\in \mc\setminus \FL(C)$, $\univ_{G_F}(D) = \univ_{G_F'}(D)$.
\end{clm}

\begin{claimproof}
Due to symmetry, we only show that~$\univ_{G_F}(D) \subseteq \univ_{G_F'}(D)$.
Assume towards a contradiction that there is a vertex~$z\in \univ_{G_F}(D) \setminus \univ_{G_F'}(D)$.
Since~$z\notin\univ_{G_F'}(D)$, there is an edge~$\{z,w\}$ which is contained in~$G_F$ but not in~$G_F'$.
Moreover, $\{z,w\}$ is not a local edge, since~$G_F'$ contains all local edges.
This further implies that there is a community~$\widetilde{C}\in\FL(C)$ such that~$\{z,w\}\subseteq \widetilde{C}$.
Since~$G_F$ is fitting for~$E^*$, $x$ is one endpoint of~$\{z,w\}$, as otherwise, $\widetilde{C}$ and~$D$ induce a local cycle in~$G_F$ on the vertices~$\{x,z,w\}$ and the edge~$\{z,w\}$ is not a local edge.
Next, we distinguish the cases whether~$\widetilde{C}$ contains a local edge.

\textbf{Case 1:} there is no local edge in~$\widetilde{C}$\textbf{.}
Since~\Cref{op nonlocal coms} is exhaustively applied, $\{x,y\}\subseteq \locpot(\widetilde{C}) \subseteq \widetilde{C}\cap D$.
Hence, if~$|\widetilde{C}\cap D| = 2$, then~$x=z$ and~$y=w$, or vice versa.
Consequently, the edge~$\{z,w\}$ is contained in~$G_F'$, a contradiction.
Otherwise, assume~$|\widetilde{C}\cap D|\geq 3$.
We show that in this case, there is no fitting solution for~$E^*$.
Since~$D$ is not in~$\FL(C)$, there is some vertex of~$\widetilde{C}\cup D$ which is locally universal for~$\widetilde{C}\cap D$.
Hence, $\locpot(\widetilde{C}) \subseteq \widetilde{C}\cap D$, since~$\widetilde{C}$ contains no local edge and~\Cref{op nonlocal coms} is exhaustively applied.
Moreover, since~\Cref{op local coms 3} is exhaustively applied and there is no local edge between any two vertices of~$\widetilde{C}\cap D$, $\locpot(\widetilde{C}) \cap  (\widetilde{C}\cap D) = \emptyset$.
We conclude that $\locpot(\widetilde{C}) = \emptyset$, which implies that there is no fitting solution for~$E^*$, a contradiction to the fact that~$G_F$ is a fitting solution for~$E^*$.

\textbf{Case 2:} there is some local edge in~$\widetilde{C}$\textbf{.}
Recall that \Cref{op local coms 2} and~\Cref{op local coms 3} are applied exhaustively with respect to~$\widetilde{C}$.
If~$x=z$ and~$y= w$, or vice versa, then the edge~$\{z,w\}$ is contained in~$G_F'$, a contradiction.
Otherwise, let~$w^*$ be the unique vertex of~$\{z,w\}\setminus \{x\}$. 
Since~$\{x,w^*\} = \{z,w\}$ is not a local edge, $x\in \locpot(\widetilde{C})$, and~\Cref{op endpoints of local} is exhaustively applied, no vertex of~$\locpot(\widetilde{C})$ is locally universal for~$w^*$ and~$w^* \notin \locpot(\widetilde{C})$.
Hence, if~$|\widetilde{C}\cap D| = 2$, then since~\Cref{op local coms 2} is exhaustively applied, $\locpot(\widetilde{C})$ has size at most one, a contradiction.
Otherwise, if~$|\widetilde{C}\cap D| \geq 3$, then since~\Cref{op local coms 3} is exhaustively applied~$x\notin\locpot(D)$ and~$w^*\notin\locpot(D)$.
Consequently, $z\notin\locpot(D)$, a contradiction.
\end{claimproof}

Since~$G_F$ is a solution, for each community~$D\in \mc\setminus \FL(C)$, \Cref{claim supset of univ} implies that~$\univ_{G_F'}(D)= \univ_{G_F}(D)$ is nonempty.
Hence, $G_F'$ is a solution.
Moreover, since~$G_F$ is a minimum solution, \Cref{claim eq univ} implies that~$G_F'$ is a minimum solution.

Next, we show that~$G_F'$ is fitting for~$E^*$.
To show that~$G_F'$ is a fitting solution for~$E^*$, it remains to show that each local cycle of~$G_F'$ uses only edges of~$E^*$.

\iflong
To show this, we first observe that~$y$ is the only center of each community of~$\FL(C)$ in~$G_F'$.
\begin{clm}[$\star$]\label{claim new sol less centers}
For each community~$D\in \FL(C)$, $\univ_{G_F'}(D) = \{y\}$.
\end{clm}
\iflong
\begin{claimproof}
Recall that each community in~$\FL(C)$ has size at least three.
This includes the community~$D$.
Moreover, $y$ is not a center of~$D$ in~$G_F$, since~$y$ is not a center of~$C$ in~$G_F$ and~$x$ is the unique center of~$D$ in~$G_F$, if~$\FL(C)$ has size at least two.  
Assume towards a contradiction, that there is a vertex~$z\in \univ_{G_F'}(D)$ distinct from~$y$.
By definition of~$G_F'$, each non-local edge of~$G_F'[D]$ is incident with~$y$.
This implies that for each vertex~$w\in D\setminus \{y,z\}$, the edge~$\{z,w\}$ is a local edge, since~$z$ is a center of~$D$ in~$G_F'$.
Since~$D$ has size at least three, this implies that each vertex of~$D\setminus \{y\}$ is incident with at least one local edge in~$D$.
Since~$y\in \locpot(D)$ and~\Cref{op endpoints of local} is exhaustively applied, $y$ is locally universal for~$D$ in~$G_F$.
Hence, $y$ is a center of~$D$ in~$G_F$, a contradiction.
\end{claimproof}
\fi

We are now ready to show that~$G_F'$ is a fitting solution.
\fi

\begin{clm}[$\star$]
Each local cycle of~$G_F'$ uses only edges of~$E^*$.
\end{clm}
\iflong
\begin{claimproof}
Assume towards a contradiction that there is a local cycle in~$G_F'$ with vertex set~$L$ and edge set~$E_L$ such that~$E_L\setminus E^* \neq \emptyset$.
Let~$L$ be a smallest cycle with this property.
Moreover, let~$Y$ and~$D$ be communities that induce the local cycle on the edges~$E_L$ in~$G_F'$.
First, we argue that we can assume without loss of generality that~$Y$ is contained in~$\FL(C)$ and that~$D$ is not contained in~$\FL(C)$.

On the one hand, due to~\Cref{claim supset of univ}, at least one of~$Y$ and~$D$ is contained in~$\FL(C)$, as otherwise, $Y$ and~$D$ induce the same local cycle on the edges of~$E_L$ in~$G_F$.
On the other hand, at most one of~$Y$ and~$D$ is from~$\FL(C)$, since due to~\Cref{claim new sol less centers}, $\univ_{G_F}(Y) = \univ_{G_F}(D) \subseteq \{y\}$.

Hence, assume without loss of generality that~$Y\in \FL(C)$ and that~$D\in \mc\setminus \FL(C)$.
Let~$d$ be an arbitrary vertex of~$\univ_{G_F'}(D) = \univ_{G_F}(D)$ such that the spanning star for~$D$ with center~$d$ and the spanning star for~$Y$ with center~$y$ induce the local cycle with the vertex set~$L$ and edge set~$E_L$ in~$G_F'$.

Let~$L'\coloneqq L \setminus \{y,d\}$.
Note that~$L'\subseteq Y\subseteq V_C$.
Hence, due to~\Cref{claim eq univ}, for each~$w\in L'\setminus \{x\}$, $\{x,w\}$ is a local edge if and only if~$\{y,w\}$ is a local edge.
Moreover, since~\Cref{op endpoints of local} is exhaustively applied, we observe the following.

\begin{fact}\label{myFact}
If there is at least one local edge in~$Y$, then~$\{x,y\}$ is a local edge.
\end{fact}

We distinguish three cases.

\textbf{Case 1:}~$x\in L\setminus L' = \{y,d\}$\textbf{.}
Note that this implies that~$x = d$, since~$x$ is distinct from~$y$.
Let~$e$ be a non-local edge in~$E_L$.

If~$e = \{x,y\} = \{d,y\}$, then~$|L| = 3$, since~$L$ is a smallest local cycle in~$G_F'$ containing at least one edge that is not contained in~$E^*$.
Let~$w$ be the unique vertex of~$L'$.
By construction of~$G_F'$ and since~$w\in Y$, $\{x,w\}$ is a local edge.
Hence, there is at least one local edge in~$Y$ and thus~$\{x,y\} = e$ is a local edge due to~\Cref{myFact}, a contradiction.

Otherwise, $e$ is incident with some vertex~$w\in L'$.
Note that the other endpoint of~$e$ is either~$x$ or~$y$.
Hence, due to~\Cref{claim eq univ}, both~$\{x,w\}$ and~$\{y,w\}$ are not local edges.
By construction of~$G_F'$ and since~$w\in Y$, $\{x,w\}$ is not contained in~$G_F'$, a contradiction.

\textbf{Case 2:}~$x \in L'$\textbf{.}

\textbf{Case 2.1:}~$L' = \{x\}$\textbf{.}
Recall that since~$Y$ and~$D$ induce the local cycle on the vertices~$L$, there are at least two vertices in~$Y\cap D \cap L$.
Hence~$y\in D$ or~$d\in Y$.
If~$d\notin Y$, then~$y\in D$ and thus~$G_F$ contains both edges~$\{d,x\}$ and~$\{d,y\}$. 
Consequently, since~$\{x,y\}$ is an edge of~$G_F$, $Y$ and~$D$ induce the same local cycle with the edges~$E_L$, a contradiction.
Otherwise, if~$d\in Y$, then~$\{x,d\}$ is a local edge if and only if~$\{y,d\}$ is a local edge, due to~\Cref{claim eq univ}.
Hence, $\{x,d\}$ and~$\{y,d\}$ are local edges, as otherwise, $\{x,d\}$ is not an edge of~$G_F'$.
Consequently~$\{x,y\}$ is the unique edge of~$E_L$ which is not a local edge.
Since~$Y$ contains at least one local edge, \Cref{myFact} implies that~$\{x,y\}$ is a local edge, a contradiction.

\textbf{Case 2.2:}~$|L'| = 2$\textbf{.}
Let~$w$ be the unique vertex besides~$x$ in~$L'$.
Recall that since~$w\in Y$, $\{x,w\}$ is an edge of~$G_F$ and that~$\{x,w\}$ is a local edge if and only if~$\{y,w\}$ is a local edge.
Hence, $Y$ and~$D$ induce a local cycle on the vertices~$\{w,x,d\}$ in~$G_F$.
Since~$G_F$ is fitting for~$E^*$, each edge of this local cycle is a local edge, that is, $\{x,w\}$, $\{x,d\}$, and~$\{y,w\}$ are local edges.
By the above, this implies that~$\{y,w\}$ is a local edge and since~$Y$ contains at least one local edge, $\{x,y\}$ is a local edge, due to~\Cref{myFact}.
Hence, $E_L$ contains only local edges.
A contradiction.

\textbf{Case 3:}~$x \notin L$\textbf{.}

\textbf{Case 3.1:}~$|L'| = 1$\textbf{.}
Let~$w$ be the unique vertex of~$L'$.
Recall that since~$Y$ and~$D$ induce the local cycle on the vertices~$L$, there are at least two vertices in~$Y\cap D \cap L$.
Hence~$y\in D$ or~$d\in Y$.
If~$y\in D$, then~$Y$ and~$D$ induce a local cycle in~$G_F$ with the edges~$E_L'\coloneqq \{\{d,y\},\{d,w\},\{x,y\},\{x,w\}\}$. 
Since~$G_F$ is fitting for~$E^*$, each edge of~$E_L \cap E_L' = \{\{d,y\},\{d,w\}\}$ is a local edge and~$\{y,w\}$ is the unique edge of~$E_L$ which is not a local edge.
Hence, due to~\Cref{claim eq univ}, $\{x,w\}$ is not a local edge and thus not all edges of~$E_L'$ are local edges, a contradiction.
Otherwise, if~$d\in Y$, then~$Y$ and~$D$ induce a local cycle in~$G_F$ with the edges~$E_L'\coloneqq \{\{d,w\},\{w,x\},\{x,d\}\}$.
Since~$G_F$ is fitting for~$E^*$, $\{d,w\}\in E_L \cap E_L'$ is a local edge.
Let~$e$ be an edge of~$E_L$ which is not a local edge.
By the above, $e$ is incident with~$y$.
Let~$w'$ be the other endpoint of~$e$.
Since~$w'$ is either~$d$ or~$w$, $w'$ is in~$Y$.  
Hence, due to~\Cref{claim eq univ}, $\{x,w'\}$ is not a local edge and thus not all edges of~$E_L'$ are local edges, a contradiction.

\textbf{Case 3.2:}~$|L'| = 2$\textbf{.}
Let~$w_1$ and~$w_2$ be the two vertices of~$L'$.
Recall that since~$L'\subseteq D$, for each~$i\in \{1,2\}$, $\{x,w_i\}$ is an edge of~$G_F$ and~$\{x,w_i\}$ is a local edge if and only if~$\{y,w_i\}$ is a local edge.
Hence, $Y$ and~$D$ induce the local cycle in~$G_F$ on the edges~$E_L' \coloneqq  \{\{x,w_i\},\{d,w_i\} :  i\in \{1,2\}\}$.
Since~$G_F$ is fitting for~$E^*$, each edge of~$E_L \cap E_L' = \{\{d,w_1\},\{d,w_2\}\}$ is a local edge.
Let~$e$ be an edge of~$E_L$ which is not a local edge.
By the above, $e = \{y,w_i\}$ for some~$i\in \{1,2\}$.
Hence, $\{x,w_i\}$ is not a local edge and thus not all edges of~$E_L'$ are local edges, a contradiction. 

Altogether, each local cycle in~$G_F'$ uses only edges of~$E^*$.
\end{claimproof}
\fi

Finally, we show that~$\agree(G_F')$ is a proper superset of~$\agree(G_F)$.
By construction, $\FL(C) \subseteq \agree(G_F')$, and due to~\Cref{claim supset of univ}, for each community~$D\in \mc\setminus \FL(C)$, $\univ_{G_F'}(D) = \univ_{G_F}(D)$.
Hence, $\agree(G_F)\subseteq \agree(G_F')$.
Moreover, since~$C\notin \agree(G_F')\setminus \agree(G_F)$ we obtain that~$\agree(G_F')$ is a proper superset of~$\agree(G_F)$.
Altogether, $G_F'$ is a fitting solution for~$E^*$ with~$\agree(G_F')\supsetneq \agree(G_F)$.
This contradicts our choice of~$G_F$.

Hence, if~$G_A$ contains more than~$\ell$ edges or has weight more than~$b$, then the algorithm correctly outputs that there is no fitting solution for~$E^*$.
\end{proof}

\iflong
Hence, to show~\Cref{thm:fitting}, it remains to show the running time of~\Cref{alg:algorithmFit}.
\fi
\begin{proof}[Proof of~\Cref{thm:fitting}]
Clearly, the partition~$\FL$ of~$\mathcal{C}$ and also the initialization of~$\locpot$ in Lines~\ref{line-compute parition} and~\ref{line-initial-fit} can be computed in polynomial time.
Note that~\Cref{op endpoints of local,op equivclasses,op nonlocal coms,op local coms 2,op local coms 3} can be exhaustively applied in polynomial time by iterating over all local edges and all pairs of communities, since for each community~$C\in\mc$, $\locpot(C)$ initially has size at most~$|C|<n$ and each application of any operation may only remove elements from~$\locpot(C)$.
Hence, Lines~\ref{line-operations}--\ref{line-local-edges} can be performed in polynomial time. 
Afterwards, Lines~\ref{alg:outer-loop-begin}--\ref{alg:outer-loop-end} can be performed in polynomial time since for each partite set of~$\FL$ we compute the vertex~$y$ with minimal cost such that~$y$ serves as the center of all communities in this partite set.
Finally, the check whether the solution has at most $\ell$~edges and weight at most~$b$ can be done in polynomial time.
Thus, \Cref{alg:algorithmFit} runs in polynomial time.
\end{proof}

\subparagraph*{Finding the correct edge set~$E^*$:}
To solve~\Stars, the main algorithmic difficulty now lies in finding an edge set~$E^*$ that contains all edges of local cycles of any optimal solution of~$I$. 
\iflonglong
In this subsection, we show that finding such an edge set can be done in XP-time for~$t$ and FPT-time for~$\Delta(G) + t$, where~$\Delta(G)$ denotes the maximum degree of~$G$.
Both algorithms rely on the following.
\else 
Hence, to prove~\Cref{xp t}, it remains to show that such an edge set can be found in $m^{4t}\cdot \poly(n+c)$~time, if it exists.
\fi

\begin{lemma}[$\star$]\label{at most 4t}
If~$I$ is a yes-instance of~\Stars, then for every optimal solution~$G'=(V,E')$, there is an edge set~$E^*\subseteq E'$ of size at most~$4t$ such that the edge set of each local cycle of~$G'$ is a subset of~$E^*$.
\end{lemma}
\iflong
\begin{proof}
Suppose that~$I$ is a yes-instance of~\Stars and let~$G'$ be an optimal solution for~$I$.
Consequently, the feedback edge number of~$G'$ is at most~$t$.
Let~$\mathcal{E}$ denote the collection of edge sets of all local cycles in~$G'$ and let~$E^* \coloneqq  \bigcup_{\widetilde{E}\in \mathcal{E}}$.
We show that~$|E^*| \leq 4t$.
To this end, we fix some arbitrary ordering on~$\mathcal{E}$ and let~$\mathcal{E}(i)$ denote the~$i$th element of~$\mathcal{E}$ according to this ordering.

For each~$i\in [1, |\mathcal{E}|]$, we consider the edge set~$E_i$ and the graph~$G_i = (V,E_i)$, where~$E_i$ is the union of the first~$i$ edge sets in~$\mathcal{E}$.
We show that for each~$i\in [1, |\mathcal{E}|]$, if~$E_i$ is a proper superset of~$E_{i-1}$, then the feedback edge number of~$G_i$ is larger than the feedback edge number of~$G_{i-1}$.
In other words, while iterating over the order of~$\mathcal{E}$, at most~$t$ different edge sets of~$\mathcal{E}$ can introduce new edges.

Let~$i\in [1, |\mathcal{E}|]$ such that~$E_i$ is a proper superset of~$E_{i-1}$ and let~$\bar{E} \coloneqq E_i \setminus E_{i-1}$.
Let~$X_i$ denote the endpoints of the edges of~$\mathcal{E}(i)$ and let~$\mathcal{V}$ denote the connected components of~$G_{i-1}$ containing at least one vertex of~$X_i$.
Since~$\mathcal{E}(i)$ is the edge set of a cycle, $\widetilde{V} \coloneqq  \bigcup_{\widehat{V}\in \mathcal{V}}\widehat{V}$ is a connected component in~$G_i$, and since~$\bar{E}$ is nonempty, $|\bar{E}| \geq |\mathcal{V}|$.
Hence, the feedback edge number of~$G_{i}$ is larger than the feedback edge number of~$G_{i-1}$.

Since each edge set of~$\mathcal{E}$ has size either 3 or 4, this implies, that~$E^* \coloneqq  E_{|\mathcal{E}|}$ has size at most~$4t$.
\end{proof}
\fi

\iflong
We are finally able to show~\Cref{xp t}.

\begin{proof}[Proof of~\Cref{xp t}]
Let~$I=(G=(V,E),\mathcal{C},\omega,\ell,b)$ be an instance of~\Stars.
The algorithm works as follows:
For each edge set~$E^*$ of size at most~$4t$, use \Cref{alg:algorithmFit} to find a solution for~$I$ with at most~$\ell$ edges and total weight at most~$b$ that contains~$E^*$ or correctly output that there is no fitting solution for~$E^*$.
If for some~$E^*$, a solution~$G'$ with at most~$\ell$ edges and total weight at most~$b$ is found, the algorithm outputs~$G'$.
Otherwise, the algorithms outputs ''no''.

Note that the algorithm runs in the stated running time, since there are at most~$m^{4t}$ edges sets of size at most~$4t$, all of them can be enumerated in $m^{4t}\cdot \poly(|I|)$~time, and for each such edge set~$E^*$, solving the subroutine can be done in polynomial time due to~\Cref{thm:fitting}.

Finally, we show that the algorithm is correct.
If~$I$ is a yes-instance of~\Stars, let~$G'$ be an optimal solution for~$I$.
Due to~\Cref{at most 4t}, $G'$ contains an edges set~$E^*$ of size at most~$4t$ such that each local cycle of~$G'$ uses only edges of~$E^*$.
Hence, $G'$ is fitting for~$E^*$.
Consequently for~$E^*$, a solution for~$I$ with at most~$\ell$ edges and total weight at most~$b$ is outputted by the algorithm.
Otherwise, if~$I$ is a no-instance of~\Stars, then there is no solution for~$I$ with at most~$\ell$ edges and total weight at most~$b$.
Hence, each subroutine outputs ''no'' and the whole algorithms correctly outputs ''no''.
\end{proof}
\else
\begin{proof}[Proof of~\Cref{xp t}]
The algorithm works as follows: iterate over all possible edge sets~$E^*$ of size at most~$4t$ and apply the algorithm behind~\Cref{thm:fitting}.
If~$I$ is a yes-instance, then for some edge set~$E^*$,~\Cref{thm:fitting} yields an optimal solution for~$I$ with at most~$\ell$ edges and weight at most~$b$.
Since there are~$\Oh(m^{4t})$ edges sets of size at most~$4t$, the algorithm achieves the stated running time.
\end{proof}
\fi

\subsection{\Con}

%It is known that one can check in linear time whether an instance of \Con where~$G$ is a clique with~$t=0$ is a yes-instance~\cite{duchet1978,flament1978,JohnsonP87,slater1978}.
%These proofs, however, do not provide algorithms to find such a guaranteed solution.
Korach and Stern presented an algorithm for \Con where~$G$ is a clique and~$t=0$ with running time~$\Oh(c^4n^2)$~\cite{KorachS03}.
This result was then improved by Klemz et al.~\cite{KMN14} who provided an $\Oh(m\cdot(c+\log(n)))$-time algorithm for \Con with~$t=0$.
Guttmann-Beck et al.~\cite{GuttmannBeckSS19} presented a similar algorithm for \UCon with~$t=0$.

\iflong
The algorithms of Klemz et al.~\cite{KMN14} and Guttmann-Beck et al.~\cite{GuttmannBeckSS19} for~$t=0$ first construct an auxiliary edge-weight function~$\omega$ for the underlying graph~$G$ and then use Kruskal's algorithm~\cite{K56} to find a minimum spanning tree with respect to~$\omega$ which, as they show, corresponds to a solution of the \UCon instance.
Recall that the communities~$\mathcal{C}$ define a \emph{hypergraph}~$(V(G),\mathcal{C})$. 
For the weighted case, Klemz et al.~\cite{KMN14} define an isomorphic hypergraph and then use the algorithm for the unweighted case. 
In the following, we present an, in our opinion, simpler algorithm for \Con with~$t=0$.
Our approach does not need another hypergraph, instead it directly uses Kruskal's algorithm~\cite{K56} by exploiting a different auxiliary weight function of the  underlying graph~$G$. 

%First, we show that the case~$t=0$ can be solved in polynomial-time.
%Second, we show that \UCon{} is NP-hard even for~$t=1$.
\iflonglong
\subparagraph*{Polynomial-time algorithm for~$t=0$.}\fi

%We use Kruskal's algorithm~\cite{K56} on the graph of the \Con{} instance, equipped with a specific edge-weight function based on the community structure to solve \Con in polynomial time.
%Here, $\alpha(n)$ is the inverse of the single-valued Ackermann-function.

\iflonglong
First, we show the statement for instances~$I$ for which the hypergraph is connected.
In \cref{thm:sparse-conn-with-l-equals-n-minus-1} we show a slightly stronger statement: a solution with minimal weight~$b'$ is computed, if one exists.
To verify that~$I$ is a yes-instance, if and only if~$b'\le b$.
Second, we use this algorithm to handle the case that the hypergraph is not connected.\fi

\subparagraph{A simpler algorithm.}
We first show our result for connected hypergraphs.
Let~$I=(G=(V,E), \mathcal{C},\omega, n-1,b)$ be an instance of \Con{} where the hypergraph~$\mathcal{H}=(V,\mathcal{C})$ is connected.
If~$G$ is not connected, we output that there is no solution for~$I$.
Otherwise, there is at least one spanning tree for~$G$.
First, we construct an edge-weight function~$\widetilde{\omega}$ for which we then apply Kruskal's algorithm~\cite{K56}.
The weight~$\widetilde{\omega}$ of an edge~$e$ in~$G$ is the number of communities containing edge~$e$ plus~$q(e)\coloneqq \frac{x-\omega(e)}{x}$.
Here, $x-1$ is the maximum edge weight of~$G$.
Note that~$z<\widetilde{\omega}(e)<z+1$ where~$z$ is the number of communities containing edge~$e$.
Also observe that for two edges~$e$ and~$e'$ that are contained in the same number of communities, $\widetilde{\omega}(e)\ge\widetilde{\omega}(e')$ if and only if~$\omega(e)
\le \omega(e')$.
Second, we use Kruskal's algorithm~\cite{K56} to find a maximum-weight spanning tree~$T$ on~$G$ equipped with~$\widetilde{\omega}$.
%We show that the edges of such a spanning tree in correspond to a solution~$G'$ of~$G$ of minimal weight, if one exist for~$t=0$.
%More precisely, we compute the minimal weight~$b'$ of a solution, if one exists, and then we compare~$b'$ with~$b$ to decide whether we have a yes- or a no-instance.
If this spanning tree~$T$ is a solution for~$I$, we output~$T$.
Otherwise, we output that there is no solution for~$I$.

    \begin{theorem}[$\star$]
        \label{thm:sparse-conn-with-l-equals-n-minus-1}
        Let~$I=(G=(V,E), \mathcal{C},\omega, n-1,b)$ be an instance of \Con{} where the hypergraph~$\mathcal{H}=(V,\mathcal{C})$ is connected.
        Then in time~$\Oh(m\cdot(c+\log(n)))$, the algorithm described above finds a minimum-weight solution for~$I$ or correctly outputs that there is no solution for~$I$.
    \end{theorem}
\todomi{bis hier iflong?}

       \iflong
\begin{proof}
Observe that the correctness follows, if the following statement is shown:
There exists a solution~$G'$ for~$I$ containing the edges in~$L$ of total weight~$b'$ if and only if there exists a solution~$G^*$ for~$I$ containing the edges in~$L\cup\{\{u,v\}\}$ of weight~$b'$.
Here, $b'$ is the minimal weight of any solution and~$\{u,v\}$ is the next edge chosen by Kruskal's algorithm, after exactly the edge of~$L$ have been chosen so far.

$(\Rightarrow)$ Let~$G'$ be a solution containing the edges~$L$.
We now show that there also exists a solution~$G^*$ containing the edges~$L\cup\{\{u,v\}\}$.
If~$\{u,v\}\in E(G')$, nothing is to show and thus, in the following we assume that~$\{u,v\}\notin E(G')$.

Let~$\widetilde{\omega}(\{u,v\})=y$ and let~$z\coloneqq \lfloor y\rfloor$.
By definition of~$\widetilde{\omega}$, there exist exactly~$z$ pairwise distinct communities~$C_i$ with~$u\in C_i$ and~$v\in C_i$ for each~$i\in[1,z]$.
Since~$\{u,v\}\notin E(G')$, for each~$i\in[1,z]$ there exists a path~$P_i\coloneqq (u=p^i_1, p^i_2, \ldots, p^i_z=v)$ such that each edge of~$P_i$ is contained in~$E(G')$.
Since~$G'$ is acyclic, we observe that~$P_{i_1}=P_{i_2}$ for each two communities~$C_{i_1}$ and~$C_{i_2}$ containing~$u$ and~$v$.
%Now, if~$z\ge 2$, consider the paths~$P_{i_1}$ and~$P_{i_2}$ connecting~$u$ and~$v$ in~$C_{i_1}$ and~$C_{i_2}$, respectively.
%Observe that if~$P_{i_1}\ne P_{i_2}$, then the edges of~$P_{i_1}$ and~$P_{i_2}$ together are not acyclic, a contradiction to the fact that~$G'$ is a solution, that is, $G'$ is a spanning tree for~$G$.
%Thus, for each disjoint indices~$i_1$ and~$i_2$ of~$z$, we obtain that~$P_{i_1}=P_{i_2}$.\todom{geht das nicht leichter? 'u und v sind in einer gemeinsamen community, also gibt es einen pfad zwischen ihnen in G'. G' ist kreisfrei, also ist der pfad eindeutig'}
Now, let~$P$ be the unique path in~$G'$ connecting~$u$ and~$v$.
Observe that~$V(P)\subseteq C_i$ for each~$i\in[1,z]$.
Thus, $\widetilde{\omega}(e)\ge z$ for each edge~$e$ on~$P$.

Recall that~$\{u,v\}$ is an edge chosen by Kruskal's algorithm.
Thus, $L\cup\{\{u,v\}\}$ is acyclic in~$G$.
Hence, $E(P)\not\subseteq L$ and there exists at least one edge~$\{a,p\}\in E(P)$ which is not contained in~$L$.
Observe that since~$\{u,v\}$ is chosen by Kruskal's algorithm, each other edge~$e$, such that~$L\cup\{e\}$ is acyclic, has weight at most~$y$.
Thus, $z\le\widetilde{\omega}(\{a,p\})\le y$.
Recall that~$\lfloor\widetilde{\omega}(\{a,p\})\rfloor=z=\lfloor\widetilde{\omega}(\{u,v\})\rfloor$ and that we have~$\{a,p,u,v\}\subseteq C_i$ for each~$i\in[1,z]$.
In other words, $\{u,v\}$ and~$\{a,p\}$ are contained in the same collection of communities.

Now, we can replace the edge~$\{a,p\}$ of the solution~$G'$ with the edge~$\{u,v\}$.
Since each community, which contains~$a$ and~$p$, contains~$u$ and~$v$ as well, we obtain a solution~$G^*$ containing the edges~$L\cup\{\{u,v\}\}$.
Since~$\widetilde{\omega}(\{u,v\})\ge \widetilde{\omega}(\{a,p\})$ and~$\lfloor\widetilde{\omega}(\{a,p\})\rfloor=z=\lfloor\widetilde{\omega}(\{u,v\})\rfloor$, we have~$q(\{u,v\})\ge q(\{a,p\})$ and thus~$\omega(\{u,v\})\le\omega(\{a,p\})$.
Since~$G'$ has total weight~$b'$, we thus conclude that~$G^*$ also has weight~$b'$.

$(\Leftarrow)$ If there exists a solution~$G^*$ for~$(G, \mathcal{C}, \ell)$ of weight~$b'$ containing the edges in~$L\cup\{\{u,v\}\}$, then~$G^*$ clearly also contains the edges in~$L$.

\textbf{Running Time:}
Initially, we check whether~$G$ is connected in $\Oh(n + m)$~time.
Afterwards, in $\Oh(m\cdot|\mathcal{C}|)$~time we compute the  edge-weight function~$\widetilde{\omega}$, followed by applying Kruskal's algorithm in $\Oh(m\cdot \log(n))$~time~\cite{K56}, to find a maximum weight spanning tree~$T$ for~$G$ with respect to~$\widetilde{\omega}$.
Finally, we check whether the spanning tree~$T$ is a solution in $\Oh(m\cdot|\mathcal{C}|)$~time.
Hence, we obtain the stated overall running time.
\end{proof}
\fi

\iflong
    In \Cref{thm:sparse-conn-with-l-equals-n-minus-1} the instances of \Con are restricted to connected hypergraphs.
    Next, we generalize the algorithm to hypergraphs with any number of connected components.

    \begin{corollary}[$\star$]
        Let~$I=(G=(V,E), \mathcal{C}, \omega, n-x,b)$ be an instance of \Con where~$x$ is the number of connected components of the hypergraph~$\mathcal{H}=(V,\mathcal{C})$.
        Such an instance~$I$ is solvable in $\Oh(m\cdot(c+\log(n)))$~time.
    \end{corollary}
    \iflong
    \begin{proof}    
    
        Initially, we split the hypergraph~$\mathcal{H}=(V,\mathcal{C})$ into its connected components $\mathcal{H}_1=(V_1,\mathcal{C}_1), \dots, \mathcal{H}_x=(V_x,\mathcal{C}_x)$.
        For each of these components~$\mathcal{H}_i$ we use the algorithm described in \Cref{thm:sparse-conn-with-l-equals-n-minus-1} to compute the minimal weight~$b_i$ of a solution for~$\mathcal{H}_i$ which has exactly $|V_i|-1$~edges.
        If any~$\mathcal{H}_i$ does not have a solution which is a tree, we output no. 
        Otherwise, we calculate the sum~$B$ of the minimal weights~$b_i$.
        If~$B\leq b$, we output the union of the solutions of all components~$\mathcal{H}_i$.
        Otherwise, we output no,
        
Note that this algorithm is correct.
        For the overall running time, note that the connected components~$\mathcal{H}_i$ can be determined in linear time.
		Then, according to \Cref{thm:sparse-conn-with-l-equals-n-minus-1}, the component~$\mathcal{H}_i$ can be solved in $\Oh(m_i\cdot(c_i+\log(n_i)))$~time, where~$m_i$ is the number of edges of the subgraph induced by~$V_i$, $c_i=|\mathcal{C}_i|$, and~$n_i=|V_i|$.
       Thus, the instance~$I$ of \Con is solvable in $\Oh(m\cdot(c+\log(n)))$~time.
    \end{proof}\fi
    
\fi
\iflonglong
\subparagraph*{NP-Hardness for~$t\ge 1$.}\fi
\fi

Next, we show that the positive result for~$t=0$ cannot be lifted to~$t=1$; in this case \Con{} is NP-hard.
% Note that Fluschnik and Kellerhals provided a reduction from \textsc{Hitting Set} to \UCon~ showing that \UCon{} has no polynomial kernel for~$\ell$, unless NP $\subseteq$ coNP/poly~\cite[Proposition 3]{FK21A}.
% This reduction also shows that \UCon{} is \W{2}-hard with respect to~$t$.
% Our result strengthens their result. \todok{I would remove this discussion. They do not claim W2-hardness, so ``strengthening'' something they do not care about is a bit weak sauce and needs a lengthy wind-up. Also, we need space.}
We obtain our result by reducing from the~\NP-hard~\textsc{Hamiltonian Cycle}-problem~\cite{ANS80,GJ79}, which asks for a given graph~$G=(V,E)$ if there is a~\emph{Hamiltonian cycle} in~$G$, that is, a cycle containing each vertex of~$G$ exactly once.
    \begin{theorem}
    \label{thm-con-np-h-for-t=1}
    Let~$\Pi$ be a graph class on which~\textsc{Hamiltonian Cycle} is~\NP-hard, then        \UCon is~\NP-complete on~$\Pi$ even if~$t=1$.
    \end{theorem}
    \begin{proof}
        Let $I\coloneqq (V,E)$ be an instance of \textsc{Hamiltonian Cycle} containing at least three vertices. 
        We obtain an equivalent instance~$I' \coloneqq (G=(V,E),\mc,\ell)$ of~\UCon as follows:
        We start with an empty set~$\mc$ and add for each vertex~$v\in V$ a community~$C_v \coloneqq V\setminus \{v\}$ to~$\mc$.
        Finally, we set~$\ell  \coloneqq |V|$.
\iflong

\fi           
        Note that~$t= \ell-n+x$, where~$x$ is the number of connected components of the graph.
        Thus, $t =n-n+1=1$.
        
        \iflong\else The correctness is deferred to the full version of this article.\fi
\iflong
        \textbf{Correctness:} We show that~$I$ is a yes-instance of~\textsc{Hamiltonian Cycle} if and only if~$I'$ is a yes-instance of~\UCon.

        $(\Rightarrow)$
Suppose that there is a Hamiltonian cycle in~$I$.
Let~$E_C$ be the edges of this Hamiltonian cycle and let~$G_C  \coloneqq (V, E_C)$.
We show that~$G_C$ is a solution for~$I'$.
Since~$C_v \coloneqq V \setminus \{v\}$ and~$G_C$ is a cycle the subgraph~$G_C[C_v]$ is connected.
Moreover, $G_C$ contains~$|V| =\ell$ edges.
Hence, $I'$ is a yes-instance of~\UCon.

$(\Leftarrow)$
Let~$G_C  \coloneqq (V, E_C)$ be a solution of~$I'$. 
We show that~$G_C$ is a Hamiltonian cycle.
To this end, we show that each vertex~$v\in V$ is incident with at least two edges in~$G_C$.
Assume towards a contradiction that there is a vertex~$v\in V$ which is incident with at most one edge in~$G_C$.
Since~$v$ is contained in each community of~$\mc\setminus \{C_v\}$ and each community of~$\mc$ has size~$|V|-1 \geq 2$, $v$ is incident with at least one edge of~$G_C$.
Suppose that~$w$ is the unique neighbor of~$v$ in~$G_C$.
Hence, $v$ has no neighbor in the community~$C_w = V \setminus \{w\}$ and thus~$G_C[C_w]$ is not connected, a contradiction.
Consequently, each vertex of~$V$ is incident with at least two edges in~$G_C$.
Since~$G_C$ contains exactly~$\ell = |V|$ edges, this implies that~$G_C$ is a Hamiltonian cycle.
Thus, $I$ is a yes-instance of~\textsc{Hamiltonian Cycle}.\fi
    \end{proof}
\iflong
\iflong

 Since \textsc{Hamiltonian Cycle} is NP-hard on cubic bipartite planar graphs~\cite{ANS80}, we obtain the following.\else From a hardness result for \textsc{Hamiltonian Cycle}~\cite{ANS80}, we obtain the following.\fi\fi
    
    \begin{corollary}
    \label{cor-con-nph-delta-3}
    \UCon{} is NP-complete even if~$t=1$ on subcubic bipartite planar graphs.
    \end{corollary}

    \section{\Stars Parameterized by the Number of Communities}
    \label{sec:hyperedges}
    
 \UCon is NP-hard even for~$c=7$~\cite[Proposition~4]{FK21A}.
In contrast, \Stars admits an XP-algorithm for~$c$ with running time~$n^{\Oh(c)}$: 
For each community~$C$, test each of the at most $|C|\leq n$ potential center vertices. Then, for each potential solution check whether it consists of at most $\ell$~edges of total weight at most~$b$. 
For \Stars, we show that it is unlikely that this brute-force algorithm can be improved, by showing W[1]-hardness. For \UStars, we obtain an FPT-algorithm for~$c$.
%In contrast, we show that \UStars admits an FPT-algorithm for~$c$.
%This FPT-result is quite brittle: \Stars is \W{1}-hard with respect to~$c$, as we show.

\iflong
\subsection{\W{1}-hardness for \Stars}

We first show that the simple $n^{\Oh(c)}$-time algorithm for \Stars cannot be lifted to an FPT-algorithm.
%We show that the FPT-algorithm for the number of communities~$c$ for \UStars{} (see \Cref{thm-stars-fpt-c}) cannot be extended to the weighted version, unless FPT$=$\W{1}.
\fi

\begin{theorem}[$\star$]
\label{thm-weighted-stars-w-hard-for-c}
\Stars is \W{1}-hard when parameterized by~$c$ even if~$G$ is a clique and each edge weight is~$1$ or~$2$.
\end{theorem}
\iflong
\begin{proof}
We provide a parameter-preserving reduction from the \W{1}-hard \textsc{Regular Multicolored Clique} problem~\cite{cygan2015parameterized}.
The input consists of an $r$-regular graph~$G$, an integer~$\kappa$, and a partition~$(V_1,\ldots V_\kappa)$ of~$V(G)$.
The question is whether there exists a clique of size~$\kappa$ containing exactly one vertex of each partite set~$V_i$.

We construct an equivalent instance~$(G',\mathcal{C},\omega,\ell,b)$ of \Stars as follows.
The vertex set of~$V(G')$ consists of a copy of~$V(G)$ and $\kappa$~additional vertex sets~$S_i$,~$i\in[\kappa]$, each of size~$n(G)^3$. We make~$G'$ a clique by adding all edges between vertices of~$V(G')$. To complete the construction, we specify the communities and edge weights.
First, for each color class~$i\in[1,\kappa]$, we add a community~$C_i\coloneqq V(G)\cup S_i$.
Afterwards, we define the edge weights: For each edge~$\{a,b\}\in E(G')$ such that~$\{a,b\}\in E(G)$, we set~$\omega(\{a,b\})\coloneqq 2$, for each edge~$\{a,b\}$ with~$a\in S_i$ and~$b\notin V_i$, we set~$\omega(\{a,b\})\coloneqq 2$, for each edge~$\{a,b\}$ with~$a\in S_i$ and~$b\in S_j$, we set~$\omega(\{a,b\})\coloneqq 2$, and for each remaining edge~$\{a,b\}\in E(G')$ we set~$\omega(\{a,b\})\coloneqq 1$.
Finally, we set~$\ell\coloneqq \kappa\cdot(n(G)^3+n(G)-1)-\binom{\kappa}{2}$ and~$b\coloneqq \kappa\cdot (n(G)^3+n(G)-1+r)-2\binom{\kappa}{2}$. 
Note that~$c=\kappa$, it thus remains to show the equivalence of the two instances.

$(\Rightarrow)$
Let~$K\coloneqq \{v_i: v_i\in V_i\}$ be a multicolored clique in~$G$.
We construct a solution of~$(G',\mathcal{C},\omega,\ell,b)$ as follows.
For community~$C_i$ we choose vertex~$v_i$ as its center.
Observe that each community has exactly $n(G)^3+n(G)$~vertices. 
Thus, in the solution $n(G)^3+n(G)-1$~edges are contained in each community.
Some edges are, however, counted twice:
Since the centers of all communities are disjoint and since the center of community~$C_i$ is contained in each other community we count each edge where both endpoints are centers twice.
Thus, the solution consists of exactly $\kappa\cdot(n(G)^3+n(G)-1)-\binom{\kappa}{2}$~edges.
Observe that since the graph~$G$ is $r$-regular and since~$v_i\in V_i$, the weight of all edges contained in the spanning star of community~$C_i$ is exactly~$n(G)^3+n(G)-1+r$.
The total weight of the solution is thus~$\kappa\cdot (n(G)^3+n(G)-1+r)$ minus the weight of the edges where both endpoints are centers since these edges are counted twice in the sum. Since~$K$ is a clique, all these edges have weight~$2$. Thus, the solution has weight~$\kappa\cdot (n(G)^3+n(G)-1+r)-2\binom{\kappa}{2}=b$.

$(\Leftarrow)$
Let~$G''$ be a solution of~$(G',\mathcal{C},\omega,\ell,b)$. First, observe that for each~$S_i$, the solution contains at least~$n(G)^3-1$ edges with one endpoint in~$S_i$ and the other endpoint in~$V(G)\cup S_i$. Let~$E_i$ denote this edge set for each~$S_i$. Since these edge sets~$E_i$ are pairwise disjoint, the total weight of these edges in~$E_1\cup \ldots \cup E_\kappa$ is at least~$\kappa\cdot (n(G)^3-1)$. Now, if for some~$C_i$, the center vertex is either contained in~$S_i$ or in~$V(G)\setminus V_i$, then the set~$E_i$ has weight at least~$2\cdot (n(G)^3-1)$. Thus, in that case the total edge weight of $E_1\cup \ldots \cup E_\kappa$ is at least~$(\kappa+1)\cdot (n(G)^3-1)$ which exceeds~$b$ since~$n^3(G)>\kappa\cdot (n(G)+r)$. As a consequence, the center the center~$c_i$ of each community~$C_i$ is contained in the color class~$V_i$.
We show that~$K\coloneqq \{c_i: i\in[1,\kappa]\}$ is a clique of size~$\kappa$ in~$G$. First, observe that since each center~$c_i$ has edge weight one to all vertices in~$C_i$ except to its~$r$ neighbors in~$G$, the sum of the edge weights in the spanning star of~$C_i$ is exactly $n(G)^3+n(G)-1+r$. Hence, the total edge weight of all spanning stars is~$\kappa\cdot (n(G)^3+n(G)-1+r)$ minus the weight of all edges that are contained in multiple spanning stars. Observe that these edges are exactly the~$\binom{\kappa}{2}$ edges between  the~$\kappa$ distinct center vertices. 
%The weight of all edges in the solution is~$k\cdot(n(G)\cdot k + r)$ minus the sum of all weights of edges which are counted twice.
Now since~$b=\kappa\cdot(n^3+n(G)-1+r)-2\binom{\kappa}{2}$,  the sum of all weights of edges which are counted twice is~$2\binom{\kappa}{2}$. Hence, each edge that is contained in multiple spanning stars has weight~$2$. Thus, each edge~$\{c_i,c_j\}$ between two centers~$c_i$ and~$c_j$ has weight~$2$ in~$G'$. 
Hence,~$\{c_i,c_j\}$ is an edge in~$G$, and we may conclude that~$K$ is a clique of size~$\kappa$ in~$G$.
\end{proof}
\fi

% Recall that there is a generic reduction from general instances of \Stars to instances of \Stars where the input graph~$G$ is a clique:
% Simply, replace each non-edge by an edge of weight~$b+1$. 
% %Finally, we provide a generic reduction to transfer hardness results of \Stars to instances of \Stars in which the input graph is a clique:\todok{We actually mentioned this reduction already in the intro. Thus, rather recall?}
% %Let~$I$ be an instance of \Stars with weight constraint~$b$.
% %Now, we construct an equivalent instance~$I'$ of \Stars.
% %Initially, $I'$ is a copy of~$I$.
% %Then, for each non-edge in the graph of instance~$I$, we add an edge in the graph of instance~$I'$ and assign this edge a weight of~$b+1$.
% %The equivalence of the two instances follows by the observation that no solution of~$I'$ can contain any edge of weight~$b+1$.
% Hence, we obtain the following.

% \begin{corollary}
% \label{cor-hardness-weighted-stars-cliques-c}
% \Stars is \W{1}-hard when parameterized by~$c$, even if~$G$ is a clique and there are only three different edge weights.
% \end{corollary}

  \iflong  
    \subsection{Analysis of \UStars}

Next, we show that, in contrast to the weighted case, \UStars admits an FPT-algorithm with respect to~$c$.
We also show that a polynomial kernel for~$c$ is unlikely.

\fi
    \begin{theorem}[$\star$]
    \label{thm-stars-fpt-c}
        \UStars is solvable in $\Oh(4^{c^2} \cdot (n+m) + n^2\cdot c)$~time.
    \end{theorem}
\iflong
    \begin{proof}
Our FPT-algorithm relies on branching on a specific partition of the vertices.
We say that two vertices~$u$ and~$w$ are \emph{center twins} if~$a)$~$u$ and~$w$ are contained in the same set of communities~$\mathcal{C}'\subseteq\mathcal{C}$ and if~$b)$ both vertices are \emph{potential centers} of the same set of communities~$\mathcal{D'}\subseteq\mathcal{C}'\subseteq\mathcal{C}$.
Here, a vertex~$v$ is a potential center for a community~$C$, if both~$v\in C$ is a vertex of~$C$ and each other vertex of~$C$ is adjacent to~$v$.
Based on center twins, we can partition the vertex set~$V(G)$ into sets~$\mathcal{T}\coloneqq \{T_1,\ldots,T_p\}$, that is, all vertices in each set~$T_i$ are center twins. 
%Note that~$T_i$ is a clique for each~$i\in[p]$.
        Next, we make an observation about yes-instances of \UStars regarding center twins.

        \begin{clm}
            \label{clm:twins-observation-in-sparse-stars}
            Let~$I$ be an yes-instance of \UStars with a minimal solution~$G'=(V,E')$ where~$\cen_{G'}: \mathcal{C} \to V(G)$ denotes a mapping of each community~$C\in\mathcal{C}$ to one vertex of~$C$ which is universal for~$C$ in~$G'$.
            Let~$u \in V(G)$ be a vertex, let~$T_u$ be the partite set of~$\mathcal{T}$ containing~$u$, and let~$\mathcal{C}_u \subseteq \mathcal{C}$ be the set of communities having its center, with respect to~$\cen_{G'}$, in~$T_u$.
            Then, there exists a solution~$G''=(V,E'')$ with~$|E''| \leq |E'|$ such that~$u$ is universal for each community of~$\mc_u$ in~$G''$ and for each community~$C\in \mc\setminus \mc_u$, $\cen_{G'}(C)$ is universal for~$C$ in~$G''$.
        \end{clm}
        \begin{claimproof}
            We define a new mapping~$\cen''$ of each community~$C\in\mathcal{C}$ to one vertex of~$C$ which serves as a center for~$C$ as follows: for each community~$C\in\mathcal{C}_u$, that is, each community having its center~$\cen_{G'}(C)$ in~$T_u$, we set~$\cen''(C)\coloneqq u$, and for each remaining community~$C$, that is, $C\in\mathcal{C}\setminus\mathcal{C}_u$, we set~$\cen''(C) \coloneqq \cen_{G'}(C)$.
            Next, we define based on the mapping~$\cen''$ a new solution~$G''\coloneqq (V,E'')$. 
            We set~$E'' \coloneqq \{\{\cen''(D),v\} : D \in \mathcal{C}, v\in D\setminus\{\cen''(D)\}\}$.
            Observe that~$G''[D]$ contains a spanning star for each community~$D$.

            Thus, it remains to show that~$|E''| \leq |E'|$.
To show this statement, we show that for each edge~$\{x,b\}\in E''\setminus E'$ there exists a corresponding edge in~$E'\setminus E''$ which accounts for~$\{x,b\}$.
Observe that~$G'$ and~$G''$ only differ in the centers of all communities in~$\mathcal{C}_u$.
Hence, one of the endpoints of~$\{x,b\}$, say~$x$, is~$u$.
Next, we consider all possibilities for~$b$.

\textbf{Case 1:}~$b$ is not a center of any community in~$\mathcal{C}$ with respect to~$G'$\textbf{.}
Since~$\{u,b\}\in E''$, there exists a community~$C\in \mathcal{C}$ such that~$\cen''(C)=u$, $b\in C$, and~$\cen''(C)=w\ne u$.
Note that~$C\in\mathcal{C}_u$.
Since~$b$ is no center of any community in~$G'$, we observe that~$\{w,b\}\in E'\setminus E''$ and thus~$\{w,b\}$ accounts for~$\{u,b\}$.

\textbf{Case 2:}~$b$ is the center of some community~$C_b$ with respect to~$G'$\textbf{.}
First, we consider that case that~$C_b\in \mathcal{C}_u$.
Since~$u$ and~$b$ are center twins, we have~$u\in C_b$.
Thus, $\{u,b\}$ is contained in~$E'$ and in~$E''$ as well.
Second, we consider the case that~$C_b\in\mathcal{C}\setminus\mathcal{C}_u$.
Let~$C_w$ be a community in~$\mathcal{C}_u$ whose center is moved from~$w$ in~$G'$ to~$u$ in~$G''$.
If there is no community~$C_b\in\mathcal{C}\setminus \mathcal{C}_u$ which contains~$u$ and thus also not~$w$, then, similar to Case~$1$, $\{w,b\}$ is an edge in~$E'\setminus E''$ and thus~$\{w,b\}$ accounts for~$\{u,b\}$.
This is possible since~$u$ and~$w$ are center twins.
Hence, in the following we can safely assume that there is a community~$C_b\in\mathcal{C}\setminus \mathcal{C}_u$ which contains~$u$ and thus also~$w$.
Since~$b$ is the center of~$C_b$ and since~$u,w\in C_b$, we observe that both~$\{u,b\}$ and~$\{w,b\}$ are contained in~$E'\cap E''$.

Hence, $|E''|\le |E'|$.   
        \end{claimproof}

    \begin{algorithm}[t]
        \SetAlgoNoEnd
        \DontPrintSemicolon
        \caption{Algorithm for \UStars: $\mathit{SolveSNWS}$}
        \label{alg:SolveSS}
        \SetKwInOut{Input}{Input}\SetKwInOut{Output}{Output}
        \Input{$G=(V,E), \mathcal{C}, \ell, E'$}
        \Output{A solution~$G'$ with at most $\ell$ edges which contains all edges in~$E'$, or $no$}

        \If{$\mathcal{C} = \emptyset$} {\label{alg:termination-condition-search-tree}
        \If{$\ell < |E'|$}{\label{alg:ensure-sparsified-graph-size}
        \Return{no}
        }
        \Return{$G'=(V,E')$}\label{alg:output-yes}
        }

        $C \leftarrow$ pick community of $\mathcal{C}$\;\label{alg:select-community-param-community-count}
        \ForAll{$T\in\mathcal{T}$ such that $\can(T)$ is a potential center for~$C$}{\label{alg:begin-branching-param-community-count}
            $E'' \leftarrow E' \cup \{\{\can(T), v\} : v \in C \setminus \{\can(T)\}) $\;\label{alg:ensure-community-contains-star}
            \If{$\mathit{SolveSNWS}(G, \mathcal{C} \setminus \{C\}, \ell, E'')$ returns a graph $G'$}{
                \Return{$G'$}
            }
        }\label{alg:end-branching-param-community-count}

        \Return{no}
    \end{algorithm}

        \textbf{Algorithm:}
        Let~$\mathcal{T}$ be a partitioning into center twins.
        According to \Cref{clm:twins-observation-in-sparse-stars}, we know that if there is a solution with at most $\ell$~edges, than there is also a solution with at most $\ell$~edges such that all communities having its center in some~$T_i\in\mathcal{T}$ have the same center.
        Thus, we can safely assume that all communities~$C$ having their center in~$T_i$ have the same arbitrary but fixed vertex~$\can(T_i)\in T_i$ as their center.
        
        Our depth-bounded search tree algorithm \textit{SolveSNWS} is shown in \Cref{alg:SolveSS}.
We branch for each community~$C \in \mathcal{C}$ and each partite set~$T$ of~$\mathcal{T}$, whether the center of~$C$ is in~$T$, that is, whether~$\can(T)$ is the center of~$T$.
Note that it is a necessary that~$\can(T)$ is a potential center of~$C$.
        After a center has been selected for each community, it is checked whether the resulting solution has at most $\ell$ edges.
        If some branch leads to a solution, then~$I$ is a yes-instance of \UStars, and otherwise, if no branch leads to a solution with at most $\ell$~edges, then~$I$ is a no-instance of \UStars.

        \textbf{Correctness:} We show that $I$ is a yes-instance of \UStars if and only if the algorithm returns a graph.

        $(\Rightarrow)$
        Let~$I$ be a yes-instance of \UStars and let~$G'=(V,E')$ be a solution.
        By applying \Cref{clm:twins-observation-in-sparse-stars}, we are able to obtain a graph~$G''=(V,E'')$ such that~$\cen_{G''}(C_i) = \cen_{G''}(C_j)$ for all pairs of communities~$C_i, C_j \in \mathcal{C}$ where~$c_{G'}(C_i)$ and~$c_{G'}(C_j)$ are center twins.
        Observe that~$G''$ is found by traversing the search tree built by \Cref{alg:SolveSS} and selecting center~$\can(T)$ for each center twin class.

        $(\Leftarrow)$
        Let~$G'$ be the solution returned by \Cref{alg:SolveSS} in Line~\ref{alg:output-yes}.
        The conditional statement in Line~\ref{alg:ensure-sparsified-graph-size} ensures that~$G'$ has at most $\ell$~edges.
        The termination condition in Line~\ref{alg:termination-condition-search-tree} ensures together with the statement in Line~\ref{alg:ensure-community-contains-star} that
        for each community~$C_i \in \mathcal{C}$ the induced subgraph~$G'[C_i]$ contains a spanning star.
        Hence, a solution with at most $\ell$~edges has been found by the algorithm which implies that~$I$ is a yes-instance of \UStars.

        \textbf{Running time:}
        Let~$C \in \mathcal{C}$ be the community selected for branching in Line~\ref{alg:select-community-param-community-count}.        
        The branching vector of the branching in the loop in Lines~\ref{alg:begin-branching-param-community-count}-\ref{alg:end-branching-param-community-count} has at most $|\mathcal{T}|$~entries of value~$1$, each decreasing the number of communities by~$1$.
We now bound~$|\mathcal{T}|$: the vertex set~$V(G)$ can be partitioned into at most~$2^c$ sets according to part~$a)$ of the definition of center twins, that is, vertices belonging to the same set of communities.
Then, for each such set there are up to $2^c$~possibilities on how to partition it according to part~$b)$ of the definition of center twins.
Thus, $|\mathcal{T}| \le  2^{c} \cdot 2^{c} = 4^{c}$.
        Since the maximum depth of the search tree is bounded by~$c$, the search tree has size~$\Oh(4^{c^2})$.
        Since the partition~$\mathcal{T}$ is computable in $\Oh(n^2\cdot c)$~time and
since the edge set in Line~\ref{alg:ensure-community-contains-star} is computable in $\Oh(n+m)$~time, the overall running time of~$\Oh(4^{c^2} \cdot (n+m) + n^2\cdot c)$ follows.
    \end{proof}

To complete the parameterized complexity picture, we show that a polynomial kernel for~$c$ is unlikely.
\fi
    \begin{theorem}[$\star$]
    \label{thm-stars-no-poly-kernel-c}
        \UStars parameterized by~$c$ does not admit a polynomial kernel unless NP~$\subseteq$ coNP/poly.
    \end{theorem}
    \iflong
    \begin{proof}
        We give a polynomial parameter transformation from \textsc{Hitting Set} which does not admit a polynomial kernel when parameterized by the number~$|\mathcal{F}|$ of sets, unless NP $\subseteq$ coNP/poly~\cite{cygan2015parameterized}.
        Let~$I_\textsc{HS} \coloneqq (U,\mathcal{F},k)$ be an instance of \textsc{Hitting Set}.
        We assume that~$|U| \ge 2$.
        
        Now, we construct an equivalent instance of \UStars.
        An example of the construction is shown in \Cref{fig:construction-example-no-poly-kernel-sparse-stars}.
        We start by defining the corresponding graph~$G$.
        First, we add a copy of~$U$ to~$V(G)$, and second, we add a vertex set~$Z$ of size~$|U|^3$ to~$V(G)$.
Next, we add edges, such that~$U$ is a clique and such that each vertex in~$U$ is adjacent to each vertex in~$Z$.   
Now, we define the communities: for each set~$F\in\mathcal{F}$, we add a community~$C_F\coloneqq\{F\cup Z\}$.     
        Finally, we set the parameter~$\ell \coloneqq k \cdot |U|^3 + |U|^2$.
        Let~$I_{\textsc{SNS}}$ denote the resulting instance of \textsc{Stars NWS}.
        Note that~$c = |\mathcal{F}|$, $|E| = {|U|\choose2} + |U|^4$ and~$|V| = |U| + |U|^3$.

        \begin{figure}[t]
            \centering
            \begin{tikzpicture}[x=0.75pt,y=0.75pt,yscale=-1.25,xscale=1.25]
%uncomment if require: \path (0,235); %set diagram left start at 0, and has height of 235

%Straight Lines [id:da41044386097454577]
                \draw    (209.98,59.94) -- (280.18,119.94) ;
%Straight Lines [id:da8550837315397126]
                \draw    (209.98,59.94) -- (160.08,120.04) ;
%Straight Lines [id:da8488843012190189]
                \draw    (209.98,59.94) -- (219.98,99.94) ;
%Straight Lines [id:da225410776846807]
                \draw    (280.18,119.94) -- (160.08,120.04) ;
%Straight Lines [id:da9793218479154459]
                \draw    (199.98,59.94) -- (160.08,120.04) ;
%Straight Lines [id:da09730158181888215]
                \draw    (199.98,59.94) -- (219.98,99.94) ;
%Straight Lines [id:da050662423628642594]
                \draw    (239.98,59.94) -- (280.18,119.94) ;
%Straight Lines [id:da29762425614675914]
                \draw    (219.98,99.94) -- (160.08,120.04) ;
%Straight Lines [id:da25036959945219905]
                \draw    (280.18,119.94) -- (219.98,99.94) ;
%Straight Lines [id:da35265338680017255]
                \draw    (239.98,59.94) -- (219.98,99.94) ;
%Straight Lines [id:da3577527428758345]
                \draw    (199.98,59.94) -- (280.18,119.94) ;
%Straight Lines [id:da016624084512765758]
                \draw    (240.08,59.94) -- (160.08,120.04) ;
                
%Shape: Circle [id:dp16723142827684745]
                \draw  [fill={rgb, 255:red, 255; green, 255; blue, 255 }  ,fill opacity=1 ] (282.68,120) .. controls (282.64,121.38) and (281.49,122.47) .. (280.11,122.44) .. controls (278.73,122.4) and (277.64,121.26) .. (277.68,119.88) .. controls (277.71,118.5) and (278.86,117.41) .. (280.24,117.44) .. controls (281.62,117.48) and (282.71,118.62) .. (282.68,120) -- cycle ;
                
%Shape: Circle [id:dp07484692941217275]
                \draw  [fill={rgb, 255:red, 255; green, 255; blue, 255 }  ,fill opacity=1 ] (222.48,100) .. controls (222.45,101.38) and (221.3,102.47) .. (219.92,102.44) .. controls (218.54,102.4) and (217.45,101.26) .. (217.48,99.88) .. controls (217.52,98.5) and (218.66,97.41) .. (220.04,97.44) .. controls (221.42,97.48) and (222.52,98.62) .. (222.48,100) -- cycle ;
                
%Shape: Circle [id:dp05173096426419321]
                \draw  [fill={rgb, 255:red, 255; green, 255; blue, 255 }  ,fill opacity=1 ] (162.58,120.1) .. controls (162.54,121.48) and (161.4,122.57) .. (160.02,122.54) .. controls (158.64,122.5) and (157.55,121.35) .. (157.58,119.97) .. controls (157.61,118.59) and (158.76,117.5) .. (160.14,117.54) .. controls (161.52,117.57) and (162.61,118.72) .. (162.58,120.1) -- cycle ;
                
%Shape: Circle [id:dp4041845797259914]
                \draw  [fill={rgb, 255:red, 255; green, 255; blue, 255 }  ,fill opacity=1 ] (32.57,50) .. controls (32.53,51.38) and (31.38,52.47) .. (30,52.43) .. controls (28.62,52.39) and (27.53,51.24) .. (27.57,49.86) .. controls (27.61,48.48) and (28.76,47.39) .. (30.14,47.43) .. controls (31.52,47.47) and (32.61,48.62) .. (32.57,50) -- cycle ;
                
%Shape: Circle [id:dp31148134768435964]
                \draw  [fill={rgb, 255:red, 255; green, 255; blue, 255 }  ,fill opacity=1 ] (92.57,50) .. controls (92.53,51.38) and (91.38,52.47) .. (90,52.43) .. controls (88.62,52.39) and (87.53,51.24) .. (87.57,49.86) .. controls (87.61,48.48) and (88.76,47.39) .. (90.14,47.43) .. controls (91.52,47.47) and (92.61,48.62) .. (92.57,50) -- cycle ;
                
%Shape: Circle [id:dp9183165497716363]
                \draw  [fill={rgb, 255:red, 255; green, 255; blue, 255 }  ,fill opacity=1 ] (62.57,110) .. controls (62.53,111.38) and (61.38,112.47) .. (60,112.43) .. controls (58.62,112.39) and (57.53,111.24) .. (57.57,109.86) .. controls (57.61,108.48) and (58.76,107.39) .. (60.14,107.43) .. controls (61.52,107.47) and (62.61,108.62) .. (62.57,110) -- cycle ;
                
%Shape: Polygon Curved [id:ds9757630213179496]
                \draw  [blue,line width=0.5mm,dotted] (22.97,50.39) .. controls (21.63,37.72) and (98.43,37.05) .. (96.97,51.06) .. controls (95.5,65.06) and (24.3,63.06) .. (22.97,50.39) -- cycle ;
                
%Shape: Polygon Curved [id:ds5733078489480286]
                \draw  [blue,line width=0.4mm,dashed] (57.63,115.72) .. controls (40.97,111.72) and (72.97,42.39) .. (92.97,45.06) .. controls (112.97,47.72) and (74.3,119.72) .. (57.63,115.72) -- cycle ;
                
%Shape: Circle [id:dp9913648798130666]
                \draw  [fill={rgb, 255:red, 128; green, 128; blue, 128 }  ,fill opacity=1 ] (242.58,60) .. controls (242.54,61.38) and (241.4,62.47) .. (240.02,62.44) .. controls (238.64,62.4) and (237.54,61.26) .. (237.58,59.88) .. controls (237.61,58.5) and (238.76,57.41) .. (240.14,57.44) .. controls (241.52,57.48) and (242.61,58.62) .. (242.58,60) -- cycle ;
                
%Shape: Polygon Curved [id:ds613298033377913]
                \draw  [blue,line width=0.5mm,dotted] (154.63,123.06) .. controls (148.6,113.19) and (186.67,59.42) .. (197.17,53.92) .. controls (207.67,48.42) and (241.17,47.42) .. (246.17,56.42) .. controls (251.17,65.42) and (234.17,99.42) .. (224.19,104.74) .. controls (214.21,110.07) and (160.67,132.92) .. (154.63,123.06) -- cycle ;
                
%Shape: Polygon Curved [id:ds9921603494750693]
                \draw  [blue,line width=0.4mm,dashed] (285.97,123.06) .. controls (294.17,108.92) and (252.17,60.42) .. (243.67,54.42) .. controls (235.17,48.42) and (200.67,48.42) .. (195.17,55.92) .. controls (189.67,63.42) and (208.21,101.07) .. (216.19,105.74) .. controls (224.17,110.42) and (277.77,137.19) .. (285.97,123.06) -- cycle ;
                
%Shape: Circle [id:dp6738813278008782]
                \draw  [fill={rgb, 255:red, 128; green, 128; blue, 128 }  ,fill opacity=1 ] (202.48,60) .. controls (202.45,61.38) and (201.3,62.47) .. (199.92,62.44) .. controls (198.54,62.4) and (197.45,61.26) .. (197.48,59.88) .. controls (197.52,58.5) and (198.66,57.41) .. (200.04,57.44) .. controls (201.42,57.48) and (202.52,58.62) .. (202.48,60) -- cycle ;
%Shape: Circle [id:dp9328780082574688]

                \draw  [fill={rgb, 255:red, 0; green, 0; blue, 0 }  ,fill opacity=1 ] (225.43,60) .. controls (225.42,60.28) and (225.19,60.5) .. (224.92,60.49) .. controls (224.64,60.48) and (224.42,60.25) .. (224.43,59.98) .. controls (224.44,59.7) and (224.67,59.48) .. (224.94,59.49) .. controls (225.22,59.5) and (225.44,59.73) .. (225.43,60) -- cycle ;
                
%Shape: Circle [id:dp6538414346900252]
                \draw  [fill={rgb, 255:red, 0; green, 0; blue, 0 }  ,fill opacity=1 ] (230.43,60) .. controls (230.42,60.28) and (230.19,60.5) .. (229.92,60.49) .. controls (229.64,60.48) and (229.42,60.25) .. (229.43,59.98) .. controls (229.44,59.7) and (229.67,59.48) .. (229.94,59.49) .. controls (230.22,59.5) and (230.44,59.73) .. (230.43,60) -- cycle ;
                
%Shape: Circle [id:dp11596306562972014] 
                \draw  [fill={rgb, 255:red, 0; green, 0; blue, 0 }  ,fill opacity=1 ] (220.43,60) .. controls (220.42,60.28) and (220.19,60.5) .. (219.92,60.49) .. controls (219.64,60.48) and (219.42,60.25) .. (219.43,59.98) .. controls (219.44,59.7) and (219.67,59.48) .. (219.94,59.49) .. controls (220.22,59.5) and (220.44,59.73) .. (220.43,60) -- cycle ;
                
%Shape: Circle [id:dp5937009819768183]
                \draw  [fill={rgb, 255:red, 128; green, 128; blue, 128 }  ,fill opacity=1 ] (212.48,60) .. controls (212.45,61.38) and (211.3,62.47) .. (209.92,62.44) .. controls (208.54,62.4) and (207.45,61.26) .. (207.48,59.88) .. controls (207.52,58.5) and (208.66,57.41) .. (210.04,57.44) .. controls (211.42,57.48) and (212.52,58.62) .. (212.48,60) -- cycle ;
% Text Node
                \draw (96.89,38.00) node [anchor=north west][inner sep=0.75pt]    {$u_{2}$};
% Text Node
                \draw (34.16,46.46) node [anchor=north west][inner sep=0.75pt]    {$u_{1}$};
% Text Node
                \draw (53.23,117.49) node [anchor=north west][inner sep=0.75pt]    {$u_{3}$};
% Text Node
                \draw (153.63,128.46) node [anchor=north west][inner sep=0.75pt]    {$u_{1}$};
% Text Node
                \draw (212.23,110.19) node [anchor=north west][inner sep=0.75pt]    {$u_{2}$};
% Text Node
                \draw (274.56,128.83) node [anchor=north west][inner sep=0.75pt]    {$u_{3}$};
% Text Node
                \draw (194.5,55.5) node [anchor=north west][inner sep=0.75pt]  [font=\Huge,rotate=-270]  {$\Bigr\}$};
% Text Node
                \draw (212.33,29.9) node [anchor=north west][inner sep=0.75pt]    {$Z$};
            \end{tikzpicture}
            \caption{Sketch of the construction of \Cref{thm-stars-no-poly-kernel-c}. 
            The left side shows the \textsc{Hitting Set} instance, and the right side shows the \textsc{Stars NWS} instance.
            The \textcolor{gray}{grey} vertices are the new vertices in the set~$Z$.
            The different line styles show the mapping between the sets in the \textsc{Hitting Set} instance and the communities in the \textsc{Stars NWS} instance.}\label{fig:construction-example-no-poly-kernel-sparse-stars}
        \end{figure}
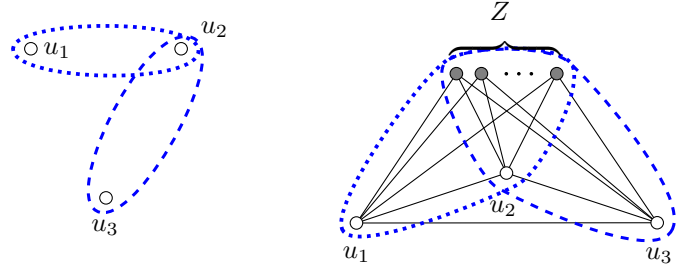

        \textbf{Correctness:}
        We show that~$I_{\textsc{HS}}$ is a yes-instance of \textsc{Hitting Set} if and only if~$I_{\textsc{SNS}}$ is a yes-instance of \UStars.

        $(\Rightarrow)$
        Let $S$ be a hitting set of size at most $k$.
        We show how to obtain a solution~$G' = (V, E')$ with~$|E'| \leq \ell$.
        We set~$E' \coloneqq \{\{s,z\} : s \in S, z \in Z\} \cup \{\{u,s\} : s \in S, u \in U \setminus \{s\}\} $.
        Observe that~$|E'| \leq k \cdot (|Z| + |U| - 1)  \leq k \cdot |U|^3 + |U|^2$.
        Recall that for each community~$C_F \in \mathcal{C}$ there exists a set~$F \in \mathcal{F}$  such that~$C_F = F \cup Z$.
        Since~$S$ is a hitting set, for each community~$C_F$, there exists at least one element~$s \in S$ with~$s \in C_F$.
        Such an element~$s$ is the center of a spanning star in~$G'[C_F]$ because~$\{\{s,z\} : z \in Z\} \subseteq E'$ and~$\{\{s,y\} : y \in C_F \setminus Z\} \subseteq \{\{s,u\} : u \in U\} \subseteq E'$.
        Thus, $G'$ is a solution which implies that $I_\textsc{SNS}$ is a yes-instance of \UStars.

        $(\Leftarrow)$
        Let~$I_{\textsc{SNS}}$ be a yes-instance of \UStars and let~$G' = (V, E')$ be a solution where~$\cen: \mathcal{C} \to V$ denotes the mapping of communities to some center vertex in~$G'$.
        We set~$S \coloneqq \{\cen(C_F) : C_F \in \mathcal{C}\}$ and show that~$S$ is a hitting set for~$I_{\textsc{HS}}$ of size at most~$k$.
        Since~$Z \subseteq C_F$ and~$Z$ is an independent set in~$G$, we obtain that~$\cen(C_F) \notin Z$ for each community~$C_F \in \mathcal{C}$.
        This implies that~$S \subseteq U$.
        Since~$|E'| \leq \ell = k \cdot |U|^3 + |U|^2$ and~$|U|\ge 2$, we conclude that $|S| \leq k$.
        Since for each~$F \in \mathcal{F}$, there exists a community~$C_F$ such that~$F = C_F \setminus Z$, we conclude that for each~$C_F\in\mathcal{F}$ there exists at least one element~$u=\cen(C)$ with~$u \in S$ and~$u \in F$.
        Thus, $X$ is a hitting set with~$|S| \leq k$ and~$I_\textsc{HS}$ is a yes-instance of \textsc{Hitting Set}.
    \end{proof}
\fi

    \section{Conclusion}
    
Presumably the most interesting open question is whether~\Stars parameterized by~$t$ admits an FPT-algorithm.
In~\Cref{thm:fitting} we showed that~\Stars can be solved in polynomial time if for some optimal solution the edge set of all local cycles is known. Hence, to obtain an FPT-algorithm, it would be sufficient to find such an edge set in FPT-time.
\iflonglong
In our research, we further discovered the following: in polynomial time we obtain an equivalent instance of~\Stars where the size of each community is upper-bounded by~$g(t)$ for some computable function~$g$\iflong\footnote{Intuitively, this can be done by three reduction rules: the first one ensures that two communities intersect only in~$\Oh(t)$ vertices.
The second rule then uses sunflower-arguments to bound for any given community~$C$ the number of vertices that are contained in intersections of size at least two of~$C$ with any other community.
The final rule then bounds the remaining vertices of~$C$.}\fi.
Since this result did not provide any kernelization or FPT-algorithm, we did not include it in \iflong{}this work\else{}the extended abstract\fi{}.
Still, this results implies the following:
To find (a superset of) the edge sets of all local cycles of any optimal solution, one only needs to determine the communities that induce these local cycles, since finding the concrete local cycles can then be done by branching on each of the at most $g(t)$~many vertices as possible centers for each of these communities in FPT-time for~$t$. 
Based on this observation, we conjecture that at least for the unweighted version, these communities can be determined in FPT-time for~$t$.\fi
%In other words, we conjecture that~\UStars admits an FPT-algorithm when parameterized by~$t$.  
Moreover, it is open whether \UCon{} can be solved in polynomial time when~$t$ is constant and the input graph is a clique. In other words, it is open whether a minimum-edge hypergraph support can be found in polynomial time when it has a constant feedback edge number. 

It would also be interesting to close the gap between the running time lower bound of~$2^{\Omega(c)}\cdot\poly(|I|)$ (see \Cref{prop-eth-bound-m-and-c}) and the upper bound of $2^{\Oh(c^2)}\cdot\poly(|I|)$ (see \Cref{thm-stars-fpt-c}) for \UStars.

    Finally, it is interesting to study $\Pi$-\textsc{NWS} for other graph properties.
    For example, Fluschnik and Kellerhals~\cite{FK21} \iflong\else{}and Wallisch~\cite{W23} \fi{}studied a variant of \textsc{Unweighted}~$\Pi$-\textsc{NWS}, denoted as \textsc{$d$-Diam NWS}, where the subgraph induced by each community is required to have diameter at most~$d$.
  \iflong{}  They showed that the problem is linear-time solvable for~$d=1$ and that it is NP-hard even if~$c=1$ and~$d=2$.
    Furthermore, Wallisch~\cite{W23} showed that \textsc{$2$-Diam NWS} is NP-hard even if~$\Delta=5$ and each community has size at most~$3$.
    
    \else
    For each~$d\geq 2$, \textsc{$d$-Diam NWS} is NP-hard~\cite{FK21}.
\fi    
     Our XP-algorithm for \Stars{} parameterized by~$t$ implies that \textsc{$d$-Diam NWS} can be solved in polynomial time when~$d=2$ and we are searching for a forest, that is, when~$t=0$ because in that case, stars are the only possibility to achieve diameter~2. 
    Can this positive result be extended to an XP-algorithm for~$t$?
    Is it also possible to achieve a polynomial-time algorithm for~$t=0$ when~$d\ge 3$? 
    
  The proofs of our ETH-bounds for \UCon and \UStars can also be used to exclude $2^{o(n^2+c)}\cdot\poly(n+c)$~time algorithms for \textsc{$d$-Diam NWS}:
  Recall that in \Cref{cor-eth-based-lowerbound-for-n-square-connected} we showed this lower bound for \UCon even if each community has size at most~$4$.
  Since for communities of size at most~$4$ the properties of being connected and having diameter at most~$d$ for~$d\ge 3$ coincide, the proof of \Cref{cor-eth-based-lowerbound-for-n-square-connected} directly implies the same running time  lower bound for~\textsc{$d$-Diam NWS} with~$d\ge 3$.
For the case~$d=2$,  the proof of \Cref{thm:eth-based-lowerbound-for-n-square} can be used to obtain the lower bound:
 For communities of size at most~$3$, the properties of having a spanning star and having diameter at most~$2$ coincide. Communities of size~$4$ with diameter 2 either have a spanning star or contain a cycle of length~$4$.
 In the instance constructed in~\Cref{thm:eth-based-lowerbound-for-n-square}, however, the solution may not contain such cycles since the budget is tight. Hence, we obtain the following.
 \begin{corollary}
\label{cor-eth-based-lowerbound-for-n-square-diameter}
If the ETH is true, then \textsc{$d$-Diam NWS}, for each~$d\ge 2$, cannot be solved in $2^{o(n^2+c)} \cdot \poly(n+c)$~time even if~$G$ is a clique and each community has size at most~$4$.
\end{corollary}
% Thus, all problems considered in this work allow trivial $2^m\cdot c\cdot \poly(n)$~time algorithm which are tight according to the ETH.
% For future work, it is interesting to search for network sparsification problems which allow $2^{\Oh(n)}\cdot\poly(n,c)$~time algorithms which are tight assuming that the ETH is true.
In light of these further hardness results we may ask the following: are there properties~$\Pi$ such that $\Pi$-NWS is NP-hard but can be solved in~$2^{\Oh(n)}\cdot \poly(n+c)$ time?

\end{document}